\documentclass[12pt]{article}
\usepackage{geometry}                % See geometry.pdf to learn the layout options. There are lots.
\usepackage{amsmath,amsthm,amssymb,url,comment}
\usepackage{ifxetex}

\DeclareMathOperator{\K}{\mathrm{C}}
\DeclareMathOperator{\KS}{\mathrm{C}}
\DeclareMathOperator{\CT}{\mathrm{CT}}
\DeclareMathOperator{\KT}{\mathrm{CT}}
\DeclareMathOperator{\KP}{\mathrm{K}}
\DeclareMathOperator{\poly}{\mathrm{poly}}
\DeclareMathOperator{\mm}{\mathbf{m}}

\usepackage{graphicx}
\newtheoremstyle{roman} % name
     {} % above skip
     {} % below skip
     {\upshape} % body font
     {} % indent
     {\bfseries} % head font
     {.} % head body punctuation
     {0.5em} % space after head
     {} % heading
\newtheorem{thm}{Theorem}
\theoremstyle{roman}
\newtheorem{example}{Example}
\newtheorem{definition}{Definition}

\newtheorem{remark}{Remark}
\newtheorem{proposition}{Proposition}
\newtheorem{lemma}{Lemma}
\newtheorem{corollary}{Corollary}

\usepackage{color}
\newcommand{\cnd}{\mskip 1.5mu | \mskip 1.5 mu}
\excludecomment{literature}
\newcommand{\bo}{\{0,1\}}

\newcommand{\calb}{\mathcal B}
\newcommand{\calc}{\mathcal C}
\newcommand{\cald}{\mathcal D}
\newcommand{\eps}{\varepsilon}
\newcommand{\ver}[1]{#1}

\ifxetex
\usepackage{fontspec,xltxtra,xunicode,unicode-math}
\defaultfontfeatures{Mapping=tex-text}
\setromanfont[Mapping=tex-text]{Linux Libertine O}
\else
\usepackage[T2A]{fontenc}
\usepackage[utf8]{inputenc}
\fi

\let\le=\leqslant
\let\ge=\geqslant

\title{Algorithmic statistics: forty years later\thanks{The work was in part funded by RFBR according to the research project grant 16-01-00362-a (N.V.) and by RaCAF ANR-15-CE40-0016-01 grant (A.S.)}}

\author{Nikolay Vereshchagin\thanks{N. Vereshchagin is with Moscow State University, National Research University Higher School of Economics and Yandex}, Alexander Shen\thanks{A. Shen is with LIRMM, CNRS \& Univ. Montpellier, 161 rue Ada, 34095, France}}

\date{}

\begin{document}

\maketitle
\begin{abstract}

Algorithmic statistics has two different (and almost orthogonal) motivations. From the philosophical point of view, it tries to formalize how the statistics works and why some statistical models are better than others. After this notion of a ``good model'' is introduced, a natural question arises: it is possible that for some piece of data there is no good model? If yes, how often these bad (\emph{non-stochastic}) data appear ``in real life''?

Another, more technical motivation comes from algorithmic information theory. In this theory a notion of complexity of a finite object (=amount of information in this object) is introduced; it assigns to every object some number, called its \emph{algorithmic complexity} (or \emph{Kolmogorov complexity}). Algorithmic statistic provides a more fine-grained classification: for each finite object some curve is defined that characterizes its behavior. It turns out that several different definitions give (approximately) the same curve.\footnote{Road-map: Section~\ref{sec:stoch} considers the notion of $(\alpha,\beta)$-stochasticity; Section~\ref{sec:two-part} considers two-part descriptions and the so-called ``minimal description length principle''; Section~\ref{sec:bcl} gives one more approach: we consider the list of objects of bounded complexity and measure how far some object is from the end of the list, getting some natural class of ``standard descriptions'' as a by-product; finally, Section~\ref{sec:depth} establishes a connection between these notions and resource-bounded complexity. The rest of the paper deals with an attempts to make theory close to practice by considering restricted classes of description (Section~\ref{sec:restricted-type}) and strong models (Section~\ref{sec:strong-models}).}

In this survey we try to provide an exposition of the main results in the field (including full proofs for the most important ones), as well as some historical comments. We assume that the reader is familiar with the main notions of algorithmic information (Kolmogorov complexity) theory. An exposition can be found in~\cite[chapters 1, 3, 4]{usv} or~\cite[chapters 2, 3]{LiVi97}, see also the survey~\cite{shen-vovk}.

A short survey of main results of algorithmic statistics was given in~\cite{vs-vovk} (without proofs); see also the last chapter of the book~\cite{usv}.

\end{abstract}

\tableofcontents

\section{Statistical models}
\label{sec:statistical-models}

Let us start with a (very rough) scheme. Imagine an experiment that produces some bit string $x$. We know nothing about the device that produced this data, and cannot repeat the experiment. Still we want to suggest some statistical model that fits the data (``explains'' $x$ in a plausible way). This model is a probability distribution on some finite set of binary strings containing $x$. What do we expect from a reasonable model?

There are, informally speaking, two main properties of a good model. First, the model should be ``simple''. If a model contains so many parameters that it is more complicated than the data itself, we would not consider it seriously. To make this requirement more formal, one can use the notion of Kolmogorov complexity.\footnote{We assume that the reader is familiar with basic notions of algorithmic information theory (complexity, a priory probability). See~\cite{shen-vovk} for a concise introduction, and \cite{LiVi97,usv} for more details.} Let us assume that measure $P$ (used as a model) has finite support and rational values. Then $P$ can be considered as a finite (constructive) object, so we can speak about Kolmogorov complexity of $P$. The requirement then says that complexity of $P$ should be much smaller than the complexity of the data string $x$ itself.

For example, if a data string $x$ contains $n$ bits, we may consider a model that corresponds to $n$ independent fair coin tosses, i.e., the uniform distribution $P$ on the set of all $n$-bit strings. Such a distribution is a constructive object that is completely determined by the value of $n$, so its complexity is $O(\log n)$, while the complexity of most $n$-bit strings is close to $n$ (and therefore is much larger than the complexity of $P$, if $n$ is large enough).

Still this simple model looks unacceptable if, for example, the sequence $x$ consists of $n$ zeros, or, more generally, if the frequency of ones in $x$ deviates significantly from $1/2$, or if zeros and ones alternate. This feeling was one of the motivations for the development of algorithmic randomness notions: why some bit sequences of length $n$ look plausible as outcomes of $n$ fair coin tosses while other do not, while all the $n$-bit sequences have the same probability $2^{-n}$ according to the model? This question does not have a clear answer in the classical probability theory, but the algorithmic approach to randomness says that plausible strings should be incompressible: the complexity of such a string (the minimal length of a program producing it) should be close to its length.

This answer works for a uniform distribution on $n$-bit strings; for arbitrary $P$ it should be modified. It turns out that for arbitrary $P$ we should compare the complexity of $x$ not with its length but with the value $(-\log P(x))$ (all logarithms are binary); if $P$ is the uniform distribution on $n$-bit strings, the value of $(-\log P(x))$ is $n$ for all $n$-bit strings $x$. Namely, we consider the difference between $(-\log P(x))$ and complexity of $x$ as \emph{randomness deficiency} of $x$ with respect to $P$. We discuss the exact definition in the next section, but let us note here that this approach looks natural: different data strings require different models.

\emph{Disclaimer}. The scheme above is oversimplified in many aspects. First, it rarely happens that we have no a priori information about the experiment that produced the data. Second, in many cases the experiment can be repeated (the same experimental device can be used again, or a similar device can be constructed). Also we often deal with a data stream: we are more interested, say, in a good prediction of oil prices for the next month than in a construction of model that fits well the prices in the past. All these aspects are ignored in our simplistic model; still it may serve as an example for more complicated cases. One should stress also that algorithmic statistics is more theoretical than practical: one of the reasons is that complexity is a non-computable function and is defined only asymptotically, up to a bounded additive term. Still the notions and results from this theory can be useful not only as philosophical foundations of statistics but as a guidelines when comparing statistical models in practice (see, for example,~\cite{deRV}).

More practical approach to the same question is provided by machine learning that deals with the same problem (finding a good model for some data set) in the ``real world''. Unfortunately, currently there is a big gap between the algorithmic statistics and machine learning: the first one provides nice results about mathematical models that are quite far from practice (see the discussion about ``standard models'' below), while machine learning is a tool that sometimes works well without any theoretical reasons. There are some attempts to close this gap (by considering models from some class or resource-bounded versions of the notions), but much more remains to be done.

\emph{A historical remark}. The principles of algorithmic statistics are often traced back to Occam's razor principle often stated as ``Don't multiply postulations beyond necessity'' or in a similar way. Poincare writes in his \emph{Science and Method'} (Chapter 1, \emph{The choice of facts}) that ``this economy of thought, this economy of effort, which is, according to Mach, the constant tendency of science, is at the same time a source of beauty and a practical advantage''. Still the mathematical analysis of these ideas became possible only after a definition of algorithmic complexity was given in 1960s (by Solomonoff, Kolmogorov and then Chaitin): after that the connection between randomness and incompressibility (high complexity) became clear. The formal definition of $(\alpha,\beta)$-stochasticity (see the next section) was given by Kolmogorov (the authors learned it from his talk given in 1981~\cite{kolmogorov81}, but most probably it was formulated earlier in 1970s; the definition appeared in print in~\cite{shen83}). For the other related approaches (the notions of logical depth and sophistication, minimal description length principle) see the discussion in the corresponding sections (see also~\cite[Chapter 5]{LiVi97}.)

\section{$(\alpha,\beta)$-stochasticity}\label{sec:stoch}

\subsection[Prefix complexity, a priori probability and randomness deficiency]{Prefix complexity, a priori probability and\\ randomness deficiency}

Preparing for the precise definition of $(\alpha,\beta)$-stochasticity, we need to fix the version of complexity used in this definition. There are several versions (plain and prefix complexities, different types of conditions), see \cite[Chapter 6]{usv}. For most of the results the choice between these versions is not important, since the difference between the different versions is small (at most $O(\log n)$ for strings of length $n$), and we usually allow errors of logarithmic size in the statements.

We will use the notion of \emph{conditional prefix complexity}, usually denoted by $\KP(x\cnd c)$. Here $x$ and $c$ are finite objects; we measure the complexity of $x$ when $c$ is given. This complexity is defined as the length of the minimal prefix-free program that, given $c$, computes $x$.\footnote{We do not go into details here, but let us mention one common misunderstanding: the set of programs should be prefix-free for each $c$, but these sets may differ for different $c$ and the union is not required to be prefix-free.} The advantage of this definition is that it has an equivalent formulation in terms of a priori probability~\cite[Chapter 4]{usv}: if $\mm(x\cnd c)$ is the conditional a priori probability, i.e., the maximal lower semicomputable function of two arguments $x$ and $c$ such that $\sum_x \mm(x\cnd c)\le 1$ for every $c$, then
$$
\KP(x\cnd c)=-\log\mm(x\cnd c)+O(1).
$$
In particular, if a probability distribution $P$ with finite support and rational values (we consider only distributions of this type) is considered as a condition, we may compare $\mm$ with function $(x,P)\mapsto P(x)$ and conclude that $\mm(x\cnd P)\ge P(x)$ up to an $O(1)$-factor, so $\KP(x\cnd P)\le -\log P(x)$. So if we define the randomness deficiency as
$$
d(x\cnd P)=-\log P(x)-\KP(x\cnd P),
$$
we get a non-negative (up to $O(1)$ additive term) function. One may also explain in a different way why $\KP(x\cnd P)\le - \log P(x)$: this inequality is a reformulation of a standard result from information theory (Shannon--Fano code, Kraft inequality).

Why do we define the deficiency in this way? The following proposition provides some additional motivation.

\begin{proposition} The function $d(x\cnd P)$ is (up to $O(1)$-additive term) the maximal lower semicomputable function of two arguments $x$ and $P$ such that
$$
 \sum_x 2^{d(x\cnd P)}\cdot P(x)\le 1 \eqno(*)
$$
for every $P$.
\end{proposition}

Here $x$ is a binary string, and $P$ is a probability distribution on binary strings with finite support and rational values. By lower semicomputable functions we mean functions that can be approximated from below by some algorithm (given $x$ and $P$, the algorithm produces an increasing sequence of rational numbers that converges to $d(x\cnd P)$; no bounds for the convergence speed are required). Then, for a given $P$, the function $x\mapsto 2^{d(x\cnd P)}$ can be considered as a random variable on the probability space with distribution $P$. The requirement $(*)$ says that its expectation is at most $1$. In this way we guarantee (by Markov inequality) that only a $P$-small fraction of strings have large deficiency: the $P$-probability of the event $d(x\cnd P)>c$ is at most $2^{-c}$. It turns out that there exists a maximal function $d$ satisfying $(*)$ up to $O(1)$ additive term, and our formula gives the expression for this function in terms of prefix complexity.

\begin{proof}
The proof uses standard arguments from Kolmogorov complexity theory. The function $\KP(x\cnd P)$ is upper semicomputable, so $d(x\cnd P)$ is lower semicomputable. We can also note that
$$
   \sum_x 2^{d(x\cnd P)}\cdot P(x) = \sum_x \frac{\mm(x\cnd P)}{P(x)}\cdot P(x)=\sum_x \mm(x\cnd P)\le1,
$$
so the deficiency function satisfies~$(*)$.

To prove the maximality, consider an arbitrary function $d'(x\cnd P)$ that is lower semicomputable and satisfies $(*)$. Then consider a function $m(x\cnd P)=2^{d'(x\cnd P)}\cdot P(x)$ (the function equals $0$ if $x$ is not in the support of $P$). Then $m$ is lower semicomputable, $\sum_x m(x\cnd P)\le 1$ for every $P$, so $m(x\cnd P)\le \mm(x\cnd P)$ up to $O(1)$-factor; this implies that $d'(x\cnd P)\le d(x\cnd P)+O(1)$.
\end{proof}

For the case where $P$ is the uniform distribution on $n$-bit strings, using $P$ as a condition is equivalent to using $n$ as the condition, so
$$
d(x\cnd P) = n -\KP(x\cnd n)
$$
in this case, and small deficiency means that complexity $\KP(x\cnd n)$ is close to the length $n$, so $x$ is incompressible.\footnote{Initially Kolmogorov suggested to consider $n -\KS(x)$ as ``randomness deficiency'' in this case, where $\KS$ stands for the plain (not prefix) complexity. One may also consider $n-\KS(x\cnd n)$. But all three deficiency functions mentioned are close to each other for strings $x$ of length $n$; one can show that the difference between them is bounded by $O(\log d)$ where $d$ is any of these three functions. The proof works by comparing the expectation and probability-bounded characterizations as explained in~\cite{bienvenu-gacs-et-al}.}

\subsection{Definition of stochasticity}

\begin{definition}
A string $x$ is called $(\alpha,\beta)$-stochastic if there exists some probability distribution $P$ (with rational values and finite support) such that $\KP(P)\le \alpha$ and $d(x\cnd P)\le\beta$.
\end{definition}

By definition every $(\alpha,\beta)$-stochastic string is $(\alpha',\beta')$-stochastic for $\alpha'\ge\alpha$, $\beta'\ge\beta$. Sometimes we say informally that a string is ``stochastic'' meaning that it is $(\alpha,\beta)$-stochastic for some reasonably small thresholds $\alpha$ and $\beta$ (for example, one can consider $\alpha,\beta=O(\log n)$ for $n$-bit strings).
\smallskip

Let us start with some simple remarks.

\begin{itemize}

\item Every simple string is stochastic. Indeed,  if $P$ is concentrated on $x$ (singleton support), then $\KP(P)\le \KP(x)$ and $d(x\cnd P)=0$ (in both cases with $O(1)$-precision), so $x$ is always $(\KP(x)+O(1),O(1))$-stochastic.

\item On the other end of the spectrum: if $P$ is a uniform distribution on $n$-bit strings, then $\KP(P)=O(\log n)$, and most strings of length $n$ have $d(x\cnd P)=O(1)$, so most strings of length $n$ are $(O(\log n), O(1))$-stochastic. The same distribution also witnesses that every $n$-bit string is $(O(\log n), n+O(1))$-stochastic.

\item It is easy to construct stochastic strings that are between these two extreme cases. Let $x$ be an incompressible string of length $n$. Consider the string $x0^n$ (the first half is $x$, the second half is zero string). It is $(O(\log n), O(1))$-stochastic: let $P$ be the uniform distribution on all the strings of length $2n$ whose second half contains only zeros.

\item For every distribution $P$ (with finite support and rational values, as usual) a random sampling according to $P$ gives us a $(\KP(P), c)$-stochastic string with probability at least $1-2^{-c}$. Indeed, the probability to get a string with deficiency greater than $c$ is at most $2^{-c}$ (Markov inequality, see above).

\end{itemize}

After these observations one may ask whether non-stochastic strings exists at all --- and how they can be constructed? A non-stochastic string should have non-negligible complexity (our first observation), but a standard way to get strings of high complexity, by coin tossing or other random experiment, can give only stochastic strings (our last observation).

We will see that non-stochastic strings do exist in the mathematical sense; however, the question whether they appear in the ``real world'', is philosophical. We will discuss both questions soon, but let us start with some mathematical results.

First of all let us note that with logarithmic precision we may restrict ourselves to uniform distributions on finite sets.

\begin{proposition}\label{prop:models-to-sets-1}
Let $x$ be an $(\alpha,\beta)$-stochastic string of length $n$. Then there exist a finite set $A$ containing $x$ such that $\KP(A)\le\alpha+O(\log n)$ and $d(x\cnd U_A)\le \beta+O(\log n)$, where $U_A$ is the uniform distribution on $A$.
\end{proposition}

Since $\KP(A)=\KP(U_A)$ (with $O(1)$-precision, as usual), this proposition means that we may consider only uniform distributions in the definition of stochasticity, and get an equivalent (up to logarithmic change in the parameters) definition.  According to this modified definition, a string $x$ in $(\alpha,\beta)$-stochastic if there exists a finite set $A$ such that $\KP(A)\le \alpha$ and $d(x\cnd A)\le \beta$, where $d(x\cnd A)$ is now defined as $\log \#A - \KP(x\cnd A)$. Kolmogorov originally proposed the definition in this form (but used plain complexity).

\begin{proof}
Let $P$ be the (finite) distribution that exists due to the definition of $(\alpha,\beta)$-stochas\-tic\-i\-ty of $x$. We may assume without loss of generality that $\beta\le n$ (as we have seen, all strings of length $n$ are $(O(\log n), n+O(1))$-stochastic, so for $\beta>n$ the statement is trivial). Consider the set $A$ formed by all strings that have sufficiently large $P$-probability. Namely, let us choose minimal $k$ such that $2^{-k}\le P(x)$ and consider the set $A$ of all strings such that $P(x)\ge 2^{-k}$. By construction $A$ contains $x$. The size of $A$ is at most $2^k$, and $-\log P(x)=k$ with $O(1)$-precision. According to our assumption, $d(x\cnd P)=k-\KP(x\cnd P)\le n$, so $k=d(x\cnd P)+\KP(x\cnd P)\le O(n)$. Then $$\KP(x\cnd A)\ge \KP(x\cnd P,k)\ge \KP(x\cnd P)-O(\log n),$$ since $A$ is determined by $P,k$, and the additional information in $k$ is $O(\log k)=O(\log n)$ since $k=O(n)$ by our assumption. So the deficiency may increase only by $O(\log n)$ when we replace $P$ by $U_A$, and $$\KP(A)\le\KP(P,k)\le\KP(P)+O(\log n)$$ for the same reasons.
\end{proof}

\begin{remark}
Similar argument can be applied if $P$ is a computable distribution (may be, with infinite support) computed by some program $p$, and we require $\KP(p)\le \alpha$ and $-\log P(x)-\KP(x\cnd p)\le \beta$. So in this way we also get the same notion (with logarithmic precision). It is important, however, that program $p$ \emph{computes} the distribution $P$ (given some point $x$ and some precision $\varepsilon>0$, it computes the probability of $x$ with error at most $\varepsilon$). It is \emph{not} enough for $P$ to be an output distribution for a randomized algorithm $p$ (in this case $P$ is called the semimeasure lower \emph{semicomputed} by $p$; note that the sum of probabilities may be strictly less than $1$ since the computation may diverge with positive probability). Similarly, it is very important in the version with finite sets $A$ (and uniform distributions on them) that the set $A$ is considered as a finite object: $A$ is simple if there is a short program that prints the list of all elements of $A$. If we allowed the set $A$ to be presented by an algorithm that enumerates $A$ (but never says explicitly that no more elements will appear), then situation would change drastically: for every string of complexity $k$ the finite set $S_k$ of strings that have complexity at most $k$, would be a good explanation for $x$, so all objects would become stochastic.
\end{remark}

\subsection{Stochasticity conservation}\label{subsec:stoch-cons}

We have defined stochasticity for binary strings. However, the same definition can be used for arbitrary finite (constructive) objects: pairs of strings, tuples of strings, finite sets of strings, graphs, etc. Indeed, complexity can be defined for all these objects as the complexity of their encodings; note that the difference in complexities for different encodings is at most $O(1)$. The same can be done for finite sets of these objects (or probability distributions), so the definition of $(\alpha,\beta)$-stochasticity makes sense.

One can also note that computable bijection preserves stochasticity (up to a constant that depends on the bijection, but not on the object). In fact, a stronger statement is true: every total computable mapping preserves stochasticity. For example, consider a stochastic pair of strings $(x,y)$. Does it imply that $x$ (or $y$) is stochastic? It is indeed the case: if $P$ is a distribution on pairs that is a reasonable model for $(x,y)$, then its projection (marginal distribution on the first components) should be a reasonable model for $x$. In fact, projection can be replaced by any \emph{total} computable mapping.

\begin{proposition}\label{prop:stoch-cons}
Let $F$ be a total computable mapping whose arguments and values are strings. If $x$ is $(\alpha,\beta)$-stochastic, then $F(x)$ is $(\alpha+O(1),\beta+O(1))$-stochastic. Here the constant in $O(1)$ depends on $F$ but not on $x,\alpha,\beta$.
\end{proposition}

\begin{proof}
Let $P$ be the distribution such that $\KP(P)\le\alpha$ and $d(x\cnd P)\le \beta$; it exists according to the definition of stochasticity. Let $Q=F(P)$ be the image distribution. In other words, if $\xi$ is a random variable with distribution $P$, then $F(\xi)$ has distribution $Q$. It is easy to see that $\KP(Q)\le\KP(P)+O(1)$, where the constant depends only on $F$. Indeed, $Q$ is determined by $P$ and $F$ in a computable way. It remains to show that $d(F(x)\cnd Q)\le d(x\cnd P)+O(1)$.

The easiest way to show this is to recall the characterization of deficiency as the maximal lower semicomputable function such that
$$
\sum_u 2^{d(u\cnd S)}S(u)\le 1
$$
for every distribution $S$. We may consider another function $d'$ defined as
$$
  d'(u\cnd S) = d(F(u)\cnd F(S))
$$
It is easy to see that
$$
\sum_u 2^{d'(u\cnd S)} S(u)=\sum_u 2^{d(F(u)\cnd F(S))}S(u)=\sum_v 2^{d(v\cnd F(S))}{}\cdot{}[F(S)](v)\le 1
$$
(in the second equality we group all the values of $u$ with the same $v=F(u)$). Therefore the maximality of $d$ guarantees that $d'(u\cnd S)\le d(u\cnd S)+O(1)$, so we get the required inequality.

This proof can be also rephrased using the definition of stochasticity with a priori probability. We need to show that for $y=P(x)$ and $Q=F(P)$ we have
$$
\frac{\mm(y\cnd Q)}{Q(y)}\le O(1)\cdot\frac{\mm(x\cnd P)}{P(x)}
$$
or
$$
\frac{\mm(F(x)\cnd F(P))\cdot P(x)}{Q(F(x))}\le O(\mm(x\cnd P)).
$$
It remains to note that the left hand side is a lower semicomputable function of $x$ and $P$ whose sum over all $x$ (for every $P$) is at most $1$. Indeed, if we group all terms with the same $F(x)$, we get the sum $\sum_y \mm(y\cnd F(P))\le 1$, since the sum of $P(x)$ over all $x$ with $F(x)=y$ equals $Q(y)$.
\end{proof}

\begin{remark}
In this proof it is important that we use the definition with distributions. If we replace is with the definition with finite sets, the results remains true with logarithmic precision, but the argument becomes more complicated, since the image of the uniform distribution may not be a uniform distribution. So if a set $A$ is a good model for $x$, we should not use $F(A)$ as a model for $F(x)$. Instead, we should look at the maximal $k$ such that $2^k\le \#F^{-1}(y)$, and consider the set of all $y'$ that have at least $2^k$ preimages in $A$.
\end{remark}

\begin{remark}\label{rem:non-total}
It is important in Proposition~\ref{prop:stoch-cons} that $F$ is a total function. If $x$ is some non-stochastic object and $x^*$ is the shortest program for $x$, then $x^*$ is incompressible and therefore stochastic. Still the interpreter (decompressor) maps $x^*$ to $x$. We discuss the case of non-total $F$ below, see Section~\ref{subsec:depth-appl}.
\end{remark}

\begin{remark}\label{rem:cons-f}
A similar argument shows that $d(F(x)\cnd F(P))\le d(x\cnd P)+\KP(F)+O(1)$ (for total $F$), so both $O(1)$-bounds in Proposition~\ref{prop:stoch-cons} may be replaced by $\KP(F)+O(1)$ where $O(1)$-constant does not depend on $F$ anymore.
\end{remark}

\subsection{Non-stochastic objects}

Note that up to now we have not shown that non-stochastic objects exist at all. It is easy to show that they exist for rather large values of $\alpha$ and $\beta$ (linearly growing with $n$).

\begin{proposition}[\cite{shen83}]\label{prop:existence-nonstochastic}
For some $c$ and all $n$:

(1)~if $\alpha+2\beta<n-c\log n$, then there exist $n$-bit strings that are not $(\alpha,\beta)$-stochastic;

(2)~however, if $\alpha+\beta>n+c\log n$, then every $n$-bit string is $(\alpha,\beta)$-stochastic.
\end{proposition}

Note that the term $c\log n$ allows us to use the definition with finite sets (i.e., uniform distributions on finite sets) instead of arbitrary finite distributions, since both versions are equivalent with $O(\log n)$-precision.

\begin{proof}
The second part is obvious (and is added just for comparison): if $\alpha+\beta=n$, then all $n$-bit strings can be split into $2^\alpha$ groups of size $2^\beta$ each. Then the complexity of each group is $\alpha+O(\log n)$, and the randomness deficiency of every string in the corresponding group is at most $\beta+O(1)$. It is slightly bigger than the bounds we need, but we have reserve $c\log n$, and $\alpha$ and $\beta$ can be decreased, say, by $(c/2)\log n$ before using this argument.

\emph{The first part}: Consider all  finite sets $A$ of strings that have complexity at most $\alpha$ and size at most $2^{\alpha+\beta}$. Since $\alpha+(\alpha+\beta)<n$, they cannot cover all $n$-bit strings. Consider then the first (say, in the lexicographical order) $n$-bit string $u$ not covered by any of these sets. What is the complexity of $u$? To specify $u$, it is enough to give $n,\alpha,\beta$ and the program of size at most $\alpha$ (from the definition of Kolmogorov complexity) that has maximal running time among programs of that size. Then we can wait until this program terminates and look at the outputs of all programs of size at most $\alpha$ after the same number of steps, select sets of strings of size at most $\alpha+\beta$, and take the first $u$ not covered by these sets. So the complexity of $u$ is at most $\alpha+O(\log n)$ (the last term is needed to specify $n,\alpha,\beta$). The same is true for conditional complexity with arbitrary condition, since it is bounded by the unconditional complexity. So  the randomness deficiency of $u$ in every set $A$ of size $2^{\alpha+\beta}$ is at least $\beta-O(\log n)$.  We see that $u$ is not $(\alpha,\beta-O(\log n))$-stochastic. Again the $O(\log n)$-term can be compensated by $O(\log n)$-change in $\beta$ (we have $c\log n$ reserve for that).
\end{proof}

\begin{remark}
There is a gap between lower and upper bounds provided by Proposition~\ref{prop:existence-nonstochastic}. As we will see later, the upper bound~(2) is tight with $O(\log n)$-precision, but we need more advanced technique (properties of two-part descriptions, Section~\ref{sec:two-part}) to prove this.
\end{remark}

Proposition~\ref{prop:existence-nonstochastic} shows that non-stochastic objects exist for rather large values of $\alpha$ and $\beta$ (proportional to $n$). This, of course, is a mathematical existence result; it does not say anything about the possibility to observe non-stochastic objects in the ``real world''. As we have discussed, random sampling (from a simple distribution) may produce a non-stochastic object only with a negligible probability;  \emph{total} algorithmic transformations (defined by programs of small complexity) also cannot not create non-stochastic object from stochastic ones. What about non-total algorithmic transformations? As we have discussed in Remark~\ref{rem:non-total}, a non-total computable transformation may transform a stochastic object into a non-stochastic one, but does it happen with non-negligible probability?

Consider a randomized algorithm that outputs some string. It can be considered as a deterministic algorithm applied to random bit sequence (generated by the internal coin of the algorithm). This deterministic algorithm may be non-total, so we cannot apply the previous result. Still, as the following result shows, randomized algorithms also generate non-stochastic objects only with small probability.

To make this statement formal, we consider the sum of $\mm(x)$ over all non-stochastic $x$ of length $n$. Since the a priori probability $\mm(x)$ is the upper bound for the output distribution of any randomized algorithm, this implies the same bound (up to $O(1)$-factor) for every randomized algorithm. The following theorem gives an upper bound for this sum:

\begin{proposition}[see~\cite{muchnik-meta}, Section 10]\label{prop:nonstochastic-counting}
$$
\sum \{\,\mm(x)\mid \text{$x$ is a $n$-bit string that is not $(\alpha,\alpha)$-stochastic}\ \} \le 2^{-\alpha+O(\log n)}
$$
for every $n$ and $\alpha$.
\end{proposition}

\begin{proof}
Consider the sum of $\mm(x)$ over \emph{all} strings of length $n$. This sum is some real number $\omega\le 1$. Let $\tilde\omega$ be the number represented by first $\alpha$ bits in the binary representation of $\omega$, minus $2^{-\alpha}$. We may assume that $\alpha\le O(n)$, otherwise all strings of length $n$ are $(\alpha,\alpha)$-stochastic.

Now construct a probability distribution as follows. All terms in a sum for $\omega$ are lower semicomputable, so we can enumerate increasing lower bounds for them. When the sum of these lower bounds exceeds  $\tilde\omega$, we stop and get some measure $P$ with finite support and rational values. Note that we have a measure, not a distribution, since the sum of $P(x)$ for all $x$ is less than $1$ (it does not exceed $\omega$). So we normalize $P$ (by some factor) to get a distribution $\tilde P$ proportional to $P$. The complexity of $\tilde P$ is bounded by $\alpha+O(\log n)$ (since $\tilde P$ is determined by $\tilde\omega$ and $n$). Note that the difference between $P$ (without normalization factor) and a priori probability $\mm$ (the sum of differences over all strings of length $n$) is bounded by $O(2^{-\alpha})$. It remains to show that for $\mm$-most strings the distribution $\tilde P$ is a good model.

Let us prove that the sum of a priori probabilities of all $n$-bit strings $x$ that have $d(x\cnd \tilde P)>\alpha+ c \log n$ is bounded by $O(2^{-\alpha})$, if $c$ is large enough. Indeed, for those strings we have
    $$
    -\log \tilde P(x) - \KP(x\cnd \tilde P) > \alpha+c \log n.
    $$
The complexity of $\tilde P$ is bounded by $\alpha+O(\log n)$ and therefore $\KP(x)$ exceeds $\KP(x\cnd \tilde P)$ at most by $\alpha+O(\log n)$, so
     $
-\log \tilde P(x) - \KP(x) > 1
     $
(or  $\tilde P(x)< \mm(x)/2$)
for those strings, if $c$ is large enough (it should exceed the constants hidden in $O(\log n)$ notation). The difference $1$ is enough for the estimate below, but we could have arbitrary constant or even logarithmic difference by choosing larger value of~$c$.

Prefix complexity can be defined in terms of a priori probability, so we get
     $$
 \log(\mm(x)/\tilde P(x))>1
     $$
for all $x$ that have deficiency exceeding $\alpha+c\log n$ with respect to $\tilde P$. The same inequality is true for $P$ instead of $\tilde P$, since $P$ is smaller. So for all those $x$ we have $P(x)<\mm(x)/2$, or $(\mm(x)-P(x))>\mm(x)/2$. Recalling that the sum of $\mathbf{m}(x)-P(x)$ over all $x$ of length $n$ does not exceed $O(2^{-\alpha})$ by construction of $\tilde\omega$, we conclude that the sum of $\mathbf{m}(x)$ over all strings of randomness deficiency (with respect to $\tilde P$) exceeding $\alpha+c\log n$ is at most $O(2^{-\alpha})$.

So we have shown that the sum of $\mm(x)$ for all $x$ of length $n$ that are not $(\alpha+O(\log n),\alpha+O(\log n))$-stochastic, does not exceed $O(2^{-\alpha})$. This differs from our claim only by $O(\log n)$-change in $\alpha$.
\end{proof}

Bruno Bauwens noted that this argument can be modified to obtain a stronger result where $(\alpha,\alpha)$-stochasticity is replaced by $(\alpha+O(\log n),O(\log n))$-stochasticity. Instead of one measure $P$, one should consider a family of measures. Let us approximate $\omega$ and look when the approximations cross the thresholds corresponding to $k$ first bits of the binary expansion of $\omega$. In this way we get $P=P_1+P_2+\ldots+P_\alpha$, where $P_i$ has total weight at most $2^{-i}$, and complexity at most $i+O(\log n)$. Let us show that all strings $x$ where $P(x)$ is close to $\mm(x)$ (say, $P(x)\ge \mm(x)/2$) are $(\alpha+O(\log n), O(\log n))$-stochastic, namely, one of the measures $P_i$ multiplied by $2^i$ is a good explanations for them. Indeed, for such $x$ and some $i$ the value of $P_i(x)$ coincides with $\mm(x)$ up to polynomial (in $n$) factor, since the sum of all $P_i$ is at least $\mm(x)/2$. On the other hand, $\mm(x\cnd 2^iP_i)\le 2^i \mm(x) \approx 2^i P_i(x)$, since the complexity of $2^iP_i$ is at most $i+O(\log n)$. Therefore the ratio $\mm(x\cnd P_i)/(2^iP_i(x))$ is polynomially bounded, and  the model $2^iP_i$ has deficiency $O(\log n)$. This better bound also follows from the Levin's explanation, see below.

This result shows that non-stochastic objects rarely appear as outputs of randomized algorithms. There is an explanation of this phenomenon (that goes back to Levin): non-stochastic objects provide a lot of information about halting problem, and the probability of appearance of an object that has a lot of information about some sequence $\alpha$, is small (for any fixed $\alpha$). We discuss this argument below, see Section~\ref{subsec:halting-information}.
\smallskip

It is natural to ask the following general question. For a given string $x$, we may consider the set of all pairs $(\alpha,\beta)$ such that $x$ is $(\alpha,\beta)$-stochastic. By definition, this set is upwards-closed: a point in this set remains in it if we increase $\alpha$ or $\beta$, so there is some boundary curve that describes the trade-off between $\alpha$ and $\beta$. What curves could appear in this way? To get an answer (to characterizes all these curves with $O(\log n)$-precision), we need some other technique, explained in the next section.

\section{Two-part descriptions}\label{sec:two-part}

Now we switch to another measure of the quality of a statistical model. It is important both for philosophical and technical reasons. The philosophical reason is that it corresponds to the so-called ``minimal description length principle''. The technical reason is that it is easier to deal with; in particular, we will use it to answer the question asked at the end of the previous section.

\subsection{Optimality deficiency}

Consider again some statistical model. Let $P$ be a probability distribution (with finite support and rational values) on strings. Then we have
$$
\KP(x)\le \KP(P)+\KP(x\cnd P)\le \KP(P)+(-\log P(x))
$$
for arbitrary string $x$ (with $O(1)$-precision). Here we use that (with $O(1)$-precision):
\begin{itemize}
\item $\KP(x\cnd P)\le -\log P(x)$, as we have mentioned;

\item the complexity of the pair is bounded by the sum of complexities: $\KP(u,v)\le \KP(u)+\KP(v)$;

\item $\KP(v)\le \KP(u,v)$ (in our case, $\KP(x)\le KP(x,P)$).
\end{itemize}

If $P$ is a uniform distribution on some finite set $A$, this inequality can be explained as follows.  We can specify $x$ in two steps:
\begin{itemize}
\item first, we specify $A$;
\item then we specify the ordinal number of $x$ in $A$ (in some natural ordering, say, the lexicographic one).
\end{itemize}
In this way we get $\KP(x)\le \KP(A)+\log \#A$ for every element $x$ of arbitrary finite set $A$. This inequality holds with $O(1)$-precision. If we replace the prefix complexity by the plain version, we can say that $\KS(x)\le\KS(A)+\log\#A$ with precision $O(\log n)$ for every string $x$ of length at most $n$: we may assume without loss of generality that both terms in the right hand side are at most $n$, otherwise the inequality is trivial.

The ``quality'' of a statistical model $P$ for a string $x$ can be measured by the difference between sides of this inequality: for a good model the ``two-part description'' should be almost minimal. We come to the following definition:

\begin{definition}
The \emph{optimality deficiency} of a distribution $P$ considered as the model for a string $x$ is  the difference
$$
\delta(x,P)=(\KP(P)+(-\log P(x)))-\KP(x).
$$
\end{definition}
As we have seen, $\delta(x,P)\ge 0$ with $O(1)$-precision.

If $P$ is a uniform distribution on a set $A$, the optimality deficiency $\delta(x,P)$ will also be denoted by $\delta(x,A)$, and
$$
\delta(x,A)=(\KP(A)+\log\#A) - \KP(x).
$$
The following proposition shows that we may restrict our attention to finite sets as models (with $O(\log n)$-precision):

\begin{proposition}\label{prop:models-to-sets-2}
Let $P$ be a distribution considered as a model for some string $x$ of length~$n$. Then there exists a finite set $A$ such that
$$
\KP(A)\le \KP(P)+O(\log n);\quad \log\#A\le -\log P(x)+O(1) \eqno(*)
$$
\end{proposition}

This proposition will be used in many arguments, since it is often easier to deal with sets as statistical models (instead of distributions). Note that the inequalities~$(*)$ evidently imply that
$$
\delta(x,A)\le \delta(x,P)+O(\log n),
$$
so arbitrary distribution $P$ may be replaced by a uniform one ($U_A$) with a logarithmic-only change in the optimality deficiency.

\begin{proof}
We use the same construction as in Proposition~\ref{prop:models-to-sets-1}. Let $2^{-k}$ be the maximal power of $2$ such that $2^{-k}\le P(x)$, and let $A=\{x\mid P(x)\ge 2^{-k}\}$. Then $k=-\log P(x)+O(1)$. We may assume that $k=O(n)$: if $k$ is much bigger than $n$, then $\delta(x,P)$ is also bigger than $n$ (since the complexity of $x$ is bounded by $n+O(\log n)$), and in this case the statement is trivial (let $A$ be the set of all $n$-bit strings).

Now we see that that $A$ is determined by $P$ and $k$, so $\KP(A)\le \KP(P)+\KP(k)\le \KP(P)+O(\log n)$. Note also that $\#A \le 2^k$, so $\log\# A\le -\log P(x)+O(1)$.
\end{proof}

Let us note that in a more general setting~\cite{milovanov-stacs} where we consider several strings as outcomes of the repeated experiment (with independent trials) and look for a model that explains all of them, a similar result is not true: not every probability distribution can be transformed into a uniform one.

\subsection{Optimality and randomness deficiencies}

Now we have two ``quality measures'' for a statistical model $P$: the randomness deficiency $d(x\cnd P)$ and the optimality deficiency $\delta(x,P)$. They are related:

\begin{proposition}\label{prop:randomness-optimality}
$$d(x\cnd P)\le \delta(x,P)$$
with $O(1)$-precision.
\end{proposition}

\begin{proof}
By definition
\begin{align*}
d(x\cnd P) &= - \log P(x) - \KP(x\cnd P);\\
\delta(x,P)&= -\log P(x) + \KP(P) - \KP(x).
\end{align*}
It remains to note that $\KP(x)\le \KP(x,P)\le \KP(P)+\KP(x\cnd P)$ with $O(1)$-precision.
\end{proof}

Could $\delta(x,P)$ be significantly larger than $d(x\cnd P)$? Look at the proof above: the second inequality $\KP(x,P)=\KP(P)+\KP(x\cnd P)$ is an equality with logarithmic precision. Indeed, the exact formula (Levin--G\'acs formula for the complexity of a pair with $O(1)$-precision) is
$$
\KP(x,P)=\KP(P)+\KP(x\cnd P,\KP(P)).
$$
Here the term $\KP(P)$ in the condition changes the complexity by $O(\log \KP(P))$, and we may ignore models $P$ whose complexity is much greater than the complexity of $x$.

On the other hand, in the first inequality the difference between $\KP(x,P)$ and $\KP(x)$ may be significant. This difference equals $\KP(P\cnd x)$ with logarithmic accuracy and, if it is large, then $\delta(x,P)$ is much bigger than $d(x\cnd P)$. The following example shows that this is possible. In this example we deal with sets as models.

\begin{example}\label{ex:stochasticity-optimality}
Consider an incompressible string $x$ of length $n$, so $\KP(x)=n$ (all equalities with logarithmic precision). A good model for this string is the set $A$ of all $n$-bit strings.  For this model we have $\#A=2^n$, $\KP(A)=0$ and $\delta(x,A)=n+0-n=0$ (all equalities have logarithmic precision). So $d(x\cnd P)=0$, too.  Now we can change the model by excluding some other $n$-bit string. Consider a $n$-bit string $y$ that is incompressible and independent of $x$: this means that $\KP(x,y)=2n$. Let $A'$ be $A\setminus\{y\}$.

The set $A'$ contains $x$ (since $x$ and $y$ are independent, $y$ differs from $x$). Its complexity is $n$ (since it determines $y$). The optimality deficiency is then $n + n - n =n$, but the randomness deficiency is still small: $d(x\cnd A')=\log \#A'-\KP(x\cnd A')= n - n = 0$ (with logarithmic precision). To see why $\KP(A'\cnd x)=n$, note that $x$ and $y$ are independent, and the set $A'$ has the same information as $(n,y)$.
\end{example}

One of the main results of this section (Theorem~\ref{th:deficiencies}) clarifies the situation: it implies that if optimality deficiency of a model is significantly larger than its randomness deficiency, then this model can be improved and another model with better parameters  can be found. More specifically, the complexity  of the new model is smaller than the complexity of the original one while both the randomness deficiency and optimality deficiency of the new model are not worse than the randomness deficiency of the original one. This is one of the main results of algorithmic statistics,  but first let us explore systematically the properties of two-part descriptions.

\subsection{Trade-off between complexity and size of a model}
\label{sub:trade-off}

It is convenient to consider only models that are sets (=uniform distribution on sets). We will call them \emph{descriptions}. Note that by Propositions~\ref{prop:models-to-sets-1} and~\ref{prop:models-to-sets-2} this restriction does not matter much since we ignore logarithmic terms. For a given string $x$ there are many different descriptions: we can have a simple large set containing $x$, and at the same time some more complicated, but smaller one. In this section we study the trade-off between these two parameters (complexity and size).

\begin{definition}\label{def:px}
A finite set $A$ is an $(i*j)$-description\footnote{This notation may look strange; however, we speak so often about finite sets of complexity at most $i$ and cardinality at most $2^j$ that we decided to introduce some short name and notation for them.} of $x$ if $x\in A$, complexity $\KP(A)$ is at most $i$, and $\log\#A\le j$. For a given $x$ we consider the set $P_x$ of all pairs $(i,j)$ such that $x$ has some $(i*j)$-description; this set will be called  \emph{the profile} of $x$.
\end{definition}

Informally speaking, an $(i*j)$-description for $x$ consists of two parts: first we spend $i$ bits to specify some finite set $A$ and then $j$ bits to specify $x$ as an element of $A$.

What can be said about $P_x$ for a string $x$ of length $n$ and complexity $k=\KP(x)$? By definition, $P_x$ is closed upwards and contains the points $(0,n)$ and $(k,0)$. Here we omit terms $O(\log n)$: more precisely, we have a $(O(\log n)*n)$-description that consists of all strings of length $n$, and a $((k+O(1))*0)$-description $\{x\}$. Moreover, the following proposition shows that  we can move the information from the second part of the description into its first part (leaving the total length almost unchanged). In this way we make the set smaller (the price we pay is that its complexity increases).

\begin{proposition}[\cite{kolm,gtv,shen99}]
     \label{prop:description-shift}
Let $x$ be a string and $A$ be a finite set that contains $x$. Let $s$ be a non-negative integer such that $s\le \log\# A$. Then there exists a finite set $A'$ containing $x$ such that $\#A' \le \# A/ 2^s$ and $\KP(A')\le \KP(A)+s + O(\log s)$.
\end{proposition}

\begin{proof}
List all the elements of $A$ in some (say, lexicographic) order. Then we split the list into $2^s$ parts (first $\#A/2^s$ elements, next $\#A/2^s$ elements etc.; we omit evident precautions for the case when $\#A$ is not a multiple of $2^s$). Then let $A'$ be the part that contains $x$. It has the required size. To specify $A'$, it is enough to specify $A$ and the part number; the latter takes at most $s$ bits. (The logarithmic term is needed to make the encoding of the part number self-delimiting.)
\end{proof}

This statement can be illustrated graphically. As we have said, the set $P_x$ is ``closed upwards'' and contains with each point $(i,j)$ all points on the right (with bigger $i$) and on the top (with bigger $j$).  It contains points $(0,n)$ and $(\KP(x),0)$; Proposition~\ref{prop:description-shift} says that we can also move down-right adding $(s,-s)$ (with logarithmic precision). We will see that movement in the opposite direction is not always possible. So, having two-part descriptions with the same total length, we should prefer the one with bigger set (since it always can be converted into others, but not vice versa).

The boundary of $P_x$ is some curve connecting the points $(0,n)$ and $(k,0)$. This curve (introduced by Kolmogorov in 1970s, see~\cite{kolmmmo}) never gets into the triangle $i+j< \KP(x)$ and always goes down (when moving from left to right) with slope at least $-1$ or more.

\begin{figure}[h]
\begin{center}\includegraphics{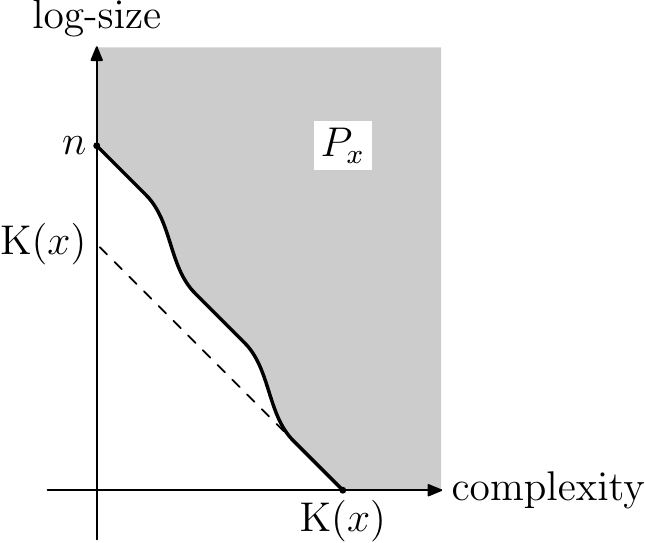}\end{center}
\caption{The set $P_x$ and its boundary curve}\label{mdl.1.eps}
\end{figure}

This picture raises a natural question: which boundary curves are possible and which are not? Is it possible, for example, that the boundary goes along the dotted line on Figure~\ref{mdl.1.eps}? The answer is positive: take a random string of desired complexity and add trailing zeros to achieve desired length. Then the point $( 0,\KP(x))$ (the left end of the dotted line) corresponds to the set $A$ of all strings of the same length having the same trailing zeros. We know that the boundary curve cannot go down slower than with slope $-1$ and that it lies above the line $i+j=\KP(x)$, therefore it follows the dotted line (with logarithmic precision).

A more difficult question: is it possible that the boundary curve starts from $( 0,n)$, goes with the slope $-1$ to the very end and then goes down rapidly to $(\KP(x),0)$ (Figure~\ref{mdl.2}, the solid line)? Such a string $x$, informally speaking, would have essentially only two types of statistical explanations: a set of all strings of length $n$ (and its parts obtained by Proposition~\ref{prop:description-shift}) and the exact description, the singleton $\{x\}$.

\begin{figure}[h]
\begin{center}\includegraphics{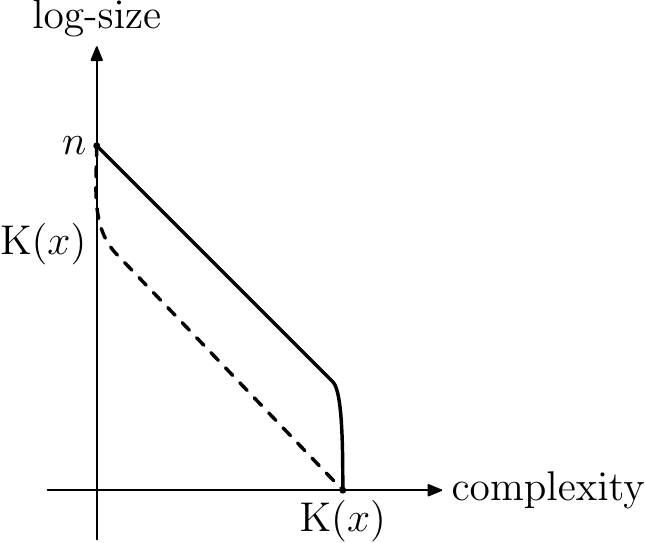}\end{center}
\caption{Two opposite possibilities for a boundary curve}\label{mdl.2}
\end{figure}

It turns out that not only these two opposite cases are possible, but also all intermediate curves (provided they decrease with slope $-1$ or faster, and are simple enough), at least with logarithmic precision. More precisely, the following statement holds:

\begin{thm}[\cite{vv}]\label{stat-any-curve}
Let $k\le n$ be two integers and let $t_0 > t_1 > \ldots > t_k$ be a strictly decreasing sequence of integers such that $t_0\le n$ and $t_k=0$; let $m$ be the complexity of this sequence. Then there exists a string $x$ of complexity $k+O(\log n)+O(m)$ and length $n+O(\log n)+O(m)$ for which the boundary curve of $P_x$ coincides with the line $(0,t_0)$--$(1,t_1)$--\ldots--$(k,t_k)$ with $O(\log n)+O(m)$ precision: the distance between the set $P_x$ and the set $T=\{( i,j)\mid (i<k)\Rightarrow (j>t_i)\}$ is bounded by $O(\log n)+O(m)$.
\end{thm}

(We say that the distance between two subsets $P,Q\subset\mathbb{Z}^2$ is at most $\varepsilon$ if $P$ is contained in the $\varepsilon$-neighborhood of $Q$ and vice versa.)

\begin{proof}
For every $i$ in the range $0\ldots k$ we list all the sets of complexity at most $i$ and size at most $2^{t_i}$. For a given $i$ the union of all these sets is denoted by $S_i$. It contains at most $2^{i+t_i}$ elements. (Here and later we omit constant factors and factors polynomial in $n$ when estimating cardinalities, since they correspond to $O(\log n)$ additive terms for lengths and complexities.) Since the sequence $t_i$ strictly decreases (this corresponds to slope $-1$ in the picture), the sums $i+t_i$ do not increase, therefore each $S_i$ has at most $2^{t_0}\le 2^n$ elements. The union of all $S_i$ therefore also has at most $2^n$ elements (up to a polynomial factor, see above). Therefore, we can find a string of length $n$ (actually $n+O(\log n)$) that does not belong to any $S_i$. Let $x$ be a first such string in some order (e.g., in the lexicographic order).

By construction, the set $P_x$ lies above the curve determined by $t_i$. So we need to estimate the complexity of $x$ and prove that $P_x$ follows the curve (i.e., that $T$ is contained in the neighborhood of $P_x$).

Let us start with the upper bound for the complexity of $x$. The list of all objects of complexity at most $k$ plus the full table of their complexities have complexity $k+O(\log k)$, since it is enough to know $k$ and the number of terminating programs of length at most~$k$. Except for this list, to specify $x$ we need to know $n$ and the sequence $t_0,\ldots,t_k$, whose complexity is $m$.

The lower bound: the complexity of $x$ cannot be less than $k$ since all the singletons of this complexity were excluded (via $S_k$).

It remains to show that for every $i\le k$ we can put $x$ into a set $A$ of complexity $i$ (or slightly bigger) and size $2^{t_i}$ (or slightly bigger). For this we enumerate a sequence of sets of correct size and show that one of the sets will have the required properties; if this sequence of sets is not very long, the complexity of its elements is bounded. Here are the details.

We start by taking the first $2^{t_i}$ strings of length $n$ as our first set $A$. Then we start enumerating all finite sets of complexity at most $j$ and of size at most $2^{t_j}$ for all $j=0,\ldots,k$, and get an enumeration of all sets $S_j$. Recall that all elements of all $S_j$ should be deleted (and the minimal remaining element should eventually be $x$). So, when a new set of complexity at most $j$ and of size at most $2^{t_j}$ appears, all its elements are included in $S_j$ and deleted. Until all elements of $A$ are deleted, we have nothing to worry about, since $A$ is covering the minimal remaining element. If (and when) all elements of $A$ are deleted, we replace $A$ by a new set that consists of first $2^{t_i}$ undeleted (yet) strings of length $n$. Then we wait again until all the elements of this new $A$ are deleted, if (and when) this happens, we take $2^{t_i}$ first undeleted elements as new $A$, etc.

The construction guarantees the correct size of the sets and that one of them covers $x$ (the minimal non-deleted element). It remains to estimate the complexity of the sets we construct in this way.

First, to start the process that generates these sets, we need to know the length $n$ (actually something logarithmically close to $n$) and the sequence $t_0,\ldots,t_k$. In total we need $m+O(\log n)$ bits. To specify each version of $A$, we need to add its version number. So we need to show that the number of different $A$'s that appear in the process is at most $2^i$ or slightly bigger.

A new set $A$ is created when all the elements of the old $A$ are deleted. These changes can be split into two groups. Sometimes a new set of complexity $j$ appears with $j\le i$. This can happen only $O(2^i)$ times since there are at most $O(2^i)$ sets of complexity at most $i$. So we may consider the other changes (excluding the first changes after each new large set was added). For those changes all the elements of $A$ are gone due to elements of $S_j$ with $j>i$. We have at most $2^{j+t_j}$ elements in $S_j$. Since $t_j+j\le t_i+i$, the total number of deleted elements only slightly exceeds $2^{t_i+i}$, and each set $A$ consists of $2^{t_i}$ elements, so we get about $2^i$ changes of $A$.
\end{proof}

\begin{remark}
It is easy to modify the proof to get a string $x$ of length exactly $n$. Indeed, we may consider slightly smaller bad sets: decreasing the logarithms of their sizes by $O(\log n)$, we can guarantee that the total number of elements in all bad sets is less than $2^n$. Then there exists a string of length $n$ that does not belong to bad sets. In this way the distance between $T$ and $P_x$ may increase by $O(\log n)$, and this is acceptable.
\end{remark}

Theorem~\ref{stat-any-curve} shows that the value of the complexity of $x$ does not describe the properties of $x$ fully; different strings of the same complexity $x$ can have different boundary curves of $P_x$. This curve can be considered as an ``infinite-dimensional'' characterization of $x$.

Strings $x$ with minimal possible $P_x$ (Figure~\ref{mdl.2}, the upper curve) may be called \emph{antistochastic}. They have quite unexpected properties. For example, if we replace some bits of an antistochastic string $x$ by stars (or some other symbols indicating erasures) leaving only $\KP(x)$ non-erased bits, then the  string $x$ can be reconstructed from the resulting string $x'$ with logarithmic advice, i.e., $\KP(x\cnd x')=O(\log n)$. This and other properties of antistochastic strings were discovered in~\cite{milovanov-antistochastic}.

\subsection{Optimality and randomness deficiency}\label{subsec:opt-rand}

In this section we establish the connection between optimality and randomness deficiency. As we have seen, the optimality deficiency can be bigger than the randomness deficiency (for the same description), and the difference is
$
\delta(x,A)-d(x\cnd A)=\KP(A)+\KP(x\cnd A) - \KP(x).
$
The Levin--G\'acs formula for the complexity of pair ($\KP(u,v)=\KP(u)+\KP(v\cnd u)$ with logarithmic precision, for $O(1)$-precision one needs to add $\KP(u)$ in the condition, but we ignore logarithmic size terms anyway) shows that the difference in question can be rewritten as
$$
\delta(x,A)-d(x\cnd A)=\KP(A,x)-\KP(x)=\KP(A\cnd x).
$$
So if the difference between deficiencies for some $(i*j)$-description $A$ of $x$ is big, then $\KP(A\cnd x)$ is big. All the $(i*j)$-descriptions of $x$ can be enumerated if $x$, $i$, and $j$ are given. So the large value of $\KP(A\cnd x)$ for some $(i*j)$-description $A$ means that there are many $(i*j)$-descriptions of $x$, otherwise $A$ can be reconstructed from $x$ by specifying $i,j$ (requires $O(\log n)$ bits) and the ordinal number of $A$ in the enumeration. We will prove that if there are many $(i*j)$-descriptions for some $x$, then there exist a description with better parameters.

Now we explain this in more detail. Let us start with the following remark. Consider all strings that have $(i*j)$-descriptions for some fixed $i$ and $j$. They can be enumerated in the following way: we enumerate all finite sets of complexity at most $i$, select those sets that have size at most $2^j$, and include all elements of these sets into the enumeration. In this construction
\begin{itemize}
\item the complexity of the enumerating algorithm is logarithmic (it is enough to know $i$ and $j$);
\item we enumerate at most $2^{i+j}$ elements;
\item the enumeration is divided into at most $2^i$ ``portions'' of size at most $2^j$.
\end{itemize}
It is easy to see that any other enumeration process with these properties enumerates only objects that have $(i*j)$-descriptions (again with logarithmic precision). Indeed, each portion is a finite set that can be specified by its ordinal number and the enumeration algorithm, the first part requires $i+O(\log i)$ bits, the second is of logarithmic size according to our assumption.

\begin{remark}\label{rem:portion}
The requirement about the portion size is redundant. Indeed, we can change the algorithm by splitting large portions into pieces of size $2^j$ (the last piece may be incomplete). This, of course, increases the number of portions, but if the total number of enumerated elements is at most $2^{i+j}$, then this splitting adds at most $2^i$ pieces. This observation looks (and is) trivial, still it plays an important role in the proof of the following proposition.
\end{remark}

\begin{proposition}\label{prop:improving-descriptions}
If a string $x$ of length $n$ has at least $2^k$ different $(i,j)$-descriptions, then $x$ has some $(i*(j-k))$-description and even some $((i-k)*j)$-description.
\end{proposition}

Again we omit logarithmic term: in fact one should write $((i+O(\log n))*(j-k+O(\log n)))$, etc. The word ``even'' in the statement refers to Proposition~\ref{prop:description-shift} that shows that indeed the second claim is stronger.

\begin{proof}
Consider the enumeration of all objects having $(i*j)$-descriptions in $2^i$ portions of size $2^j$ (we ignore logarithmic additive terms and respective polynomial factors) as explained above. After each portion (i.e., new $(i*j)$-description) appears, we count the number of descriptions for each enumerated object and select objects that have at least $2^k$ descriptions. Consider a new enumeration process that enumerates only these ``rich'' objects (rich = having  many descriptions). We have at most $2^{i+j-k}$ rich objects (since they appear in the list of size $2^{i+j}$ with multiplicity $2^k$), enumerated in $2^i$ portions (new portion of rich objects may appear only when a new portion appears in the original enumeration). So we apply the observation above to conclude that all rich objects have $(i*(j-k))$-descriptions.

To get the second (stronger) statement we need to decrease the number of portions (while not increasing too much the number of enumerated objects). This can be done using the following trick: when a new rich object (having $2^k$ descriptions) appears, we enumerate not only rich objects, but also ``half-rich'' objects, i.e., objects that currently have at least $2^k/2$ descriptions. In this way we enumerate more objects --- but only twice more. At the same time, after we dumped all half-rich objects, we are sure that next $2^k/2$ new $(i*j)$-descriptions will not create new rich objects, so the number of portions is divided by $2^k/2$, as required.
\end{proof}

Let us say more accurately how we deal with logarithmic terms. We may assume that $i,j=O(n)$, otherwise the claim is trivial. Then we allow polynomial (in $n$) factors and $O(\log n)$ additive terms in all our considerations.

\begin{remark}
If we unfold this construction, we see that new descriptions (of smaller complexity) are not selected from the original sequence of descriptions but constructed from scratch. In Section~\ref{sec:restricted-type} we deal with much more complicated case where we restrict ourselves to descriptions from some class (say, Hamming balls). Then the proof given above does not work, since the description we construct is not a ball even if we start with ball descriptions. Still some other (much more ingenious) argument can be used to prove a similar result for the restricted case.
\end{remark}

Now we are ready to prove the promised results (see the discussion after Example~\ref{ex:stochasticity-optimality}).

\begin{thm}~\label{thm:improving-descriptions}
If a string $x$ of length $n$ is $(\alpha,\beta)$-stochastic, then there exists some finite set $B$ containing $x$ such that $\KP(B)\le \alpha+O(\log n)$ and $\delta(x,B)\le \beta+O(\log n)$.
\end{thm}

\begin{proof}
Since $x$ is $(\alpha,\beta)$-stochastic, there exists some finite set $A$ such that $\KP(A)\le\alpha$ and $d(x\cnd A)\le \beta$. Let $i=\KP(A)$ and $j=\log\# A$, so $A$ is an $(i*j)$-description of $x$. We may assume without loss of generality that both $\alpha$ and $\beta$ (and therefore $i$ and $j$) are $O(n)$, otherwise the statement is trivial. The value $\delta(x,A)$ may exceed $d(x\cnd A)$, as we have discussed at the beginning of this section. So we assume that
$$
k=\delta(x,A)-d(x\cnd A)>0;
$$
if not, we can let $B=A$. Then, as we have seen, $\KP(A\cnd x)\ge k-O(\log n)$, and there are at least $2^{k-O(\log n)}$ different $(i*j)$-descriptions of $x$. According to Proposition~\ref{prop:improving-descriptions}, there exists some finite set $B$ that is an $(i*(j-k+O(\log n)))$-description of $x$. Its optimality deficiency $\delta(x,B)$ is $(k-O(\log n))$-smaller (compared to $A$) and therefore $O(\log n)$-close to $d(x\cnd A)$.	
\end{proof}

In this argument we used the simple part of Proposition~\ref{prop:improving-descriptions}. Using the stronger statement about complexity decrease, we get the following result:

\begin{thm}[\cite{vv}]\label{th:deficiencies}
Let $A$ be a finite set containing a string $x$ of length $n$ and let $k=\delta(x,A)-d(x\cnd A)$. Then there is a finite set $B$ containing $x$ such that $\KP(B)\le \KP(A)-k+O(\log n)$ and $\delta(x,B)\le d(x\cnd A)+O(\log n)$.
\end{thm}

\begin{proof}
Indeed, if $B$ is an $((i-k)*j)$-description of $x$ (up to logarithmic terms, as usual), then its optimality deficiency is again $(k-O(\log n))$-smaller (compared to $A$) and therefore $O(\log n)$-close to $d(x\cnd A)$.
\end{proof}

Note that the statement of the theorem implies that $d(x\cnd B)\le d(x\cnd A)+O(\log n)$.

Theorem~\ref{thm:improving-descriptions} and Proposition~\ref{prop:randomness-optimality} show that we can replace the randomness deficiency in the definition of $(\alpha,\beta)$-stochastic strings by the optimality deficiency (with logarithmic precision). More specifically, for every string $x$ of length $n$ consider the sets\label{def:qx}
$$
Q_x=\{( \alpha,\beta)\mid \text{$x$ is
$(\alpha,\beta)$-stochastic}\},
$$
and
$$
\tilde Q_x=\{( \alpha,\beta)\mid \text{there exists $A\ni x$ with }
\KP(A)\le \alpha,\ \delta(x,A)\le\beta)\}.
$$
Then these sets are at most $O(\log n)$ apart (each is contained in the $O(\log n)$-neighborhood of the other one).

This remark, together with the existence of antistochastic strings of given complexity and length, allows us to improve the result about the existence of non-stochastic objects (Proposition~\ref{prop:existence-nonstochastic}).

\begin{proposition}[\protect{\cite[Theorem IV.2]{gtv}}]\label{prop:existence-nonstochastic-strong}
For some $c$ and for all $n$: if $\alpha+\beta<n-c\log n$, there exist strings of length $n$ that are not $(\alpha,\beta)$-stochastic.
\end{proposition}

\begin{figure}[h]
\begin{center}
\includegraphics[scale=1]{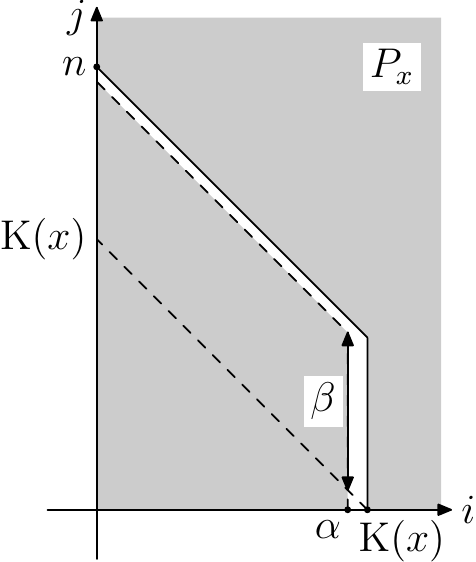}
\end{center}
\caption{Non-stochastic strings revisited. Left gray area corresponds to descriptions $A$ with $\KP(A)\le\alpha$ and $\delta(x,A)\le\beta$.}\label{mdl-9}
\end{figure}

\begin{proof}
Assume that integers $n,\alpha,\beta$ are given  such that $\alpha+\beta<n-c\log n$ (where the constant $c$ will be chosen later).  Let $x$ be an antistochastic string of length $n$  that has complexity $\alpha+d$ where $d$ is some positive number (see below about the choice of $d$). More precisely, for every given $d$ there exists a string $x$ whose complexity is $\alpha + d+O(\log n)$, length is $n+O(\log n)$, and the set $P_x$ is $O(\log n)$-close to the upper gray area (Figure~\ref{mdl-9}).

Assume that $x$ is $(\alpha,\beta)$-stochastic. Then (Theorem~\ref{thm:improving-descriptions}) the string $x$ has an  $(i*j)$-description with $i\le\alpha$ and $i+j\le\KP(x)+\beta$ (with logarithmic precision). The set of pairs $(i,j)$ satisfying these inequalities is shown as the lower gray area.   We have to choose $c$ in such a way that for some $d$  these two gray are disjoint and even separated by a gap of logarithmic size  (since they are known only with $O(\log n)$-precision). Note first that for $d=c'\log n$ with large enough $c'$ we guarantee the vertical gap (the vertical segments of the boundaries of two gray areas are far apart). Then we select $c$ large enough to guarantee that the diagonal segments of the boundaries of two gray areas are far apart ($\alpha+\beta<n$ with logarithmic margin).
\end{proof}

The transition from randomness deficiency to optimality deficiency (Theorem~\ref{thm:improving-descriptions}) has the following geometric interpretation.

\begin{thm}\label{thm:def-opt}
The sets $Q_x$ and $P_x$ are related to each other via an affine transformation $(\alpha,\beta)\mapsto (\alpha, \KP(x)-\alpha+\beta)$, as Figure~\ref{mdl.8} shows.\footnote{Technically speaking, this holds only for $\alpha\le\KP(x)$. For $\alpha>\KP(x)$ both sets contain all pairs with first component $\alpha$.}
\end{thm}

\begin{figure}[h]
\begin{center}\includegraphics{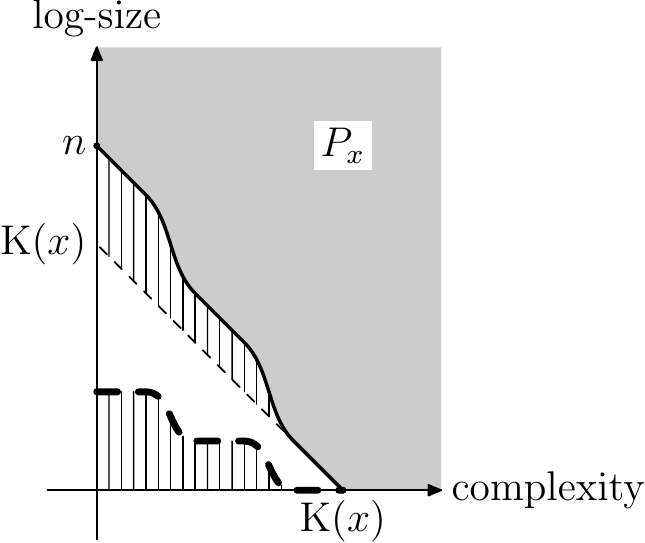}\end{center}
\caption{The set $P_x$ and the boundary of the set $Q_x$ (bold dotted line); on every vertical line two intervals have the same length.}\label{mdl.8}
\end{figure}

As usual, this statement is true with logarithmic accuracy: the distance between the image of the set $Q_x$ under this transformation and the set $P_x$ is claimed to be $O(\log n)$ for string $x$ of length $n$.

\begin{proof}
As we have seen, we may use the optimality deficiency instead of randomness deficiency, i.e., use the set $\tilde Q_x$ in place of $Q_x$. The preimage of the pair $(i,j)$ under our affine transformation is the pair $(i,i+j-\KP(x))$. Hence we have to prove that a pair $(i,j)$ is in $P_x$ if and only if the pair $(i,i+j-\KP(x))$ is in $\tilde Q_x$. Note that $\KP(A)=i$ and $\log\#A=j$ is equivalent to $\KP(A)=i$ and $\delta(x,A)=i+j-\KP(x)$ just by definition of $\delta(x,A)$. (See Figure~\ref{mdl.8}: the optimality deficiency of a description $A$ with $\KP(A)=i$ and $\log\#A =j$ is the vertical distance between $(i,j)$ and the dotted line.)

But there is some technical problem:  in the definition of $P_x$ we used inequalities $\KP(A)\le i$ and $\log\#A\le j$, not the equalities $\KP(A)=i$ and $\log\#A=j$. The same applies to the definition of $\tilde Q_x$. So we have two sets that correspond to each other, but their $\le$-closures could be different. Obviously, $\KP(A)\le i$ and $\log\#A\le j$ imply $\KP(A)\le i$ and $\KP(A)+\log\#A -\KP(x)\le i+j-\KP(x)$, but not vice versa.

In other words, the set of pairs $(\KP(A),\log\#A)$ satisfying the latter inequalities (see the right set on Figure~\ref{mdl-10-11}) is bigger than the set of pairs $(\KP(A),\log\#A)$ satisfying the former inequalities (see the left  set   on Figure~\ref{mdl-10-11}).
\begin{figure}[h]
\begin{center}
\includegraphics[scale=1]{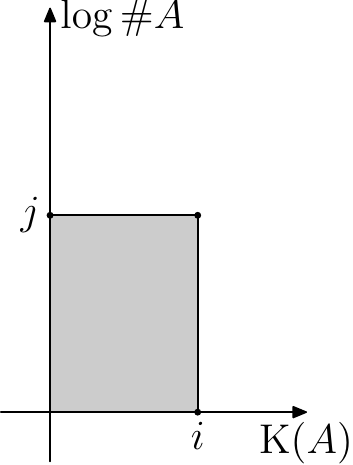} \qquad\qquad
\includegraphics[scale=1]{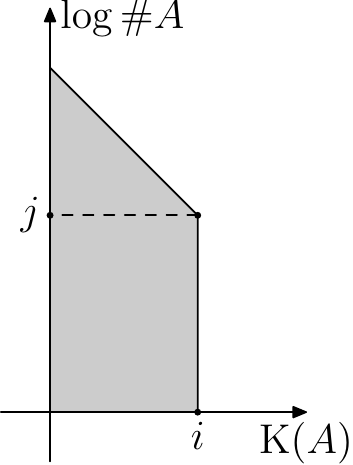}
\end{center}
\caption{The left picture shows (for given $i$ and $j$) the set of all pairs $(\KP(A),\log\#A)$ such that $\KP(A)\le i$ and $\log\#A\le j$; the right picture shows the pairs $(\KP(A),\log\#A)$ such that $\KP(A)\le i$ and $\delta(x,A)\le i+j-\KP(x)$.}\label{mdl-10-11}
\end{figure}
Now Proposition~\ref{prop:description-shift} helps: we may use it to convert any set with parameters from the right region into a set with parameters from the left region.
\end{proof}

\begin{remark}
Let us stress again that Theorem~\ref{thm:improving-descriptions} claims only that the \emph{existence} of a set $A\ni x$  with $\KP(A)\le\alpha$ and $d(x\cnd A)\le\beta$ is  equivalent to  the existence of a set $B\ni x$ with $\KP(B)\le\alpha$ and $\delta(x\cnd A)\le\beta$ (with logarithmic accuracy). The theorem does \emph{not} claim that for \emph{every} set $A\ni x$ with complexity at most $\alpha$ the  inequalities $d(x\cnd A)\le \beta$ and  $\delta(x,A)\le \beta$  are equivalent (with logarithmic accuracy). Indeed, the Example~\ref{ex:stochasticity-optimality} shows that this is not true: the first inequality does not imply the second one in general case. However, Theorems~\ref{thm:improving-descriptions} and~\ref{th:deficiencies} show that this can happen only for non-minimal descriptions (for which the description with smaller complexity and the same optimality deficiency) exists. Later we will see that all the minimal descriptions of the same (or almost the same) complexity have almost the same information. Moreover, if $A$ and $B$ are minimal descriptions and the complexity of $A$ is less than that  of $B$ then $\KS(A\cnd B)$ is small.
\end{remark}

For the people with taste for philosophical speculations the meaning of Theorems~\ref{thm:improving-descriptions} and~\ref{th:deficiencies} can be advertised as follows. Imagine several scientists that compete in providing a good explanation for some data $x$. Each explanation is a finite set $A$ containing $x$ together with a program $p$ that computes $A$.

How should we compare different explanations? We want the randomness deficiency $d(x\cnd A)$ of $x$ in $A$ to be negligible (no features of $x$ remain unexplained). Among these descriptions we want to find the simplest one (with the shortest $p$). That is, we look for a set $A$ corresponding to the point where the bold dotted line on Fig.~\ref{mdl.8} touches the horizontal axis. (In fact, there is always some trade-off between the parameters, not the specific exact point where the curve touches the horizontal axis, but we want to keep the discussion simple though imprecise.)

However, this approach meets the following obstacle: we are unable to compute randomness deficiency $d(x\cnd A)$. Moreover, the inventor of the model $A$ has no ways to convince us that the deficiency is indeed negligible if it is the case (the function $d(x\cnd A)$ is not even upper semicomputable). What could be done? Instead, we may look for an explanation with (almost) minimal sum $\log\#A+|p|$ (minimum description length principle). Note that this quantity is known for competing explanation proposals. Theorems~\ref{thm:improving-descriptions} and~\ref{th:deficiencies} provide the connection between these two approaches.

Returning to mathematical language, we have seen in this section that two approaches (based on $(i*j)$-descriptions and $(\alpha,\beta)$-stochasticity) produce essentially the same curve, though in different coordinates. The other ways to get the same curve will be discussed in Sections~\ref{sec:bcl} and~\ref{sec:depth}.

\subsection{Historical remarks}
The idea to consider $(i*j)$-descriptions with optimal parameters can be traced back to Kolmogorov. There is a short record for his talk given in 1974~\cite{kolmmmo}. Here is the (full) translation of this note:
\begin{quote}
For every constructive object $x$ we may consider a function $\Phi_x(k)$ of an integer argument $k\ge 0$ defined as a logarithm of the minimal cardinality of a set of complexity at most $k$ containing $x$. If $x$ itself has a simple definition, then $\Phi_x(1)$ is equal to one [a typo: cardinality equals $1$, and logarithm equals $0$] already for small $k$. If such a simple definition does not exist, $x$ is ``random'' in the negative sense of the word ``random''. But $x$ is positively ``probabilistically random'' only if the function $\Phi$ has a value $\Phi_0$ for some relatively small $k$ and then decreases approximately as $\Phi(k)=\Phi_0 -(k-k_0)$. [This corresponds to approximate $(k_0,0)$-stochasticity.]
\end{quote}

Kolmogorov also gave a talk in 1974~\cite{kolm}; the content of this talk was reported by Cover~\cite[Section 4, page 31]{cover1985}. Here $l(p)$ stands for the length of a binary string $p$ and $|S|$ stands for the cardinality of a set $S$.
\begin{quote}
4. \textbf{Kolmogorov's $H_k$ Function}

Consider the function $H_k\colon \{0,1\}^k\to N$,  $H_k(x)=\min_{p\colon l(p)\le k} \log |S|$, where the minimum is taken over all subsets $S\subseteq\{0,1\}^n$, such that $x\in S$, $U(p)=S$, $l(p)\le k$. This definition was introduces by Kolmogorov in a talk at the Information Symposium, Tallinn, Estonia, in 1974. Thus $H_k(x)$ is the log of the size of the smallest set containing $x$ over all sets specifiable by a program of $k$ or fewer bits. Of special interest is the value
$$
k^*(x)=\min\{k\colon H_k(x)+k=K(x)\}.
$$
Note that $\log |S|$ is the maximal number of bits necessary to describe an arbitrary element $x\in S$. Thus a program for $x$ can be written in two stages: ``Use $p$ to print the indicator function for $S$; the desired sequence is the $i$th sequence in a lexicographic ordering of the elements of this set''. This program has length $l(p)+\log |S|$, and $k^*(x)$ is the length of the shortest program $p$ for which this $2$-stage description is as short as the best $1$-stage description $p^*$. We observe that $x$ must be maximally random with respect to $S$ --- otherwise the $2$-stage description could be improved, contradicting the minimality of $K(x)$. Thus $k^*(x)$ and its associated program $p$ constitute a minimal sufficient description for $x$. $\langle\ldots\rangle$

Arguments can be provided to establish that $k^*(x)$ and its associated set $S^*$ describe all of the ``structure'' of $x$. The remaining details about $x$ are conditionally maximally complex. Thus $pp^{**}$, the program for $S^*$, plays the role of a sufficient statistic.

\end{quote}
In both places Kolmogorov speaks about the place when the boundary curve of $P_x$ reaches its lower bound determined by the complexity of $x$.

Later the same ideas were rediscovered and popularized by many people. Koppel in~\cite{koppel87} reformulates the definition using total algorithms. Instead of a finite set $A$ he considered a total program $P$ that terminates on all strings of some length. The two-part description of some $x$ is then formed by this program $P$ and the input $D$ for this program that is mapped to $x$. In our terminology this corresponds to the set $A$ of all values of $P$ on the strings of the same length as $D$. He writes then~\cite[p.~1089]{koppel87}
\begin{quote}
\textbf{Definition 3.} The $c$-sophistication of a finite string $S$ [is defined as] $$\text{SOPH}_c(S)=\min\{|P|\mid \exists D\text{ s. t. } (P,D)\text{ is a $c$-minimal description of $\alpha$}\}.$$
\end{quote}
There is a typo in this paper: $S$ should be replaced by $\alpha$ (two times). Before in Definition 1 the description is called $c$-minimal if $|P|+|D|\le H(\alpha)+c$ (here $P$ and $D$ are the program and and its input, respectively, $H$ stands for complexity).

Though  this paper (as well as the subsequent papers~\cite{koppel,koppel-atlan}) is not technically clear (e.g., it does not say what are the requirements for the algorithm $U$ used in the definition, and in~\cite{koppel,koppel-atlan} only universality is required, which is not enough: if $U$ is not optimal, the definition does not make sense), the philosophic motivation for this notion is explained clearly~\cite[p.~1087]{koppel87}:
\begin{quote}
The total complexity of an object is defined as the size of its most concise description. The total complexity of an object can be large while its ``meaningful'' complexity is low; for example, a random object is by definition maximally complex but completely lacking in structure.

$\langle\ldots\rangle$ The ``static'' approach to the formalization of meaningful complexity is ``sophistication'' defined and discussed by Koppel and Atlan [reference to unpublished paper ``Program-length complexity, sophistication, and induction'' is given, but later a paper of same authors~\cite{koppel-atlan} with a similar title appeared]. Sophistication is a generalization of the ``H-function'' or ``minimal sufficient statistic'' by Cover and Kolmogorov~$\langle\ldots\rangle$ The sophistication of an object in the size of that part of that object which describes its structure, i.e. the aggregate of its projectible properties.
\end{quote}

One can also mention the formulation of ``minimal description length'' principle by Rissanen~\cite{rissanen78}; the abstract of this paper says: ``Estimates of both integer-valued structure parameters and real-valued system parameters may be obtained from a model based on the shortest data description principle''; here ``integer-valued structure parameters'' may correspond to the choice of a statistical hypothesis (description set) while ``real-valued system parameters'' may correspond to the choice of a specific element in this set. The author then says that ``by finding the model which minimizes the description length one obtains estimates of both the integer-valued structure parameters and the real-valued system parameters''.

We do not try here to follow the development of these and similar ideas. Let us mention only that the traces of the same ideas (though even more vague) could be found in 1960s in the classical papers of Solomonoff~\cite{sol1,sol2} who tried to use shortest descriptions for inductive inference (and, as a side product, gave the definition of complexity later rediscovered by Kolmogorov~\cite{kolm65}). One may also mention a ``minimum message length principle'' that goes back to~\cite{wallace-boulton68}; the idea of two-part description is explained in~\cite{wallace-boulton68} as follows:
\begin{quote}
If the things are now classified then the measurements can be recorded by listing the following:

1. The class to which each thing belongs.

2. The average properties of each class.

3. The deviations of each thing from the average properties of its parent class.

If the things are found to be concentrated in a small area of the region of each class in the measurement space then the deviations will be small, and with reference to the average class properties most of the information about a thing is given by naming the class to which it belongs. In this case the information may be recorded much more briefly than if a classification had not been used. We suggest that the best classification is that which results in the briefest recording of all the attribute information.
\end{quote}
Here the ``class to which thing belongs'' corresponds to a set (statistical model, description in our terminology); the authors say that if this set is small, then only few bits need to be added to the description of this set to get a full description of the thing in question.

The main technical results of this sections  (Theorems~\ref{stat-any-curve},~\ref{thm:improving-descriptions}, and~\ref{th:deficiencies}) are taken from~\cite{vv} (where some historical account is provided).

\section{Bounded complexity lists}\label{sec:bcl}

In this section we show one more classification of strings that turns out to be  equivalent (up to coordinate change) to the previous ones: for a given string $x$ and $m\ge \KS(x)$ we look how close $x$ is to the end in the enumeration of all strings of complexity at most $m$. For technical reasons it is more convenient to use plain complexity $\KS(x)$ instead of the prefix version $\KP(x)$. As we have mentioned, the difference between them is only logarithmic, and we mainly ignore terms of that size.

\subsection{Enumerating strings of complexity at most $m$}

Consider some integer $m$, and all strings $x$ of (plain) complexity at most $m$. Let $\Omega_m$ be the number of those strings. The following properties of $\Omega_m$ are well known and often used (see, e.g., \cite{bienvenu-desfontaines-shen}).

\begin{proposition}\label{prop:omegas}
\leavevmode
\begin{itemize}
\item $\Omega_m=\Theta(2^m)$ (i.e., $c_12^m\le \Omega_m\le c_2 2^m$ for some positive constants $c_1,c_2$ and for all~$m$;
\item $\KS(\Omega_m)=m+O(1)$.
\end{itemize}
\end{proposition}

\begin{proof}
The number of strings of complexity at most $m$ is bounded by the total number of programs of length at most $m$, which is $O(2^m)$. On the other hand, if $\Omega_m$ is an $(m-d)$-bit number, we can specify a string of complexity greater than $m$ using $m-d+O(\log d)$ bits: first we specify $d$ in a self-delimiting manner using $O(\log d)$ bits, and then append $\Omega_m$ in binary. This information allows us to reconstruct $d$, then $m$ and $\Omega_m$, then enumerate strings of complexity at most $m$ until we have $\Omega_m$ of them (so all strings of complexity at most $m$ are enumerated), and then take the first string $x_m$ that has not been enumerated. As $m<\KS(x_m)\le m-d+O(\log d)$, the value of $d$ is bounded by a constant and hence $\Omega_m$ is an $(m-O(1))$-bit number.

In this argument the binary representation of $\Omega_m$ can be replaced by its program, so $\KS(\Omega_m)\ge m-O(1)$. The upper bound $m+O(1)$ is obvious, since $\Omega_m=O(2^m)$.
\end{proof}

Given $m$, we can enumerate all strings of complexity at most $m$. How many steps needs the enumeration algorithm to produce all of them? The answer is provided by the so-called \emph{busy beaver numbers}; let us recall their definition in terms of Kolmogorov complexity (see~\cite[section 1.2.2]{usv} for details).

By definition, the number $B(m)$ is the maximal integer of complexity at most $m$. It is not hard to see that $\KS(B(m))= m+O(1)$. Indeed, $\KS(B(m))\le m$ by definition. On the other hand, the complexity of the next number $B(m)+1$ is greater than $m$ and at the same time is bounded by $\KS(B(m))+O(1)$.

Note that $B(m)$ can be undefined for small $m$ (if there are no integers of complexity at most $m$) and that $B(m+1)\ge B(m)$ for all $m$. For some $m$ this inequality may not be strict. This happen, for example, if the optimal algorithm used to define Kolmogorov complexity is defined only on strings of, say, even lengths; this restriction does not prevent it from being optimal, but then $B(2n)=B(2n+1)$ for all $n$, since there are no objects of complexity exactly $2n+1$. However, for some  constant $c$ we have $B(m+c)>B(m)$ for all $m$. Indeed, consider a program $p$ of length at most $m$ that prints  $B(m)$. Transform it to a program $p'$ that runs $p$ and then adds $1$ to the result. This program witnesses that $\KS(B(m)+1)\le m+c$ for some constant $c$. Hence $B(m+c)\ge B(m)+1$.

Now we define $B'(m)$ as follows. As we have said, the set of all strings of complexity at most $m$ can be enumerated given $m$. Fix some enumeration algorithm $A$ (with input $m$) and some computation model. Then let $B'(m)$ be the number of steps used by this algorithm to enumerate all the strings of complexity at most $m$.

\begin{proposition}\label{prop:busy-beavers}
The numbers $B(m)$ and $B'(m)$ coincide up to $O(1)$-change in $m$. More precisely, we have
$$
B'(m)\le B(m+c), \qquad
B(m)\le B'(m+c)
$$
for some $c$ and for all $m$.
\end{proposition}

\begin{proof}
To find $B'(m)$, it is enough to know $m$-bit binary string that represents $\Omega_m$ (this string also determines $m$). Therefore $\KS(B'(m))\le m+c$ for some constant $c$. As $B(m+c)$ is the largest number of complexity $m+c$ or less, we have $B'(m)\le B(m+c)$.

On the other hand, if some integer $N$ exceeding both $m$ and $B'(m)$ is given, we can run the enumeration algorithm $A$ within $N$ steps for each input smaller than $N$. Consider the first string that has not been enumerated. Its complexity is greater than $m$, so $\KS(N)>m-c$ for some constant $c$. Thus the complexity of every number $N$ starting from $\max\{m,B'(m)\}$ is greater than $m-c$, which means that $\max\{m,B'(m)\}>B(m-c)$. It remains to note that for all large enough $m$ we have $m\le B(m-c)$, as the complexity of $m$ is $O(\log m)$. Thus for all large enough $m$ the number $B'(m)$ (and not $m$) must be bigger than $B(m-c)$. Replacing here $m$ by $m+c$ and increasing the constant $c$ if needed, we conclude that $B'(m+c)>B(m)$ for all $m$.
\end{proof}

A similar argument shows that $B(n)$ coincides (up to $O(1)$-change in the argument) with the maximal computation time of the universal decompressor (from the definition of plain Kolmogorov complexity) on inputs of size at most $m$, see~\cite[section~1.2.2]{usv}

The next result says how many strings require long time to be enumerated.

\begin{proposition}\label{prop:enumeration-tail}
After $B'(m-s)$ steps of the enumeration algorithm on input $m$ there are $2^{s+O(\log m)}$ strings that are not yet enumerated.
\end{proposition}

We assume  that the algorithm enumerates strings (for every input $m$) without repetitions. Note also that here $B'$ can be replaced by $B$, since they differ at most by a constant change in the argument.

\begin{proof} To make the notation simpler
we omit $O(1)$- and $O(\log m)$-terms in this argument. Given $\Omega_{m-s}$, we can determine $B'(m-s)$. If we also know how many strings of complexity at most $m$ appear after $B'(m-s)$ steps, we can wait until that many strings appear and then find a string of complexity greater than $m$. If the number of remaining strings is smaller than $2^{s-O(\log m)}$, we get a prohibitively short description of this high complexity string.

On the other hand, let $x$ be the last element that has been enumerated in $B'(m-s)$ steps. If there are significantly more than $2^s$ elements after $x$, say, at least $2^{s+d}$ for some $d$, we can split the enumeration in portions of size $2^{s+d}$ and wait until the portion containing $x$ appears. By assumption this portion is full. The number $N$ of steps needed to finish this portion is at least $B'(m-s)$ . This number $N$ and its successor $N+1$ can be reconstructed from the portion number that contains about $m-s-d$ bits. Thus the complexity of $N+1$ is at most $m-s-d+O(\log m)$. Hence we have
$$
B(m-s-d+O(\log m)) > N \ge B'(m-s).
$$
By Proposition~\ref{prop:busy-beavers}
we can replace $B'$ by $B$ here:
$$
B(m-s-d+O(\log m))> B(m-s).
$$
(with some other constant in $O$-notation). Since $B$ is a non-decreasing function, we get $d=O(\log m)$.
\end{proof}

\subsection{$\Omega$-like numbers}

G.~Chaitin introduced the ``Chaitin $\Omega$-number''
$
\Omega=\sum_k \mm(k);
$
it can also be defined as the probability of termination if the optimal prefix decompressor is applied to a random bit sequence~(see~\cite[section 5.7]{usv}).%
\footnote{This number depends on the choice of the prefix decompressor, so it is not a specific number but a class of numbers. The elements of this class can be equivalently characterized as random lower semicomputable reals in $[0,1]$, see~\cite[section~5.7]{usv}.}
The numbers $\Omega_n$ are finite versions of Chaitin's $\Omega$-number. The information contained in $\Omega_n$ increases as $n$ increases; moreover, the following proposition is true. In this proposition we consider $\Omega_n$ as a bit string (of length $n+O(1)$) identifying the number $\Omega_n$ and its binary representation.

\begin{proposition}\label{prop:omega-equivalence}
Assume that $k\le m$. Consider the string $(\Omega_m)_k$ consisting of the first $k$ bits of $\Omega_m$. It is $O(\log m)$-equivalent to $\Omega_k$: both conditional complexities $\KS(\Omega_k\cnd(\Omega_m)_k)$ and $\KS((\Omega_m)_k\cnd\Omega_k)$ are $O(\log m)$.
\end{proposition}

\begin{proof}
This is essentially the reformulation of the previous statement (Proposition~\ref{prop:enumeration-tail}).

Run the algorithm that enumerates strings of complexity at most $m$. Knowing $(\Omega_m)_k$, we can wait until less than $2^{m-k}$ strings are left in the enumeration of strings of complexity at most $m$; we know that this happens after more than $B(k)$ steps, and in this time we can enumerate all strings of complexity at most $k$ and compute $\Omega_k$. (In this argument we ignore $O(\log m)$-terms, as usual.)

Now the second inequality follows by the symmetry of information property. Indeed, since $\KS(\Omega_k)=k+O(1)$ and $\KS((\Omega_m)_k)\le k+O(1)$, the inequality $\KS(\Omega_k\cnd(\Omega_m)_k)=O(\log m)$ implies the inequality $\KS((\Omega_m)_k\cnd\Omega_k)=O(\log m)$.

A direct argument is also easy.  Knowing $\Omega_k$ and $k$, we can find the list of all the strings of complexity at most $k$ and the number $B'(k)$. Then we make $B'(k)$ steps in the enumeration of the list of strings of complexity at most $m$. Proposition~\ref{prop:enumeration-tail} then guarantees that at that moment $\Omega_m$ is known with error about $2^{m-k}$, so the first $k$ bits of $\Omega_m$ can be reconstructed with small advice (of logarithmic size; we omit terms of that size in the argument).
\end{proof}

There is a more direct connection with Chaitin's $\Omega$-number: one can show that the number $\Omega_m$ is $O(\log m)$-equivalent to the $m$-bit prefix of Chaitin's $\Omega$-number. Since in this survey we restrict ourselves to finite objects, we do not go into details of the proof here, see~\cite[section 5.7.7]{usv}.

\subsection{Position in the list is well defined}

We discussed how much time is needed to enumerate all strings of complexity at most $m$ and how many strings remain not enumerated before this time. Now we want to study \emph{which} strings remain not enumerated.

More precisely, let $x$ be some string of complexity at most $m$, so $x$ appears in the enumeration of all strings of complexity at most $m$. How close $x$ is to the end, that is, how many strings are enumerated after $x$? The answer depends on the enumeration, but only slightly, as the following proposition shows.

\begin{proposition}\label{prop:pos-def}
 Let $A$ and $B$ be algorithms that both for any given $m$ enumerate (without repetitions) the set of strings of complexity at most $m$. Let $x$ be some string and let $a_x$ and $b_x$ the number of strings that appear after $x$ in $A$- and $B$-enumerations. Then $|\log a_x - \log b_x|=O(\log m)$.
\end{proposition}

We may also assume that $A$ and $B$ are algorithms of complexity $O(\log m)$ without input that enumerate strings of complexity at most $m$.

\begin{proof}
Assume that $a_x$ is small: $\log a_x \le k$. Why $\log b_x$ cannot be much larger than $k$? Given the first $m-\log b_x$ bits of $\Omega_m$ and $B$, we can compute a finite set of strings $B'$ that contains $x$ and consists only of strings of complexity at most $m$. Then we can wait until all strings from $B'$ appear in $A$-enumeration. After then at most $2^k$ strings are left, and we need $k$ bits to count them. In this way we can describe $\Omega_m$ by $m-\log b_x+k+O(\log m)$ bits; however, Proposition~\ref{prop:omegas} says that $\KS(\Omega_m)=m+O(1)$. Hence
$\log b_x \le k+O(\log m)$.

The other inequality is proven by a symmetric argument.
\end{proof}

In this theorem $A$ and $B$ enumerate exactly the same strings (though in different order). However, the complexity function is essentially defined with $O(1)$-precision only: different optimal programming languages lead to different versions. Let $\KS$ and $\tilde{\KS}$ be two (plain) complexity functions; then $\tilde{\KS}(x)\le\KS(x)+c$ for some $c$ and for all $x$. Then the list of all $x$ with $\KS(x)\le m$ is contained in the list of all $x$ with $\tilde{\KS}(x)\le m+c$. The same argument shows that the number of elements after $x$ in the first list cannot be much larger than the number of elements after $x$ in the second list. The reverse inequality is not guaranteed, however, even for the same version of complexity (small increase in the complexity bound may significantly increase the number of strings after $x$ in the list). We will return to this question in Section~\ref{subsec:rel-px}, but let us note first that some increase is guaranteed.

\begin{proposition}\label{prop:tail-monotonicity}
If for a string $x$ there are at least $2^s$ elements after $x$ in the enumeration of all strings of complexity at most $m$, then for every $d\ge 0$ there are at least $2^{s+d-O(\log m)}$ strings after $x$ in the enumeration of all strings of complexity at most $m+d$.
\end{proposition}

\begin{proof}
Essentially the same argument works here: if there are much less than $2^{s+d}$ strings after $x$ in the bigger list, then this bigger list can be determined by $2^{m-s}$ bits needed to cover $x$ in the smaller list and less than $s+d$ bits needed to count the elements in the bigger list that follow the last covered element.
\end{proof}

The last proposition can be restated in the following way. Let us fix some complexity function and and some algorithm that, given $m$, enumerates all strings of complexity at most $m$. Then, for a given string $x$, consider the function that maps every $m\ge \KS(x)$ to the logarithm of the number of strings after $x$ in the enumeration with input $m$. Proposition~\ref{prop:tail-monotonicity} says that $d$-increase in the argument leads at least to $(d-O(\log m))$-increase of this function (but the latter increase could be much bigger). As we will see, this function is closely related to the set $P_x$ (and therefore $Q_x$): it is one more representation of the same boundary curve.

\subsection{The relation to $P_x$}\label{subsec:rel-px}

To explain the relation, consider the following procedure for a given binary string $x$. For every $m\ge \KS(x)$ draw the line $i+j=m$  on $(i,j)$-plane. Then draw the point on this line with second coordinate $s$ where $s$ is the logarithm of the number of elements after $x$ in the enumeration of all strings of complexity at most $m$. Mark also all points on this line on the right of (=below) this point. Doing this for different $m$, we get a set (Figure~\ref{mdl-e-12}).
\begin{figure}[h]
\begin{center}
\includegraphics[scale=1]{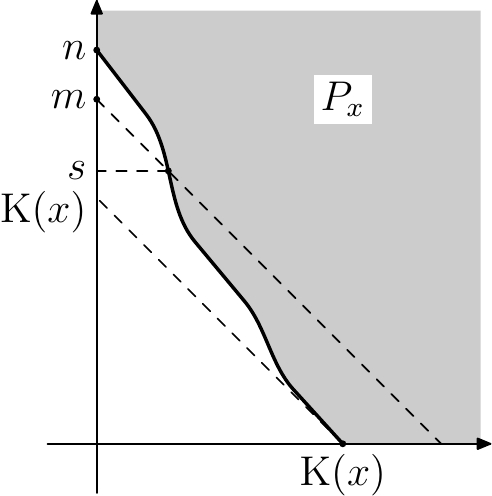}
\caption{For each $m$ between $\KP(x)$ and $n$ (length of $x$) we count elements after $x$ in the list of strings having complexity at most $m$; assuming there is about $2^s$ of them, we draw point $(m-s,s)$ and get a point on some curve. This curve turns out to be the boundary of $P_x$ (with logarithmic precision).}
\label{mdl-e-12}
\end{center}
\end{figure}
Proposition~\ref{prop:tail-monotonicity} guarantees that this set is upward closed with logarithmic precision: if some point $(i,j)$ belongs to this set, then the point $(i,j+d)$ is in $O(\log(i+j))$-neighborhood of this set. This implies that the point $(i+d,j)$ is also in the neighborhood, since our set is closed by construction in the direction $(1,-1)$.

It turns out  that this set coincides with $P_x$ \textup(Definition~\ref{def:px}\textup) with $O(\log n)$-precision for a string $x$ of length $n$ (this means, as usual, that each of the two sets is contained in the $O(\log n)$-neighborhood of the other one):

\begin{thm}\label{thm:tail-characterization}
Let $x$ be a string of length $n$. If $x$ has a $(i*j)$-description then $x$ is at least $2^{j-O(\log n)}$-far from the end of $(i+j+O(\log n))$-list. Conversely, if there are at least $2^j$ elements that follow $x$ in the $(i+j)$-list then $x$ has a $((i+O(\log n))*j)$-description.
\end{thm}

\begin{proof}
We need to verify two things. First, assuming that $x$ has a $(i*j)$-description, we need to show that it is at least $2^j$-far from the end of $(i+j)$-list. (With error terms: in $(i+j+O(\log n))$-list there are at least $2^{j-O(\log n)}$ elements after $x$.) Indeed, knowing some $(i*j)$-description $A$ for $x$, we can wait until all the elements of $A$ appear in $(i+j)$-list (as usual, we omit $O(\log n)$-term: all elements of $A$ have complexity at most $i+j+O(\log n)$, so we should consider $(i+j+O(\log n))$-list to be sure that it contains all elements of $A$). In particular, $x$ has appeared at that moment. If there are (significantly) less than $2^j$ elements after $x$, then we can encode the number of remaining elements by (significantly) less than $j$ bits, and together with the description of $A$ we get less than $i+j$ bits to describe $\Omega_{i+j}$, which is impossible.

Second, assume that there are at least $2^j$ elements that follow $x$ in the $(i+j)$-list. Then, splitting this list into $2^j$-portions, we get at most $2^i$ full portions, and $x$ is covered by one of them. Each portion has complexity at most $i$ and log-size at most $j$, so we get an $(i*j)$-description for $x$. (As usual, logarithmic terms are omitted.)
\end{proof}

Now we can reformulate the properties of stochastic and antistochastic objects. Every object of complexity $k$ appears in the list of objects of complexity at most $k'$ for all $k'>k$. Each stochastic object is far from the end of these lists (except, may be, for some $k'$-lists with $k'$ very close to $k$). Each antistochastic object of length $n$ is maximally close to the end of all $k'$-lists with $k'<n$ (there are about $2^{k'-k}$ objects after $x$), except, may be, for some $k'$-lists with $k'$ very close to $n$. When $k'$ becomes greater than $n$, then even antistochastic strings are far from the end of the $k'$-list. What we have  said is just the description of the corresponding curves (Figure~\ref{mdl.2}) using Theorem~\ref{thm:tail-characterization}.

\subsection{Standard descriptions}\label{subsec:std-descriptions}

The lists of objects of bounded complexity provide a natural class of descriptions. Consider some $m$ and the number $\Omega_m$ of strings of complexity at most $m$. This number can be represented in binary:
   $$
\Omega_m = 2^a + 2^b +\ldots,
   $$
where $a>b>\ldots$. The list itself then can be split into pieces of size $2^a$, $2^b$,\ldots, and these pieces can be considered as description of corresponding objects. In this way for each string $x$ and for each $m\ge\KS(x)$ we get some description on $x$, a piece than contains $x$. Descriptions obtained in this way will be called \emph{standard} descriptions. Note that for a given $x$ we have many standard descriptions (depending on the choice of $m$). One should have in mind also that the class of standard descriptions depends on the choice of the complexity function and the enumeration algorithm, and we assume in the sequel that they are fixed.

The following results show that standard descriptions are in a sense universal. First let us note that the standard descriptions have parameters close to the boundary curve of $P_x$ (more precisely, to the boundary curve of the set constructed in the previous section that is close to $P_x$).\footnote{In general, if two sets $X$ and $Y$  in $\mathbb{N}^2$ are close to each other (each is contained in the small neighborhood of the other one), this does not imply that their boundaries are close. It may happen that one set has a small ``hole'' and the other does not, so the boundary of the first set has points that are far from the boundary of the second one. However, in our case both sets are closed by construction in two different directions, and this implies that the boundaries are also close.}

\begin{proposition}\label{prop:std-pos}
Consider the standard description $A$ of size $2^j$ obtained from the list of all strings of complexity at most $m$. Then $\KS(A)=m-j+O(\log m)$, and the number of elements in the list that follow the elements of $A$ is $2^{j+O(\log m)}$.
\end{proposition}

This statement says that parameters of $A$ are close to the point on the line $i+j=m$ considered in the previous section (Figure~\ref{mdl-e-12}).

\begin{proof}
To specify $A$, it is enough to know the first $m-j$ bits of $\Omega_m$ (and $m$ itself). The complexity of $A$ cannot be much smaller, since knowing $A$ and the $j$ least significant bits of $\Omega_m$ we can reconstruct $\Omega_m$.

The number of elements that follow $A$ cannot exceed $2^j$ (it is a sum of smaller powers of $2$); it cannot be significantly less since it determines $\Omega_m$ together with the first $m-j$ bits of $\Omega$. (In other words, since $\Omega_m$ is an incompressible string of length $m$, it cannot have more that $O(\log m)$ zeros in a row.)
\end{proof}

This result does \emph{not} imply that every point on the boundary of $P_x$ is close to parameters of some standard description. If some part of the boundary has slope $-1$, we cannot guarantee that there are standard descriptions along this part. For example, consider the list of strings of complexity at most $m$; the maximal complexity of strings in this list is $m-c$ for some $c=O(1)$; if we take first string of this complexity, there are $2^{m+O(1)}$ strings after it, so the corresponding point is close to the vertical axis, and due to Proposition~\ref{prop:tail-monotonicity} all other standard descriptions of $x$ are also close to the vertical axis.   However, descriptions with parameters close to arbitrary points on the boundary of $P_x$ can be obtained from standard descriptions by chopping them into smaller parts, as in Proposition~\ref{prop:description-shift}. In that shopping it is natural to use the order in which the strings were enumerated. In other words, chop the list of strings of complexity at most $m$ into portions of size $2^j$. Consider all the full  portions (of size exactly $2^j$) obtained in this way (they are parts of standard descriptions of bigger size). Descriptions obtained in this way are ``universal'' in the following sense: if a pair $(i,j)$ is on the boundary of $P_x$ then there is a set $A\ni x$ of this type of complexity $i+O(\log(i+j))$ and log-cardinality $j+O(\log(i+j))$.

\smallskip
\ver{
  The following result says more: for every description  $A$ for $x$ there is a ``better'' standard description that is simple given $A$.
  %(note that $d\ge 0$ in the following proposition and that optimality deficiency of $B$ does not exceed that of $A$ up to logarithmic term).

\begin{proposition}\label{prop:better-std}
Let $A$ be an $(i*j)$-description of a string $x$ of length $n$. Then there is 
$$
m\le \min\{n,i+j\}+O(\log n)
$$ 
such that the parameters of the standard description $B$ 
for $x$ obtained from the list of strings of complexity at most 
$m$ satisfy the inequalities 
$$
\KS(B)\le i+O(\log m),\qquad \KS(B)+\log\#B =m+O(\log m).%\le i+j+O(\log m).
$$
Moreover, 
$B$ is simple given $A$, i.e., $\KS(B\cnd A)=O(\log m)$. 
\end{proposition}

\begin{proof}
If $i+j\le n$ then the complexity of every element from $A$ is at most
$i+j+O(\log j)=\min\{n,i+j\}+O(\log n)$. Otherwise remove from
$A$ all strings of length different from $n$. In this way $A$ becomes $(i*j)$-description for $x$ with slightly larger $i$ than before the removal and the same or smaller $j$. Now all the elements of $A$ have complexity at most 
$n+O(1)=\min\{n,i+j\}+O(1)$. Thus w.l.o.g. we may assume that the
complexity of all strings from $A$ does not exceed some 
$m=\min\{n,i+j\}+O(\log n)$.

Consider the list of all strings of complexity at most $m$ and the standard description $B$ of $x$ obtained from this list. As we know from Proposition~\ref{prop:std-pos}, the sum of the parameters of this description is 
$m+O(\log m)$. 
We need to show that the size of $B$ is at least $2^{m-i-O(\log m)}$ and hence
the complexity of $B$ is at most $i+O(\log m)$. Why is this the case? Consider elements that appear after the last element of $A$ in the list. There are at least $2^{m-i-O(\log m)}$ of them, otherwise the total number of elements in the list could be described in much less than $m$ bits. 
%Indeed,
%t%hat number can be specified by $m,A$ and the number $N$ of elements after the last element of $A$,
%therefore
%$$
%m\le \K(A)+\log N\le i+\log N
%$$ 
%(with accuracy $O(\log m)$). 
Therefore there are at least $2^{m-i-O(\log m)}$ elements in the list that appear after $x$.
As $x\in B$ the number of elements that appear after $x$ is less than $2\# B$ therefore $\# B\ge 2^{m-i-O(\log m)}$ and $\K(B)\le i+O(\log m)$ .

Why $B$ is simple given $A$? Denote the size of $B$ by $2^{j'}$. Given $A$ and $m$, we can find the last element of $A$, call it $x'$, in the list of strings of complexity at most $m$. Chop the list into portions of size $2^{j'}$. Then $B$ is the last complete portion. If $B$ contains $x'$, we can find $B$ from $m$, $j'$, and $x'$ as the complete portion containing $x'$. Otherwise, $x'$ appears in the list after all the elements from $B$. In this case we can find $B$ from $m$ and $x'$ as the last complete portion before $x'$.  Thus in any  case we are able to find $B$ from $m$, $j'$, and $x'$ plus one extra bit.
\end{proof}
}

For the same reason every standard description $B$ of some $x$ is simple given $x$ (and this is not a surprise, since we know that all optimal descriptions of $x$ are simple given $x$, see Proposition~\ref{prop:improving-descriptions}).

Proposition~\ref{prop:better-std} has the following corollary which we formulate in an informal way. Let $A$ be some $(i*j)$-description with parameters on the boundary of $P_x$. Assume that on the left of this point the boundary curve decreases fast (with slope less than $-1$). Then in Proposition~\ref{prop:better-std} the value of $d$ is small, otherwise the point $(i-d,j+d)$ would be far from $P_x$. So the complexities of $A$ and the standard description $B$ are close to each other. We know also that $A$ is simple given $B$, therefore $B$ is also simple given $A$, and $A$ and $B$ have the same information (have small conditional complexities in both directions).

If \label{discussion-minimal}
we have two different descriptions $A,A'$ with approximately the same parameters on the boundary of $P_x$, and the curve decreases fast on the left of the corresponding boundary point, the same argument shows that $A$ and $A'$ have the same information. Note that the condition about the slope is important: if the point is on the segment with slope $-1$, the situation changes. For example, consider a random $n$-bit string $x$ and two its descriptions. The first one consists of all $n$-bit strings that have the same left half as $x$, the second one consists of all $n$-bit strings that have the same right half. Both have the same parameters: complexity $n/2$ and log-size $n/2$, so they both correspond to the same point on the boundary of $P_x$. Still the information in these two descriptions is different  (left and right halves of a random string are independent).

These results sound as good news. Let us recall our original goal: to formalize what is a good statistical model. It seems that we are making some progress. Indeed, for a given $x$ we consider the boundary curve $P_x$ and look at the place when it first touches the lower bound $i+j=\KS(x)$; after that it stays near this bound. In other terms, we consider models with negligible optimality deficiency, and select among them the model with minimal complexity. Giving a formal definitions, we need to fix some threshold $\varepsilon$. Then we say that a set $A$ is a \emph{$\varepsilon$-sufficient statistic} if $\delta(x, A)<\varepsilon$, and may choose the simplest one among them and call it the \emph{minimal $\varepsilon$-sufficient statistic}. If the curve goes down fast on the left of this point, we see that all the descriptions with parameters corresponding to minimal sufficient statistic are equivalent to each other.

Trying to relate these notion to practice, we may consider the following example. Imagine that we have digitized some very old recording and got some bit string $x$. There is a lot of dust and scratches on the recording, so the originally recorded signal is distorted by some random noise. Then our string $x$ has a two-part description: the first part specifies the original recording and the noise parameters (intensity, spectrum, etc.) and the second part specifies the noise exactly. May be, the first part is the minimal sufficient statistic --- and therefore sound restoration (and lossy compression in general) is a special case of the problem of finding a minimal sufficient statistic? The uniqueness result above (saying that all the minimal sufficient statistics contain the same information under some conditions) seem to support this view: different good models for the same object contain the same explanation.

Still the following observation (that easily follows from what we know) destroys this impression completely.

\begin{proposition}\label{prop:std-omega}
Let $B$ be some standard description of complexity $i$ obtained from the list of all strings of complexity at most $m$. Then $B$ is $O(\log m)$-equivalent to $\Omega_i$.
\end{proposition}

This looks like a failure. Imagine that we wanted to understand the nature of some data string $x$; finally we succeed  and find a description for $x$ of reasonable complexity and negligible randomness and optimality deficiencies (and all the good properties we dreamed of). But Proposition~\ref{prop:std-omega} says that the information contained in this description is more related to the computability theory than to specific properties of $x$. Recalling the construction, we see that the corresponding standard description is determined by some prefix of some $\Omega$-number, and is an interval in the enumeration of objects of bounded complexity. So if we start with two old recordings, we may get the same information, which is not what we expect from a restoration procedure. Of course, there is still a chance that some $\Omega$-number was recorded and therefore the restoration process indeed should provide the information about it, but this looks like a very special case that hardly should happen for any practical situation.

What could we do with this? First, we could just relax and be satisfied that we now understand much better the situation with possible descriptions for $x$. We know that every $x$ is characterized by some curve that has several equivalent definitions (in terms of stochasticity, randomness deficiency, position in the enumeration --- as well as time-bounded complexity, see Section~\ref{sec:depth} below). We know that standard descriptions cover the parts of the curve where it goes down fast, and to cover the parts where the slope is $-1$ one may use standard descriptions and their pieces; all these descriptions are simple given $x$. When curve goes down fast, the description is essentially unique (all the descriptions with the same parameters contain the same information, equivalent to the corresponding $\Omega$-number); this is not true on parts with slope $-1$. So, even if this curve is of no philosophical importance, we have a lot of technical information about possible models.

The other approach is to go farther and consider only models from some class (Section~\ref{sec:restricted-type}), or add some additional conditions and look for ``strong models'' (Section~\ref{sec:strong-models}).

\subsection{Non-stochastic objects revisited}\label{subsec:halting-information}

Now we can explain in a different way why the probability of obtaining a non-stochastic object in a random process is negligible (Proposition~\ref{prop:nonstochastic-counting}). This explanation uses the notion of mutual information from algorithmic information theory. The mutual information in two strings $x$ and $y$ is defined as
$$ I(x : y)=\KS(x)-\KS(x\cnd y) =\KS(y)-\KS(y\cnd x)=\KS(x)+\KS(y)-\KS(x,y);$$
all three expressions are $O(\log n)$-close if $x$ and $y$ are strings of length $n$ (see, e.g., \cite[Chapter 2]{usv}).

Consider an arbitrary string $x$ of length $n$; let $k$ be the complexity of $x$. Consider the list of all objects of complexity at most $k$, and the standard description $A$ for $x$ obtained from this list. If $A$ is large, then $x$ is stochastic; if $A$ is small, then $x$ contains a lot of information about $\Omega_k$ and $\Omega_n$.

More precisely, let us assume that $A$ has size $2^{k-s}$ (i.e., is $2^s$ times smaller than it could be). Then (recall Proposition~\ref{prop:std-pos}) the complexity of $A$ is $s+O(\log k)$, since we can construct $A$ knowing $k$ and the first $s$ bits of $\Omega_k$ (before the bit that corresponds to $A$). So we get $(s+O(\log k))*(k-s)$-description with optimality deficiency $O(\log k)$.

On the other hand, knowing $x$ and $k$, we can find the ordinal number of $x$ in the enumeration, so we know $\Omega_k$ with error at most $2^{k-s}$, so $\KS(\Omega_k\cnd x)\le k-s+O(\log k)$, and $I(x:\Omega_k)\ge s-O(\log k)$ (recall that $\KS(\Omega_k)=k+O(1)$). In the last statement we may replace $\Omega_k$ by $\Omega_n$ (where $n$ is the length of $x$): we know from Proposition~\ref{prop:omega-equivalence} that $\Omega_k$ is simple given $\Omega_n$, so if condition $\Omega_k$ decreases complexity of $x$ by almost $s$ bits, the same is true for condition $\Omega_n$.

Comparing arbitrary $i\le n$ with this $s$ (it can be larger than $s$ or smaller than $s$), we get the following result:

\begin{proposition}\label{prop:dilemma}
Let $x$ be a string of length $n$. For every $i\le n$
\begin{itemize}
\item either $x$ is $(i+O(\log n),O(\log n))$-stochastic,
\item or $I(x:\Omega_n)\ge i-O(\log n)$.
\end{itemize}
\end{proposition}

Now we may use the following (simple and general) observation: for every string $u$ the probability to generate (by a randomized algorithm) an object that contains a lot of information about $u$ is negligible:

\begin{proposition}\label{prop:information-rare}
For every string $u$ and for every number $d$, we have
$$
\sum\{ \mm(x)\mid \KP(x)-\KP(x\cnd u)\ge d\} \le 2^{-d}.
$$
\end{proposition}
In this proposition the sum is taken over all strings $x$ that have the given property (have a large mutual information with $u$). Note that we have chosen the representation of mutual information that makes the proposition easy (in particular, we have used prefix complexity). As we mentioned, other definitions differ only by $O(\log n)$ if we consider strings $x$ and $u$ of length at most $n$, and logarithmic accuracy  is enough for our purposes.

\begin{proof}
Recall the definition of prefix complexity: $\KP(x)=-\log \mm(x)$, and $\KP(x\cnd u)=-\log \mm(x\cnd u)$. So $\KP(x)-\KP(x\cnd u)\ge d$ implies $\mm(x)\le 2^{-d}\mm(x\cnd u)$, and it remains to note that $\sum_x\mm(x\cnd u)\le 1$ for every $u$.
\end{proof}

Propositions~\ref{prop:dilemma} and~\ref{prop:information-rare} immediately imply the following improved version of Proposition~\ref{prop:nonstochastic-counting} (page~\pageref{prop:nonstochastic-counting}):

\begin{proposition}\label{prop:nonstochastic-counting-improved}
$$
\sum \{\,\mm(x)\mid \text{$x$ is a $n$-bit string that is not $(\alpha,O(\log n))$-stochastic}\ \} \le 2^{-\alpha+O(\log n)}
$$
for every $\alpha$.
\end{proposition}
The improvement here is the better upper bound for the randomness deficiency: $O(\log n)$ instead of $\alpha+O(\log n)$.

\subsection{Historical comments}

The relation between busy beaver numbers and Kolmogorov complexity was pointed out in~\cite{gacs1984} (see Section 2.1). The enumerations of all objects of bounded complexity and their relation to stochasticity were studied in~\cite{gtv} (see Section III, E).

\section{Computational and logical depth}\label{sec:depth}

In this section we reformulate the results of the previous one in terms of bounded-time Kolmogorov complexity and discuss the various notions of computational and logical depth that appeared in the literature. (The impatient reader may skip this section; it is not technically used in the sequel).

\subsection{Bounded-time Kolmogorov complexity}

The usual definition of Kolmogorov complexity of $x$ as the minimal length $l(p)$ of a program $p$ that produces $x$ does not take into account the running time of the program $p$: it may happen that the minimal program for $x$ requires a lot of time to produce $x$ while other programs produce $x$ faster but are longer (for example,  program ``print $x$'' is rather fast). To analyze this trade-off, the following definition is used.

\begin{definition}
Let $D$ be some algorithm; its input and output are binary strings. For a string $x$ and integer $t$, define
$$
\KS^t_D=\min\{l(p)\colon \text{$D$ produces $x$ on input $p$ in at most $t$ steps}\},
$$
the time-bounded Kolmogorov complexity of $x$ with time bound $t$ with respect to $D$.
\end{definition}
This definition was mentioned already in the first paper by Kolmogorov~\cite{kolm65}:
\begin{quote}
Our approach has one important drawback: it does not take into account the efforts needed to transform the program $p$ and object $x$ [the description and the condition] to the object $y$ [whose complexity is defined]. With appropriate definitions, one may prove mathematical results that could be interpreted as the existence of an object $x$ that has simple programs (has very small complexity $K(x)$) but all short programs that produce $x$ require an unrealistically long computation. In another paper I plan to study the dependence of the program complexity $K^t(x)$ on the difficulty $t$ of its transformation into~$x$. Then the complexity $K(x)$ (as defined earlier) reappears as the minimum value of $K^t(x)$ if we remove restrictions on~$t$.
\end{quote}
Kolmogorov never published a paper he speaks about, and this definition is less studied than the definition without time bounds, for several reasons.

First, the definition is machine-dependent: we need to decide what computation model is used to count the number of steps. For example, we may consider one-tape Turing machines, or multi-tape Turing machine, or some other computational model. The computation time depends on this choice, though not drastically (e.g., a multi-tape machine can be replaced with a one-tape machine with quadratic increase in time, and most popular models are polynomially related --- this observation is used when we argue that the class P of polynomial-time computable functions is well defined).

Second, the basic result that makes the Kolmogorov complexity theory possible is the Solomonoff--Kolmogorov theorem saying that there exists an optimal algorithm $D$ that makes the complexity function minimal up to $O(1)$ additive term. Now we need to take into account the time bound, and get the following (not so nice) result.

\begin{proposition}\label{prop:time-bounded-optimal}
There exists an optimal algorithm $D$ for time-bounded complexity in the following sense: for every other algorithm $D'$ there exists a constant $c$ and a polynomial $q$ such that
$$
   \KS_{D'}^t(x)\le \KS_D^{q(t)}(x)+c
$$
for all strings $x$ and integers $t$.
\end{proposition}

In this result, by ``algorithm'' we may mean a $k$-tape Turing machine, where $k$ is an arbitrary fixed number. However, the claim remains true even when $k$ is not fixed, i.e., we may allow $D'$ to have more tapes than $D$ has.

The proof remains essentially the same: we choose some simple self-delimiting encoding of binary strings $p\mapsto\hat p$ and some universal algorithm $U(\cdot,\cdot)$ and then let
$$
D(\hat p x)=U(p,x)
$$
Then the proof follows the standard scheme; the only thing we need to note is that the decoding of $\hat p$ runs in polynomial time (which is true for most natural ways of self-delimiting encoding) and that the universal algorithm simulation overhead is polynomial (which is also true for most natural constructions of universal algorithms).

A similar result is true for conditional decompressors, so the conditional time-bounded complexity can be defined as well.

For Turing machines with fixed number of tapes the statement is true for some linear polynomial $q(n)=O(n)$. For the proof we need to consider a universal machine $U$ that simulates other machines efficiently: it should move the program along the tape, so the overhead is bounded by a factor that depends on the size of the program and not on the size of the input or computation time.\footnote{%
This observation motivates Levin's version of complexity ($Kt$, see~\cite[Section 1.3, p.~21]{levin-conservation}) where the program size and logarithm of the computation time are added: linear overhead in computation time matches the constant overhead in the program size. However, this is a different approach and we do not use the Levin's notion of time bounded complexity in this survey.}

Let $t(n)$ be an arbitrary total computable function with integer arguments and values; then the function
$$
x \mapsto \KS^{t(l(x))}_{D} (x)
$$
is a computable upper bound for the complexity $\KS(x)$ (defined with the same $D$; recall that $l(x)$ stands for the length of $x$). Replacing the function $t(\cdot)$ by a bigger function, we get a smaller computable upper bound. An easy observation: in this way we can match every computable upper bound for Kolmogorov complexity.

\begin{proposition}\label{prop:any-upper-bound}
Let $\tilde{\KS}(x)$ be some total computable upper bound for Kolmogorov complexity function based on the optimal algorithm $D$ from Proposition~\ref{prop:time-bounded-optimal}. Then there exists a computable function $t$ such that $\KS^{t(l(x))}_D(x)\le \tilde{\KS}(x)$ for every $x$.
\end{proposition}

\begin{proof}
Given a number $n$, we wait until every string $x$ of length  at most $n$ gets a program that has complexity at most $\tilde{\KS}(x)$, and let $t(n)$ be the maximal number of steps used by these programs.
\end{proof}

So the choice of a computable time bound is essentially equivalent to the choice of a computable total upper bound for Kolmogorov complexity.

In the sequel we assume that some optimal (in the sense of Proposition~\ref{prop:time-bounded-optimal}) $D$ is fixed and omit the subscript $D$ in $\KS^t_D(\cdot)$. Similar notation $\KS^t(\cdot\cnd\cdot)$ is used for conditional time-bounded complexity.

\subsection{Trade-off between time and complexity}

We use the extremely fast growing sequence $B(0),B(1),\ldots$ as a scale for measuring time. This sequence grows faster than any computable function (since the complexity of $t(n)$ for any computable $t$ is at most $\log n+O(1)$, we have $B(\log n+O(1))\ge t(n)$). In this scale it does not matter whether we use time or space as the resource measure: they differ at most by an exponential function, and $2^{B(n)}\le B(n+O(1))$ (in general, $f(B(n))\le B(n+O(1))$ for every computable $f$). So we are in the realm of general computability theory even if we technically speak about computational complexity, and the problems related to the unsolved P=NP question disappear.

Let $x$ be a string of length $n$ and complexity $k$. Consider the time-bounded complexity $\KS^t(x)$ as a function of $t$. (The optimal algorithm from Proposition~\ref{prop:time-bounded-optimal} is fixed, so we do not mention it in the notation.) It is a decreasing function of $t$. For small values of $t$ the complexity $\KS^t(x)$ is bounded by $n+O(1)$ where $n$ stands for the length of $x$. Indeed, the program that prints $x$ has size $n+O(1)$ and works rather fast. Formally speaking, $\KS^t(x)\le n+O(1)$ for $t=B(O(\log n))$. As $t$ increases, the value of $\KS^t(x)$ decreases and reaches $k=\KS(x)$ as $t\to\infty$. It is guaranteed to happen for $t=B(k+O(1))$, since the computation time for the shortest program for $x$ is determined by this program.

We can draw a curve that reflects this trade-off using $B$-scale for the time axis. Namely, consider the graph of the function
$$
 i \mapsto \KS^{B(i)}(x)-\KS(x)
$$
and the set of points above this graph, i.e., the set
$$
D_x=\{(i,j)\mid \KS^{B(i)}(x)-\KS(x)\le j\}.
$$

\begin{thm}[\cite{bauwens,ABST}]\label{thm:depth}
The set $D_x$ coincides with the set $Q_x$ with $O(\log n)$-precision for a string $x$ of length~$n$.
\end{thm}

Recall that the set $Q_x$ consists of pairs $(\alpha,\beta)$ such that $x$ is $(\alpha,\beta)$-stochastic (see p.~\pageref{def:qx}).

\begin{proof}
As we know from Theorem~\ref{thm:def-opt}, the sets $P_x$ and $Q_x$ are related by an affine transformation (see Figure~\ref{mdl.8}). Taking this transformation into account, we need to prove two statements:
\begin{itemize}
\item if there exists an $(i*j)$-description $A$ for $x$, then $$\KS^{B(i+O(\log n))}(x)\le i+j+O(\log n);$$
\item if $\KS^{B(i)}(x)\le i+j$, then
$$\text{there exist an }((i+O(\log n))*(j+O(\log n)))\text{-description for $x$.}$$
\end{itemize}
Both statements are easy to prove using the tools from the previous section. Indeed, assume that $x$ has an $(i*j)$-description $A$. All elements of $A$ have complexity at most $i+j+O(\log n)$. Knowing $A$ and this complexity, we can find the minimal $t$ such that $C^t(x')\le i+j+O(\log n)$ for all $x'$ from $A$. This $t$ can be computed from $A$, which has complexity $i$, and an $O(\log n)$-bit advice (the value of complexity). Hence  $t\le B(i+O(\log n))$ and $C^t(x)\le i+j+O(\log n)$, as required.

The converse: assume that $C^{B(i)}(x)\le i+j$. Consider all the strings $x'$ that satisfy this inequality. There are at most $O(2^{i+j})$ such strings. Thus we only need to show that given $i$ and $j$ we are able to enumerate all those strings in at most $O(2^i)$ portions.

One can get a list of all those strings $x'$ if $B(i)$ is given, but we cannot compute $B(i)$ given~$i$. Recall that $B(i)$ is the maximal integer that has complexity at most $i$; new candidates for $B(i)$ may appear at most $2^i$ times. The candidates increase with time; when this happens, we get a new portion of strings that satisfy the inequality $C^{B(i)}(x)\le i+j$. So we have at most $O(2^{i+j})$ objects including $x$ that are enumerated in at most $2^i$ portions, and this implies that $x$ has an $((i+O(\log n))*j)$-description. Indeed, we make all portions of size at most $2^j$ by splitting larger portions into pieces. The number of portions increases at most by $O(2^i)$, so it remains $O(2^i)$. Each portion (including the one that contains $x$) has then complexity at most $i+O(\log n)$ since it can be computed with logarithmic advice from its ordinal number.
\end{proof}

This theorem shows that the results about the existence of non-stochastic objects can be considered as the ``mathematical results that could be interpreted as the existence of an object $x$ that has simple programs (has very small complexity $K(x)$) but all short programs that produce $x$ require an unrealistically long computation'' mentioned by Kolmogorov (see the quotation above), and the algorithmic statistics can be interpreted as an implementation of Kolmogorov's plan ``to study the dependence of the program complexity $K^t(x)$ on the difficulty $t$ of its transformation into~$x$'', at least for the simple case of (unrealistically) large values of $t$.

\subsection{Historical comments}

Section~\ref{sec:depth} has title ``logical and computational depth'' but we have not defined these notions yet.  The name ``logical depth'' was introduced by C.~Bennett in~\cite{bennett}. He explains the motivation as follows:
\begin{quote}
Some mathematical and natural objects (a random sequence, a sequence of zeros, a perfect crystal, a gas) are intuitively trivial, while others (e.g., the human body, the digits of $\pi$) contain internal evidence of a nontrivial causal history. $\langle\ldots\rangle$

We propose depth as a formal measure of value. From the earliest days of information theory it has been appreciated that information per se is not a good measure of message value. For example, a typical sequence of coin tosses has high information content but little value; an ephemeris, giving the positions of the moon and the planets every day for a hundred years, has no more information than the equations of motion and initial conditions from which it was calculated, but saves its owner the effort of recalculating these positions.  The value of a message thus appears to reside not in its information (its absolutely unpredictable parts), nor in its obvious redundancy (verbatim repetitions, unequal digit frequencies), but rather is what might be called its buried redundancy --- parts predictable only with difficulty, things the receiver could in principle have figured out without being told, but only at considerable cost in money, time, or computation. In other words, the value of a message is the amount of mathematical or other work plausibly done by its originator, which its receiver is saved from having to repeat.
\end{quote}
Trying to formalize this intuition, Bennett suggests the following possible definitions:
\begin{quote}
\textbf{Tentative Definition 0.1}: A string's depth might be defined as the execution time of its minimal program.
\end{quote}
This notion is not robust (it depends on the specific choice of the optimal machine used in the definition of complexity). So Bennett considers another version:
\begin{quote}
\textbf{Tentative Definition 0.2}: A string's depth at significance level $s$ [might] be defined as the time required to compute the string by a program no more than $s$ bits larger than the minimal program.
\end{quote}
   We see that Definition 0.2 consider the same trade-off as in Theorem~\ref{thm:depth}, but in reversed coordinates (time as a function of difference between time-bounded and limit complexities). Bennett is still not satisfied by this definition, for the following reason:
\begin{quote}
This proposed definition solves the stability problem, but is unsatisfactory in the way it treats multiple programs of the same length. Intuitively, $2^k$ distinct $(n+k)$-bit programs that compute same output ought to be accorded the same weight as one $n$-bit program $\langle\ldots\rangle$
\end{quote}
In other language, he suggests to consider a priori probability instead of complexity:
\begin{quote}
\textbf{Tentative Definition 0.3}: A string's depth at significance level $s$ might be defined as the time $t$ required for the string's time-bounded algorithmic probability $P_t(x)$ to rise to within a factor $2^{-s}$ of its asymptotic time-unbounded value $P(x)$.
\end{quote}
Here $P_t (x)$ is understood as a total weight of all self-delimiting programs that produce $x$ in time at most $t$ (each program of length $s$ has weight $2^{-s}$). For our case (when we consider busy beaver numbers as time scale) the exponential time increase needed to switch from a priori probability to prefix complexity does not matter. Still Bennett is interested in more reasonable time bounds (recall that in his informal explanation a polynomially computable sequence of $\pi$-digits was an example of a deep sequence!), and prefers a priori probability approach. Moreover, he finds a nice reformulation of this definition (almost equivalent one) in terms of complexity:

\begin{quote}
Although Definition 0.3 satisfactorily captures the informal notion of depth, we propose a slightly stronger definition for the technical reason that it appears to yield a stronger slow growth property $\langle\ldots\rangle$

\textbf{Definition 1} (Depth of Finite Strings): Let $x$ and $w$ be strings [probably $w$ is a typo: it is not mentioned later] and $s$ a significance parameter. A string's \emph{depth} at significance level $s$, denoted $D_s(x)$, will be defined as $$\min\{T(p)\colon (|p|-|p^*|<s) \land (U(p)=x)\},$$ the least time required to compute it by a $s$-incompressible program.
\end{quote}
Here $p^*$ is a shortest self-delimiting program for $p$, so its length $|p^*|$ equals $\KP(p)$.

Actually, this \emph{Definition 1} has a different underlying intuition than all the previous ones: a string $x$ is deep if \emph{all programs that compute $x$ in a reasonable time, are compressible}.  Note the before we required a different thing: that all programs that compute $x$ in a reasonable time are much longer than the minimal one. This is a weaker requirement: one may imagine a long incompressible program that computes $x$ fast. This intuition is explained in the abstract of the paper~\cite{bennett} as follows:
\begin{quote}
[We define] an object's ``logical depth'' as the time required by a standard universal Turing machine to generate it from an input that is algorithmically random.
\end{quote}
Bennett then proves a statement (called Lemma 3 in his paper) that shows that his \emph{Definition 1} is almost equivalent to \emph{Tentative Definition 0.3}: the time remains exactly the same, while $s$ changes at most logarithmically (in fact, at most by $\KP(s)$). So if we use Bennett's notion of depth (any of them, except for the first one mentioned) with busy beaver time scale, we get the same curve as in our definition.

A natural question arises: is there a direct proof that the output of an incompressible program with not too large running time is stochastic? In fact, yes, and one can prove a more general statement: the output of a  \emph{stochastic} program with reasonable running time is stochastic (see Section~\ref{subsec:depth-appl}); note that stochasticity is a weaker condition than incompressibility.

Let us mention also the notion of \emph{computational depth} introduced in~\cite{AFM}.  There are several versions mentioned in this paper;  the first one exchanges coordinates in the Bennett's tentative definition 0.2 (reproduced in~\cite{AFM} as Definition 2.5). The authors write: ``The first notion of computational depth we propose is the difference between a time-bounded Kolmogorov complexity and traditional Kolmogorov complexity'' (Definition 3.1, where time bound is some function of input length). The other notions of computation depth are more subtle (they use distinguishing complexity or Levin complexity involving the logarithm of the computation time).

The connections between computational/logical depth and sophistication were anticipated for a long time; for example, Koppel writes in~\cite{koppel}:
\begin{quote}
$\langle\ldots\rangle$ The ``dynamic'' approach to the formalization of meaningful complexity is ``depth'' defined and discussed by Bennett [1]. [Reference to an unpublished paper ``On the logical `depth' of sequences and their reducibilities to incompressible sequences''.] The depth of an object is the running-time of its most concise description. Since it is reasonable to assume that an object has been generated by its most concise description, the depth of an object can be thought of as a measure of its evolvedness.

Although sophistication is measured in integers [not clear what in meant here: sophistication of $S$ is also a function $c\mapsto SOPH_c(S)$] and depth is measured in functions, it is not difficult to translate to a common range.
\end{quote}

Strangely, the direct connection between the most basic versions of these notions (Theorem~\ref{thm:depth}) seems to be noticed only recently in~\cite[Section 3]{bauwens}, and \cite{ABST}.

\subsection{Why so many equivalent definitions?}\label{subsec:depth-appl}

We have shown several equivalent (with logarithmic precision and up to affine transformation) ways to defined the same curve:
\begin{itemize}
\item $(\alpha,\beta)$-stochasticity (Section~\ref{sec:stoch});
\item two-part descriptions and optimality deficiency, the
set $P_x$ (Section~\ref{sec:two-part});
\item position in the enumeration of objects of bounded complexity (Section~\ref{sec:bcl});
\item logical/computational depth (resource-bounded complexity, Section~\ref{sec:depth}).
\end{itemize}
One can add to this list a characterization in terms of split enumeration (Section~\ref{subsec:opt-rand}): the existence of $(i*j)$-description for $x$ is equivalent (with logarithmic precision) to the existence of a simple enumeration of at most $2^{i+j}$ objects in at most $2^{i}$ portions (see Remark~\ref{rem:portion}, p.~\pageref{rem:portion}, and the discussion before it).

Why do we need so many equivalent definitions of the same curve? First, this shows that this curve is really fundamental --- almost as fundamental characterization of an object $x$ as its complexity. As Koppel writes in~\cite{koppel87}, speaking about (some versions of) sophistication and depth:
\begin{quote}
One way of demonstrating the naturalness of a concept is by proving the equivalence of a variety of prime facie different formalizations $\langle\ldots\rangle$. It is hoped that the proof of the equivalence of two approaches to meaningful complexity, one using static resources (program size) and the other using dynamic resources (time), will demonstrate not only the naturalness of the concept but also the correctness of the specifications used in each formalization to ensure robustness and generality.
\end{quote}

\ver{
Another, more technical reason: different results about stochasticity use different equivalent definitions, and a statement that looks quite mysterious for one of them may become almost obvious for another. Let us give two examples of this type (the first one is stochasticity conservation when random noise is added, the second one is a direct proof of Bennett's characterization mentioned above). The first example is the following theorem 
from~\cite[Theorem 14]{ver2015} (though the proof there is different).

\begin{thm}\label{prop:add-noise}
  Let $x$ be some binary string, and let $y$ be another string (``noise'') that is conditionally random with respect to $x$, i.e., $\KS(y\cnd x)\approx l(y)$. Then the pair $(x,y)$ has the same stochasticity profile as $x$: the sets $Q_x$ and $Q_{(x,y)}$ are close to each other. More specifically, if
  $\KS(y\cnd x)\ge l(y)-\eps$
  then the sets $Q_x$ and $Q_{(x,y)}$ are in a
  $O(\log l(x)+\log l(y)+\eps)$-neighborhood of each other.
\end{thm}

%The full statement of Theorem~\ref{prop:add-noise} would introduce some bound for the difference $l(y)-\KS(y\cnd x)$ that is allowed to appear in the final estimate for the distance between sets. 
Recall that we can speak about profiles of arbitrary finite objects, in particular, pairs of strings, using some natural encoding (Section~\ref{subsec:stoch-cons}).
Before giving a proof sketch, let us make two remarks.

\begin{remark}
    By Theorem~\ref{thm:def-opt} the set $P_x$ can be obtained
  by a simple transformation from the set $Q_x$ and $\K(x)$, and the other
  way around. Thus Theorem~\ref{prop:add-noise} can be re-formulated in terms of the
  profiles $P_x$ and $P_{(x,y)}$. However the statement becomes more involved:

\begin{thm}[Theorem~\ref{prop:add-noise} in terms of profiles]~\label{rem:add-noise} 
  Let $x$ be a binary string
  %of complexity $k$,
  and let $y$ be another string such that $\KS(y\cnd x)\ge l(y)-\eps$. 
  Then the set  $P_{(x,y)}$ can be obtained from the set $P_x$ by the following
  transformation $\phi$:
  $$
  \phi(P_x)=\{(i,j+l(y))\mid i\le\K(x), (i,j)\in P_x\}\cup
  \{(i,j)\mid i>\K(x),\ i+j\ge \K(x,y)\}
  $$
More accurately, the sets $P_{(x,y)}$ and $\phi(P_x)$ are in  
an $O(\log l(x)+\log l(y)+\eps)$-neighborhood of each other.
\end{thm}
The transformation of $P_x$ to $P_{(x,y)}$
is shown on Fig.~\ref{ver-1}.
\begin{figure}[h]
\begin{center}\includegraphics{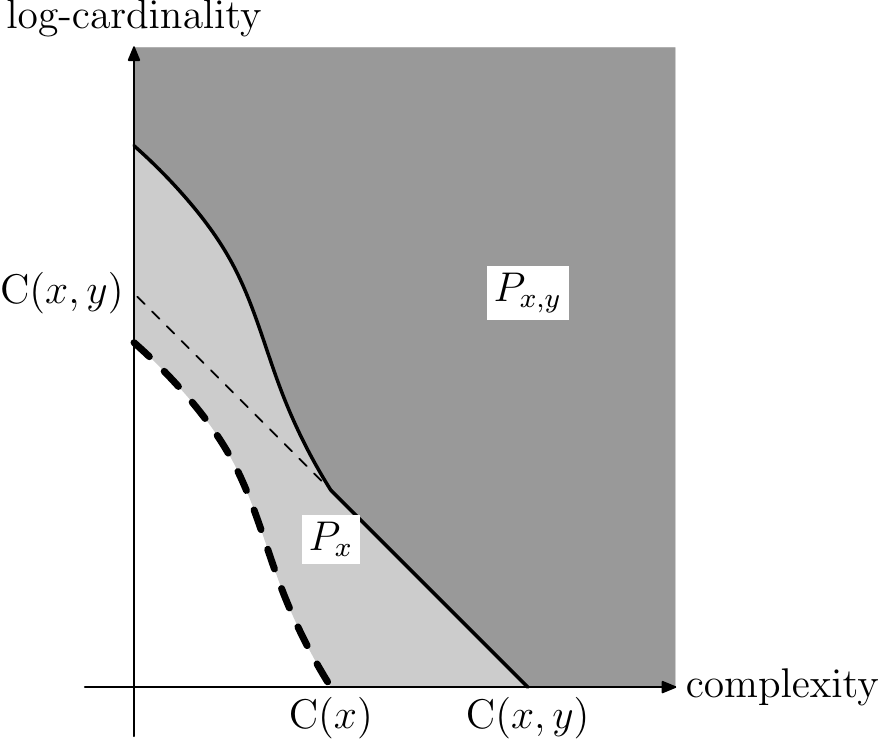}\end{center}
\caption{%The boundary of the set $P_x$ is plotted as the bold dotted line.
  \ver{The boundary of $P_{xy}$
  is obtained by shifting the boundary of $P_x$
  vertically by $l(y)\approx \K(x,y)-\K(x)$ and adding the sloping segment
  with the endpoints $(\K(x),\K(x,y)-\K(x))$ and
  $(\K(x,y),0)$.}}\label{ver-1}
\end{figure}
\end{remark}

\begin{remark}
  An interesting special case of this theorem is obtained if we consider a string $u$ and its description $X$ with small randomness deficiency: $d(u\cnd X)\approx 0$. Let $y$ be the ordinal number of $u$ in $X$. Then the small randomness deficiency guarantees that $y$ is conditionally random with respect to $X$. Therefore the pair $(X,y)$ has the same stochasticity profile as $u$. Since this pair is mapped to $u$ by a simple total computable function, we conclude (Proposition~\ref{prop:stoch-cons}) that the stochasticity profile of $X$ is contained in the stochasticity profile of $u$ (more precisely, in its $O(\log n + d(u\cnd X))$-neighborhood).
  For profiles, there is a more simple and direct proof of the
  inclusion of $\phi(P_X)$ into a small neighborhood of  $P_{u}$:
  if $\cal{U}$ is an $(i*j)$-description for $X$, we consider the ``lifting'' of
  $\mathcal U$, i.e. the union of all elements of $\cal{U}$ that have approximately the same cardinality as $X$; in this way we obtain a $(i*(j+l(y))$-description for $u$. This shows
  that the set
  $\{(i,j+l(y))\mid i\le\K(x), (i,j)\in P_x\}$
  is included in $P_u$.
  For the set  $\{(i,j)\mid i>\K(x),\ i+j\ge \K(x,y)\}$
  the inclusion is obvious, as for all $i>\K(x)$,
  $j\ge \K(x,y)-i$
  the set   of all strings in $X$ whose index has the same
  $l(y)-j$ leading
  bits, as $u$, is a $(i*j)$-description of $u$. 
\end{remark}
}

\smallskip

\begin{proof}[A proof sketch of Theorem~\ref{prop:add-noise}]
Using the depth characterization of stochasticity profile, we need to show that
$$
\KS^{B(i)}(x,y) -\KS(x,y) \approx \KS^{B(i)}(x)-\KS(x).
$$
Here ``approximately'' means that these two quantities may differ by a logarithmic term, and also we are allowed to add logarithmic terms to $i$ (see below what does it mean). The natural idea is to rewrite this equality as
$$
\KS^{B(i)}(x,y)-\KS^{B(i)}(x) \approx \KS(x,y)-\KS(x).
$$
The right hand side is equal to $\KS(y\cnd x)$ (with logarithmic precision) due to Kol\-mo\-go\-rov--Le\-vin formula for the complexity of a pair (see, e.g.,~\cite[Chapter 2]{usv}), and $\KS(y\cnd x)$ equals $l(y)$, as $y$ is random and independent of $x$. Thus it suffices to show that the left hand side also equals $l(y)$. To this end we can prove a version of Kolmogorov--Levin formula for bounded complexity and show that the left hand side equals to $\KS^{B(i)}(y\cnd x)$. Again, since  $y$ is random and independent of $x$, $\KS^{BB(i)}(y\cnd x)$ equals $l(y)$.

This plan needs clarification. First of all, let us explain which version of Kol\-mo\-go\-rov--Le\-vin formula for bounded complexity we need. (Essentially it was published by Longpr\'e in~\cite{longpre-thesis} though the statement was obscured by considering time bound as a function of the input length.)

The equality $\KS(x,y)=\KS(x)+\KS(y\cnd x)$ should be considered as two inequalities, and each one should be treated separately.

\begin{lemma}
1. There exist some constant $c$ and some polynomial $p(\cdot,\cdot)$ such that
$$
\KS^{p(n,t)}(x,y)\le \KS^t(x)+\KS^t(y\cnd x)+c\log n
$$
for all $n$ and $t$ and for all strings $x$ and $y$ of length at most $n$.

2. There exist some constant $c$ and some polynomial $p(\cdot,\cdot)$ such that
$$
\KS^{p(2^n,t)}(x)+\KS^{p(2^n,t)}(y\cnd x) \le \KS^t(x,y)+c \log n
$$
for all $n$ and $t$ and for all strings $x$ and $y$ of length at most $n$.
\end{lemma}

\begin{proof}[Proof of the lemma]
The proof of this time-bounded version is obtained by a straightforward analysis of the time requirements in the standard proof. The first part says that if there is some program $p$ that produces $x$ in time $t$, and some program $q$ that produces $y$ from $x$ in time $t$, then the pair $(p,q)$ can be considered as a program that produces $(x,y)$ in time $\poly(t,n)$ and has length $l(p)+l(q)+O(\log n)$ (we may assume without loss of generality that $p$ and $q$ have length $O(n)$, otherwise we replace them by shorter fast programs).

The other direction is more complicated. Assume that $\KS^t(x,y)=m$. We have to  count for a given $x$ the number of strings $y'$ such that $\KS^{t}(x,y')\le m$. These strings ($y$ is one of them) can be enumerated in time $\poly(2^n,t)$, so if there are $2^s$ of them,  then $\KS^{\poly(2^n,t)}(y\cnd x)\le s+O(\log n)$ (the program witnessing this inequality is the ordinal number of $y$ in the enumeration plus $O(\log n)$ bits of auxiliary information. Note that we do not need to specify $t$ in advance, we enumerate $y'$ in order of increasing time, and $y$ is among first $2^s$ enumerated strings.

On the other hand, there are at most $2^{m-s+O(1)}$ strings $x'$ for which this number (of different $y'$ such that $\KS^t(x',y')\le m$) is at least $2^{s-1}$, and these strings also could be enumerated in time $\poly(2^n,t)$, so $\KS^{\poly(2^n,t)}(x)\le m-s+O(\log n)$ (again we do not need to specify $t$, we just increase gradually the time bound). When these two inequalities are added, $s$ disappears and we get the desired inequality.
\end{proof}

Of course, the exponent in the lemma is disappointing (for space bound it is not needed, by the way), but since we measure time in busy beaver units, it is not a problem for us: indeed, $\poly(2^n,B(i))\le B(i+O(\log n))$, and we allow logarithmic change in the argument anyway.

Now we should apply this lemma, but first we need to give a full statement of what we want to prove. There are two parts (as in the lemma):
\begin{itemize}
\item for every $i$ there exists $j\le i+O(\log n)$ such that
 $$
   \KS^{B(j)}(x,y)-\KS(x,y) \le \KS^{B(i)}(x)-\KS(x)+\varepsilon+O(\log n)
 $$
for all strings $x$ and $y$ of length at most $n$ such that $\KS(y\cnd x)\le l(y)-\varepsilon$;
 \item for every $i$ there exists $j\le i+O(\log n)$ such that
 $$
   \KS^{B(j)}(x)-\KS(x) \le \KS^{B(i)}(x,y)-\KS(x,y)+O(\log n)
 $$
for all strings $x$ and $y$ of length at most $n$;
\end{itemize}
Both statements easily follow from the lemma. Let us start with the second statement where the hard direction of the lemma is used. As planned, we rewrite the inequality as
$$
 \KS^{B(j)}(x)+\KS(y\cnd x)\le \KS^{B(i)}(x,y)+O(\log n)
$$
using the unbounded formula. Our lemma guarantees that
$$
 \KS^{B(j)}(x)+\KS^{B(j)}(y\cnd x)\le \KS^{B(i)}(x,y)+O(\log n)
$$
for some $j\le i+O(\log n)$, and it remains to note that
$
\KS(y\cnd x) \le \KS^{B(j)}(y\cnd x).
$
For the other direction the argument is similar: we rewrite the inequality as
$$
 \KS^{B(j)}(x,y)\le \KS(y\cnd x)+\KS^{B(i)}(x)+O(\log n)
$$
and note that
$\KS(y\cnd x)\ge l(y)-\varepsilon \ge C^{B(i)}(y\cnd x)-\varepsilon$, assuming that $B(i)$ is greater than the time needed to print $y$ from its literary description (otherwise the statement is trivial). So the lemma again can be used (in the simple direction).
\end{proof}

This proof used the depth representation of the stochasticity curve; in other cases some other representation are more convenient. Our second example is the change in stochasticity profile when a simple algorithmic transformation is applied. We have seen (Section~\ref{subsec:stoch-cons}) that a total mapping with a short program preserves stochasticity, and noted that for non-total mapping it is not the case (Remark~\ref{rem:non-total}, p.~\pageref{rem:non-total}). However, if the time needed to perform the transformation is bounded, we can get some bound (first proven by A.~Milovanov in a different way):

\ver{
\begin{thm}\label{prop:stoch-nontotal}
Let $F$ be a partial computable mapping whose arguments and values are strings. If some $n$-bit string $x$ is $(\alpha,\beta)$-stochastic, and $F(x)$ is computed in time $B(i)$ for some $i$, then $F(x)$ is $(\max(\alpha,i)+O(\log n),\beta+O(\log n))$-stochastic. (The constant in $O(\log n)$-notation depends on $F$ but not on $n,x,\alpha,\beta$.)
\end{thm}
}

\begin{proof}[Proof sketch]
Let us denote $F(x)$ by $y$. By assumption there exist a $(\alpha*(\KS(x)-\alpha+\beta))$-description of $x$ (recall the definition with optimality deficiency; we omit logarithmic terms as usual). So there exists a simple enumeration of at most $2^{\KS(x)+\beta}$ objects $x'$ in at most $2^\alpha$ portions that includes $x$. Let us count $x'$ in this enumeration such that $F(x')=y$ and the computation uses time at most $B(i)$; assume there are $2^s$ of them. Then we can enumerate all $y$'s that have at least $2^s$ preimages in time $B(i)$, in $2^\alpha + 2^i$ portions. Indeed, new portions appear in two cases: (1)~a new portion appears in the original enumeration; (2)~candidate for $B(i)$ increases. The first event happens at most $2^\alpha$ times, the second at most $2^i$ times.  The total number of $y$'s enumerated is $2^{\KS(x)+\beta-s}$; it remains to note that $\KS(x)-s \le \KS(y)$. Indeed, $\KS(x)\le \KS(y)+\KS(x\cnd y)$, and $\KS(x\cnd y)\le s$, since we can enumerate all the preimages of $y$ in the order of increasing time, and $x$ is determined by $s$-bit ordinal number of $x$ in this enumeration.
\end{proof}

A special case of this proposition is Bennett's observation: if some $d$-incompressible program $p$ produces $x$ in time $B(i)$, then $p$ is $(0,d)$-stochastic,  and $p$ is mapped to $x$ by the interpreter (decompressor) in time $B(i)$, so $x$ is $(0+i,d)$-stochastic. (For simplicity we omit all the logarithmic terms in this argument, as well as in the previous proof sketch.)

\begin{remark}
One can combine Remark~\ref{rem:cons-f} (page~\pageref{rem:cons-f}) with Proposition~\ref{prop:stoch-nontotal} and show that if a program $F$ of complexity at most $j$ is applied to an $(\alpha,\beta)$-stochastic string $x$ of length $n$ and the computation terminates in time $B(i)$, then $F(x)$ is $(\max(i,\alpha)+j+O(\log n),\beta+j+O(\log n))$-stochastic, where the constant in $O(\log n)$ notation is absolute (does not depend on $F$). To show this, one may consider the pair $(x,F)$; it is easy to show (this can be done in different ways using different characterizations of the stochasticity curve) that this pair is $(\alpha+j+O(\log n), \beta+j+O(\log n))$-stochastic.
\end{remark}

Let us note also that there are some results in algorithmic information theory that are true for stochastic objects but are false or unknown without this assumption. We will discuss (without proofs) two examples of this type. The first is Epstein--Levin theorem saying that for a stochastic set $A$ its total a priori probability is close to the maximum a priori probability of $A$'s elements; see~\cite{shen-survey} for details. Here the result is (obviously) false without stochasticity assumption.

In the next example~\cite{muchnik-romashchenko} the stochasticity assumption is used in the proof, and it is not known whether the statement remains true without it: \emph{for every triple of strings $(x,y,z)$ of length at most $n$ there exists a string $z'$ such that
\begin{itemize}
\item $\KS(x\cnd z) = \KS(x\cnd z')+O(\log n)$,
\item $\KS(y\cnd z) = \KS(y\cnd z')+O(\log n)$,
\item $\KS(x,y\cnd z) = \KS(x,y\cnd z') + O(\log n)$;
\item $\KS(z') \le I((x,y):z)+O(\log n)$,
\end{itemize}
assuming that $(x,y)$ is $(O(\log n),O(\log n))$-stochastic}.

This proposition is related to the following open question on ``irrelevant oracles'': assume that the mutual information between $(x,y)$ and some $z$ is negligible. Can an oracle $z$ (an ``irrelevant oracle'') change substantially natural properties of the pair $(x,y)$ formulated in terms of Kolmogorov complexity? For instance, can such an oracle $z$ allow us to extract some common information of $x$ and $y$? In~\cite{muchnik-romashchenko} a negative answer to the latter question is given, but only for stochastic pairs $(x,y)$.

\section{Descriptions of restricted type}\label{sec:restricted-type}

\subsection{Families of descriptions}

In this section we consider the restricted case: the sets (considered as descriptions, or statistical hypotheses) are taken from some family $\mathcal{A}$ that is fixed in advance.\footnote{One can also consider some class of probability distributions, but we restrict our attention to sets (uniform distributions).} (Elements of $\mathcal{A}$ are finite sets of binary strings.) Informally speaking, this means that we have some \emph{a priori} information about the black box that produces a given string: this string is obtained by a random choice in one of the $\mathcal{A}$-sets, but we do not know in which one.

Before we had no restrictions (the family $\mathcal{A}$ was the family of all finite sets). It turns out that the results obtained so far can be extended (sometimes with weaker bounds) to other families that satisfy some natural conditions. Let us formulate these conditions.

(1)~The family $\mathcal{A}$ is enumerable. This means that there exists an algorithm that prints elements of $\mathcal{A}$ as lists, with some separators (saying where one element of $\mathcal{A}$ ends and another one begins).

(2)~For every $n$ the family $\mathcal{A}$ contains the set $\mathbb{B}^n$ of all $n$-bit strings.

(3)~There exists some polynomial $p$ with the following property: for every $A\in\mathcal{A}$, for every natural $n$ and for every natural $c<\#A$ the set of all $n$-bit strings in $A$ can be covered by at most $p(n)\cdot\#A/c$ sets of cardinality at most $c$ from $\mathcal{A}$.

The last condition is a replacement for splitting: in general, we cannot split a set $A\in\mathcal{A}$ into pieces from $A$, but at least we can cover a set $A\in\mathcal{A}$ by smaller elements of $\mathcal{A}$ (of size at most $c$) with polynomial overhead in the number of pieces, compared to the required minimum $\#A/c$ (more precisely, we have to cover only $n$-bit elements of $A$).

We assume that some family $\mathcal{A}$ that has properties (1)--(3) is fixed. For a string $x$ we denote by $P_x^\mathcal{A}$ the set of pairs $( i,j)$ such that $x$ has $(i*j)$-description \emph{that belongs to $\mathcal{A}$}. The set $P_x^\mathcal{A}$ is a subset of $P_x$ defined earlier; the bigger $\mathcal{A}$ is, the bigger is $P_x^\mathcal{A}$. The full set $P_x$ is $P_x^\mathcal{A}$ for the family $\mathcal{A}$ that contains all finite sets.

For every string $x$ the set $P_x^\mathcal{A}$ has properties close to the properties of $P_x$ proved earlier.

\begin{proposition}\label{prop:a-family}
For every string $x$ of length $n$ the following is true:

\begin{enumerate}
    \item\label{a1} The set $P_x^\mathcal{A}$ contains a pair that is $O(\log n)$-close to $( 0,n)$.

    \item\label{a2} The set $P_x^\mathcal{A}$ contains a pair that is $O(1)$-close to $( \KS(x),0)$.

    \item\label{a3} The adaptation of Proposition~\ref{prop:description-shift} is true: if $( i,j)\in P_x^\mathcal{A}$, then $( i+k+O(\log n),j-k)$ also belongs to $P_x^\mathcal{A}$ for every $k\le j$. (Recall that $n$ is the length of $x$.)
\end{enumerate}
\end{proposition}

\begin{proof}
\ref{a1}. The property (2) guarantees that the family $\mathcal{A}$ contains the set $\mathbb{B}^n$ that is an $(O(\log n)*n)$-description of $x$.

\ref{a2}.  The property (3) applied to $c=1$ and $A=\mathbb{B}^n$ says that every singleton belongs to $A$, therefore each string has $((\KS(x)+O(1))*0)$-description.

\ref{a3}. Assume that $x$ has $(i*j)$-description $A\in\mathcal{A}$. For a given $k$ we enumerate $\mathcal{A}$ until we find a family of $p(n)2^k$ sets of size $2^{-k}\#A$ (or less) in $\mathcal{A}$ that covers all strings of length $n$ in $A$. Such a family exists due to (3), and $p$ is the polynomial from~(3). The complexity of the set that covers $x$ does not exceed $i+k+O(\log n+\log k)$, since this set is determined by $A$, $n$, $k$ and the ordinal number of the set in the cover. We may assume without loss of generality that $k\le n$, otherwise $\{x\}$ can be used as $((i+k+O(\log n))*(j-k))$-description of $x$. So the term $O(\log k)$ can be omitted.
\end{proof}

For example, we may consider the family that consists of all ``cylinders'': for every $n$ and for every string $u$ of length at most $n$ we consider the set of all $n$-bit strings that have prefix $u$.  Obviously the family of all such sets (for all $n$ and $u$) satisfies the conditions (1)--(3).

We may also fix some bits of a string (not necessarily forming a prefix).  That is, for every string $z$ in ternary alphabet $\{0,1,*\}$ we consider the set of all bit strings that can be obtained from $z$ by replacing stars with some bits. This set contains $2^k$ strings, if $u$ has $k$ stars. The conditions (1)--(3) are fulfilled for this larger family, too.

A more interesting example is the family $\mathcal{A}$ formed by all balls in Hamming sense, i.e., the sets $B_{y,r}=\{x\mid  l(x)=l(y), d(x,y)\le r\}$. Here $l(u)$ is the length of binary string $u$, and $d(x,y)$ is the Hamming distance between two strings $x$ and $y$ of the same length. The parameter $r$ is called the \emph{radius} of the ball, and $y$ is its \emph{center}. Informally speaking, this means that the experimental data were obtained by changing at most $r$ bits in some string $y$ (and all possible changes are equally probable). This assumption could be reasonable if some string $y$ is sent via an unreliable channel. Both parameters $y$ and $r$ are not known to us in advance.

It turns out that the family of Hamming balls satisfies the conditions (1)--(3). This is not completely obvious. For example, these conditions imply that for every $n$ and for every $r\le n$ the set $\mathbb{B}^n$ of $n$-bit strings can be covered by $\poly(n)2^n/V$ Hamming balls of radius $r$, where $V$ stands for the cardinality of such a ball (i.e., $V=\binom{n}{0}+\ldots+\binom{n}{r}$), and $p$ is some polynomial. This can be shown by a probabilistic argument: take $N$ balls of radius $r$ whose centers are randomly chosen in $\mathbb{B}^n$. For a given $x\in\mathbb{B}^n$ the probability that $x$ is not covered by any of these balls equals $(1-V/2^n)^N < e^{-VN/2^n}$. For $N=n\ln 2\cdot 2^n/V$ this upper bound is $2^{-n}$, so for this $N$ the probability to leave some $x$ uncovered is less than~$1$. A similar argument can be used to prove (1)--(3) in the general case.

\begin{proposition}[\cite{vv10}]\label{prop:hamming-balls}
The family of all Hamming balls satisfies conditions (1)--(3) above.
\end{proposition}

\begin{proof}[Proof sketch]
Let $A$ be a ball of radius $a$ and let $c$ be a number less than $\#A$. We need to cover $A$ by balls of cardinality $c$ or less, using almost minimal number of balls, close to the lower bound $\#A/c$ up to a polynomial factor. Let us make some observations.

(1)~The set of all $n$-bit strings can be covered by two balls of radius $n/2$. So we can assume without loss of generality that $a\le n/2$, otherwise we can apply the probabilistic argument above.

(2)~Clearly the radius of covering balls should be maximal possible (to keep cardinality less than $c$); for this radius the cardinality of the ball equals $c$ up to polynomial factors, since the size of the ball increases at most by factor $n+1$ when its radius increases by $1$.

(3)~It is enough to cover spheres instead of balls (since every ball is a union of polynomially many spheres); it is also enough to consider the case when the radius of the sphere that we want to cover ($a$) is bigger than the radius of the covering ball ($b$), otherwise one ball is enough.

(4)~We will cover $a$-sphere by randomly chosen $b$-balls whose centers are uniformly taken at some distance $f$ from the center of $a$-sphere. (See below about the choice of $f$.) We use the same probabilistic argument as before (for the set of all strings). It is enough to show that for a $b$-ball whose center is at that distance, the polynomial fraction of points belong to $a$-sphere. Instead of $b$-balls we may consider $b$-spheres, the cardinality ratio is polynomial.

(5)~It remains to choose some $f$ with the following property:  if the center of a $b$-sphere $S$ is at a distance $f$ from the center of $a$-sphere $T$, then the polynomial fraction of $S$-points  belong to $T$. One can compute a suitable $f$ explicitly. In probabilistic terms we just change $f/n$-fraction of bits and then change random $b/n$ fraction of bits. The expected fraction of twice changed bits is, therefore, about $(f/n)(b/n)$, and the total fraction of changed bits is about $f/n+b/n-2(f/n)(b/n)$. So we need to write an equation saying that this expression is $a/n$ and the find the solution $f$. (Then one can perform the required estimate for binomial coefficients.)

However, one can avoid computations with the following probabilistic argument: start with $b$ changed bits, and then change all the bits one by one in a random order. At the end we hat $n-b$ changed bits, and $a$ is somewhere in between, so there is a moment where the number of changed bits is exactly $a$. And if the union of $n$ events covers the entire probability space, one of these events has probability at least $1/n$.
\end{proof}

When a family $\mathcal{A}$ is fixed, a natural question arises: does the restriction on models (when we consider only models in $\mathcal{A}$) changes the set $P_x$? Is it possible that a string has good models in general, but not in the restricted class? The answer is positive for the class of Hamming balls, as the following proposition shows.

\begin{proposition}\label{prop:hamming-gap}
Consider the family $\mathcal{A}$ that consists of all Hamming balls. For some positive $\varepsilon$ and for all sufficiently large $n$ there exists a string $x$ of length $n$ such that the distance between $P_x^\mathcal{A}$ and $P_x$ exceeds $\varepsilon n$.
\end{proposition}

\begin{proof}[Proof sketch]
Fix some $\alpha$ in $(0,1/2)$ and let $V$ be the cardinality of the Hamming ball of radius $\alpha n$. Find a set $E$ of cardinality $N=2^n/V$ such that every Hamming ball of radius $\alpha n$ contains at most $n$ points from $E$. This property is related to \emph{list decoding} in the coding theory. The existence of such a set can be proved by a probabilistic argument: $N$ randomly chosen $n$-bit strings have this property with positive probability. Indeed, the probability of a random point to be in $E$ is an inverse of the number of points, so the distribution is close to Poisson distribution with parameter~$1$, and tails decrease much faster that $2^{-n}$ needed.

Since $E$ with this property can be found by an exhaustive search, we can assume that $\KS(E)=O(\log n)$ and ignore the complexity of $E$ (as well as other $O(\log n)$ terms) in the sequel. Let $x$ be a random element in $E$, i.e., a string $x\in E$ of complexity about $\log\#E$.  The complexity of a ball $A$ of radius $\alpha n$ that contains $x$ is at least $\KS(x)$, since knowing such a ball and an ordinal number of $x$ in $A\cap E$, we can find $x$. Therefore $x$ does not have $(\log\#E,\log V)$-descriptions in $\mathcal A$. On the other hand, $x$ does have $(0,\log\#E)$-description if we do not require the description to be in $\mathcal A$; the set $E$ is such a description. The point $(\log\#E, \log V)$ is above the line $\KS(A)+\log\#A=\log\#E$, so $P_x^\mathcal A$ is significantly smaller than $P_x$.
\end{proof}

This construction gives a stochastic $x$ ($E$ is the corresponding model) that becomes maximally non-stochastic if we restrict ourselves to Hamming balls as descriptions (Figure~\ref{mdl-e-13}).

\begin{figure}
\begin{center}
\includegraphics[scale=1]{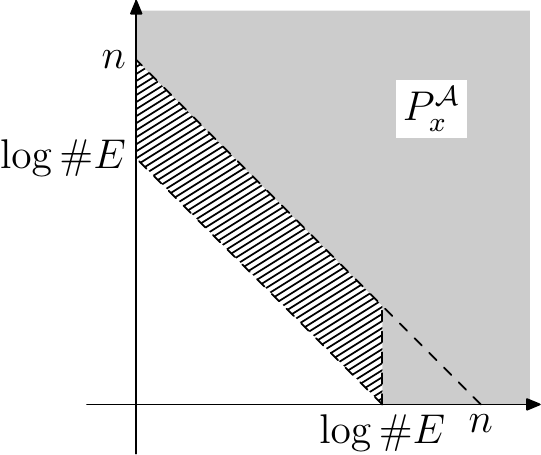}
\end{center}
\caption{Theorem~\ref{thm:improving-descriptions-1-gen} can be used (together with the argument above) to show that the border of the set $P_x^\mathcal{A}$ (shown in gray) consists of a vertical segment $\KS(A)=n-\log V$, $\log\#A\le \log V$, and the segment of slope $-1$ defined by $\KS(A)+\log\#A = n$, $\log V \le \log\#A$. The set $P_x$ contains also the hatched part.}\label{mdl-e-13}
\end{figure}

\subsection{Possible shapes of boundary curve}

Our next goal is to extend some results proven for non-restricted descriptions to the restricted case. Let $\mathcal{A}$ be a family that has properties (1)--(3). We prove a version of Theorem~\ref{stat-any-curve} where the precision (unfortunately) is significantly worse: $O(\sqrt{n\log n})$ instead of $O(\log n)$. Note that with this precision the term $O(m)$ (proportional to the complexity of the curve) that appeared in Theorem~\ref{stat-any-curve} is not needed. Indeed, if we draw the curve on the cell paper with cell size $\sqrt{n}$ or larger, then it touches only $O(\sqrt{n})$ cells, so it is determined by $O(\sqrt{n})$ bits with $O(\sqrt{n})$-precision, and we may assume without loss of generality that the complexity of the curve is $O(\sqrt{n})$.

\begin{thm}[\cite{vv10}]\label{thm:family-curve}
Let $k\le n$ be two integers and let $t_0>t_1>\ldots>t_k$ be a strictly decreasing sequence of integers such that $t_0\le n$ and $t_k=0$.. Then there exists a string $x$ of complexity $k+O(\sqrt{n\log n})$ and length $n+O(\log n)$ for which the distance between $P_x^\mathcal{A}$ and $T=\{( i,j) \mid (i\le k)\Rightarrow (j\ge t_i)\}$ is at most $O(\sqrt{n\log n})$.
\end{thm}

We will see later (Theorem~\ref{thm:improving-descriptions-1-gen}) that for every $x$ the boundary curve of $P_x^{\mathcal{A}}$ goes down at least with slope $-1$, as for the unrestricted case, so this theorem describes all possible shapes of the boundary curve.

\begin{proof}
The proof is similar to the proof of Theorem~\ref{stat-any-curve}. Let us recall this proof first. We consider the string $x$ that is the lexicographically first string (of suitable length $n'$) that is not covered by any ``bad'' set, i.e., by any set of complexity at most $i$ and size at most $2^j$, where the pair $(i,j)$ is at the boundary of the set $T$. The length $n'$ is chosen in such a way that the total number of strings in all bad sets is strictly less than $2^{n'}$. On the other hand, we need ``good sets'' that cover $x$. For every boundary point $(i,j)$ we construct a set $A_{i,j}$ that contains $x$, has complexity close to $i$ and size $2^j$. The set $A_{i,j}$ is constructed in several attempts. Initially $A_{i,j}$ is the set of lexicographically first $2^j$ strings of length $n'$. Then we enumerate bad sets and delete all their elements from $A_{i,j}$. At some step $A_{i,j}$ may become empty; then we refill it with $2^j$ lexicographically first strings that are not in the bad sets (at the moment). By construction the final $A_{i,j}$ contains the first $x$ that is not in bad sets (since it is the case all the time). And the set $A_{i,j}$ can be described by the number of changes (plus some small information describing the process as a whole and the value of $j$). So it is crucial to have an upper bound for the number of changes. How do we get this bound? We note that when $A_{i,j}$ becomes empty, it is refilled again, and all the new elements should be covered by bad sets before the new change could happen. Two types of bad sets may appear: ``small'' ones (of size less than $2^j$) and ``large ones'' (of size at least $2^j$). The slope of the boundary line for $T$ guarantees that the total number of elements in all small bad sets does not exceed $2^{i+j}$ (up to a $\poly(n)$-factor), so they may make $A_{i,j}$ empty only $2^i$ times. And the number of large bad sets is $O(2^i)$, since the complexity of each is bounded by $i$. (More precisely, we count separately the number of changes for $A_{i,j}$ that are first changes after a large bad set appears, and the number of other changes.)

Can we use the same argument in our new situation? We can generate bad sets as before and have the same bounds for their sizes and the total number of their elements. So the length $n'$ of $x$ can be the same (in fact, almost the same, as we will need now  that the union of all bad sets is less than half of all strings of length $n'$, see below).  Note that we now may enumerate only bad sets in $\mathcal{A}$, since $\mathcal{A}$ is enumerable, but we do not even need this restriction. What we cannot do is to let $A_{i,j}$ to be the set of the first non-deleted elements: we need $A_{i,j}$ to be a set from $\mathcal{A}$.

So we now go in the other direction. Instead of choosing $x$ first and then finding suitable ``good'' $A_{i,j}$ that contain $x$, we construct the sets $A_{i,j}\in\mathcal{A}$ that change in time in such a way  that (1)~their intersection always contains some non-deleted element (an element that is not yet covered by bad sets); (2) each $A_{i,j}$ has not too many versions. The non-deleted element in their intersection (in the final state) is then chosen as $x$.

Unfortunately, we cannot do this for all points $(i,j)$ along the boundary curve. (This explains the loss of precision in the statement of the theorem.) Instead, we construct ``good'' sets only for some values of $j$. These values go down from $n$ to $0$ with step $\sqrt{n\log n}$. We select $N=\sqrt{n/\log n}$ points $(i_1,j_1),\ldots,(i_N,j_N)$ on the boundary of $T$; the first coordinates $i_1,\ldots,i_N$ form a non-decreasing sequence, and the second coordinates $j_1,\ldots,j_N$ split the range $n\ldots 0$ into (almost) equal intervals ($j_1=n$, $j_N=0$). Then we construct good sets of sizes at most $2^{j_1},\ldots,2^{j_N}$, and denote them by $A_1,\ldots,A_N$. All these sets belong to the family $\mathcal{A}$. We also let $A_0$ to be the set of all strings of length $n'=n+O(\log n)$; the choice of the constant in $O(\log n)$ will be discussed later.

Let us first describe the construction of $A_1,\ldots,A_N$ assuming that the set of deleted elements is fixed. (Then we discuss what to do when more elements are deleted.) We construct $A'$ inductively (first $A_1$, then $A_2$ etc.). As we have said, $\#A'\le 2^{j_s}$ (in particular, $A_N$ is a singleton), and we keep track of the ratio
   $$(\text{the number of non-deleted strings in $A_0\cap A_1\cap\ldots\cap A'$})/2^{j_s}.$$
For $s=0$ this ratio is at least $1/2$; this is obtained by a suitable choice of $n'$ (the union of all bad sets should cover at most half of all $n'$-bit strings). When constructing the next $A'$, we ensure that this ratio decreases only by $\poly(n)$-factor. How? Assume that $A_{s-1}$ is already constructed; its size is at most $2^{j_{s-1}}$. The condition $(3)$ for $\mathcal{A}$ guarantees that $A_{s-1}$ can be covered by $\mathcal{A}$-sets of size at most $2^{j_s}$, and we need about $2^{j_{s-1}-j_s}$ covering sets (up to $\poly(n)$-factor). Now we let $A'$ be the covering set that contains maximal number of non-deleted elements in $A_0\cap\ldots\cap A_{s-1}$. The ratio can decrease only by the same $\poly(n)$-factor. In this way we get
  $$(\text{the number of non-deleted strings in $A_0\cap A_1\cap\ldots\cap A'$})\ge \alpha^{-s}2^{j_s}/2,$$
where $\alpha$ stands for the $\poly(n)$-factor mentioned above.\footnote{Note that for the values of $s$ close to $N$ the right-hand side can be less than $1$; the inequality then claims just the existence of non-deleted elements. The induction step is still possible: non-deleted element is contained in one of the covering sets.}

Up to now we assumed that the set of deleted elements is fixed. What happens when more strings are deleted? The number of the non-deleted in $A_0\cap\ldots\cap A_{s}$ can decrease, and at some point and for some $s$ can become less than the declared threshold $\nu_s=\alpha^{-s} 2^{j_s}/2$. Then we can find minimal $s$ where this happens, and rebuild all the sets $A',A_{s+1},\ldots$ (for $A'$ the threshold is not crossed due to the minimality of $s$). In this way we update the sets $A'$ from time to time, replacing them (and all the consequent ones) by new versions when needed.

The problem with this construction is that the number of updates (different versions of each $A'$) can be too big. Imagine that after an update some element is deleted, and the threshold is crossed again. Then a new update is necessary, and after this update next deletion can trigger a new update, etc. To keep the number of updates reasonable, we agree that after the update \emph{for all the new sets $A_l$} (starting from $A'$) \emph{the number of non-deleted elements in $A_0\cap\ldots\cap A_l$ is twice bigger than the threshold $\nu_l=\alpha^{-l}2^{j_l}/2$}. This can be achieved if we make the factor $\alpha$ twice bigger: since for $A_{s-1}$ we have not crossed the threshold, for $A'$ we can guarantee the inequality with additional factor $2$.

Now let us prove the bound for the number of updates for some $A'$. These updates can be of two types: first, when $A'$ itself starts the update (being the minimal $s$ where the threshold is crossed); second, when the update is induced by one of the previous sets. Let us estimate the number of the updates of the first type. This update happens when the number of non-deleted elements (that was at least $2\nu_s$ immediately after the previous update of any kind) becomes less than $\nu_s$. This means that at least $\nu_s$ elements were deleted. How can this happen? One possibility is that a new bad set of complexity at most $i_s$ (``large bad set'') appears after the last update. This can happen at most $O(2^{i_s})$ times, since there is at most $O(2^i)$ objects of complexity at most $i$. The other possibility is the accumulation of elements deleted due to ``small'' bad sets, of complexity at least $i_s$ and of size at most $2^{j_s}$. The total number of such elements is bounded by $nO(2^{i_s+j_s})$, since the sum $i_l+j_l$ may only decrease as $l$, increases. So the number of updates of $A'$ not caused by large bad sets is bounded by
      $$ n O(2^{i_s+j_s}) /\nu_s  =\frac{O(n2^{i_s+j_s})}{\alpha^{-s}2^{j_s}} = O(n\alpha^s 2^{i_s})=2^{i_s+NO(\log n)}=2^{i_s+O(\sqrt{n\log n})}$$
(recall that $s\le N$, $\alpha=\poly(n)$, and $N\approx \sqrt{n/\log n}$). This bound remains valid if we take into account the induced updates (when the threshold is crossed for the preceding sets: there are at most $N\le n$ these sets, and additional factor $n$ is absorbed by $O$-notation).

We conclude that all the versions of $A'$ have complexity at most $i_s+O(\sqrt{n\log n})$, since each of them can be described by the version number plus the parameters of the generating process (we need to know $n$ and the boundary curve, whose complexity is $O(\sqrt{n})$ according to our assumption, see the discussion before the statement of the theorem). The same is true for the final version. It remains to take $x$ in the intersection of the final sets $A'$. (Recall that $A_N$ is  a singleton, so final $A_N$ is $\{x\}$.) Indeed, by construction this $x$ has no bad $(i*j)$-descriptions where $(i,j)$ is on the boundary of $T$. On the other hand, $x$ has good descriptions that are $O(\sqrt{n\log n})$-close to this boundary and whose vertical coordinates are $\sqrt{n\log n}$-apart. (Recall that the slope of the boundary guarantees that horizontal distance is less than the vertical distance.) Therefore the position of the boundary curve for $P_x^\mathcal{A}$ is determined with precision $O(\sqrt{n\log n})$, as required.\footnote{Now we see why $N$ was chosen to be $\sqrt{n/\log n}$: the bigger $N$ is, the more points on the curve we have, but then the number of versions of the good sets and their complexity increases, so we have some trade-off. The chosen value of $n$ balances these two sources of errors.}
 \end{proof}

\begin{remark}\label{rem:family}
In this proof we may use bad sets not only from $\mathcal{A}$. Therefore, the set $P_x$ is also close to $T$ (and the same is true for for every family $\mathcal{B}$ that contains $\mathcal{A}$). It would be interesting to find out what are the possible combinations of $P_x$ and $P_x^\mathcal{A}$; as we have seen, it may happen that $P_x$ is maximal and $P_x^\mathcal{A}$ is minimal, but this does not say anything about other possible combinations.  \end{remark}

For the case of Hamming balls the statement of Theorem~\ref{thm:family-curve} has a natural interpretation. To find a simple ball of radius $r$ that contains a given string $x$ is the same as to find a simple string in a radius $r$ ball centered at $x$. So this theorem show the possible behavior of the ``approximation complexity'' function
$$
 r\mapsto \min \{\KS(x')\mid d(x,x')\le r\}
$$
where $d$ is Hamming distance. One should only rescale the vertical axis replacing the log-sizes of Hamming balls by their radii. The connection is described by the Shannon entropy function: a ball in $\mathbb{B}^n$ of radius $r$ has log-size about $nH(r/n)$ for $r\le n/2$, and has almost full size for $r\ge n/2$. For example, error correcting codes (in classical sense, or with list decoding) are example of strings where this function is almost a constant for small values of $r$: it is almost as easy to approximate a codeword as give it precisely (due to the possibility of error correction).

\subsection{Randomness and optimality deficiencies: restricted case}

Not all the results proved for unrestricted descriptions have natural counterparts in the restricted case. For example, one hardly can relate the set $P_x^\mathcal{A}$ with bounded-time complexity (is completely unclear how $\mathcal{A}$ could enter the picture). Still some results remain valid (but new and much more complicated proofs are needed). This is the case for Proposition~\ref{prop:description-shift} and~\ref{prop:improving-descriptions}.

Let again $\mathcal{A}$ be the class of descriptions that satisfies requirements (1)--(3).

\begin{thm}[\cite{vv10}]\label{thm:improving-descriptions-1-gen}
\leavevmode
\begin{itemize}
\item If a string $x$ of length $n$ has an $(i*j)$-description in $\mathcal{A}$, then it has $((i+d+O(\log n))*(j-d+O(\log n)))$-description in $\mathcal{A}$ for every $d\le j$.

\item Assume that $x$ is a string of length $n$ that has at least $2^k$ different $(i*j)$-descriptions in $\mathcal{A}$. Then it has $((i-k+O(\log n))*(j+O(\log n))$-description in $\mathcal{A}$.
\end{itemize}
\end{thm}

In fact, the second part uses only condition (1); it says that $\mathcal{A}$ is enumerable. The first part uses also (3). It can be combined with the second part to show that $x$ has also $((i+O(\log n))*(j-k+O(\log n))$-description in $\mathcal{A}$.
\smallskip

Though theorem \ref{thm:improving-descriptions-1-gen} looks like a technical statement, it has important consequences; it implies that the two approaches based on randomness and optimality deficiencies remain equivalent in the case of bounded class of descriptions. The proof technique can be also used to prove Epstein--Levin theorem~\cite{epstein-levin}, as explained in~\cite{shen-survey}; similar technique was used by A.~Milovanov in \cite{milovanov-stacs} where a common model for several strings is considered.

\begin{proof}
The first part is easy: having some $(i*j)$-description for $x$, we can search for a covering by the sets of right size that exists due to condition~(3); since $\mathcal{A}$ is enumerable, we can do it algorithmically until we find this covering. Then we select the first set in the covering that contains $x$; the bound for the complexity of this set is guaranteed by the size of the covering.

The proof of the second statement is much more interesting. In fact, there are two different proofs: one uses a probabilistic existence argument and the second is more explicit. But both of them start in the same way.

Let us enumerate all $(i*j)$-descriptions from $\mathcal A$, i.e., all finite sets that belong to $\mathcal A$, have cardinality at most $2^j$ and complexity at most $i$. For a fixed $n$, we start a selection process: some of the generated descriptions are marked (=selected) immediately after their generation. This process should satisfy the following requirements: (1)~at any moment every $n$-bit string $x$ that has at least $2^k$ descriptions (among enumerated ones) belongs to one of the marked descriptions; (2)~the total number of marked sets does not exceed $2^{i-k}p(n)$ for some polynomial~$p$. Note that for $i\ge n$ or $j\ge n$ the statement is trivial, so we may assume that $i$, $j$ (and therefore $k$) do not exceed $n$; this explains why the polynomial depends only on $n$.

If we have such a strategy (of logarithmic complexity), then the marked set containing $x$ will be the required description of complexity $i-k+O(\log n)$ and log-size $j$. Indeed, this marked set can be specified by its ordinal number in the list of marked sets, and this ordinal number has $i-k+O(\log n)$ bits.

So we need to construct a selection strategy of logarithmic complexity. We present two proofs: a probabilistic one and an explicit construction.

\textsc{Probabilistic proof}. First we consider a finite game that corresponds to our situation. Two players alternate, each makes $2^i$ moves. At each move the first player presents some set of $n$-bit strings, and the second player replies saying whether it \emph{marks} this set or not. The second player loses if after some moves the number of marked sets exceeds $2^{i-k+1}(n+1)\ln 2$ (this specific value follows from the argument below) or if there exists a string $x$ that belongs to $2^k$ sets of the first player but does not belong to any marked set.

Since this is a finite game with full information, one of the players has a winning strategy. We claim that the second player can win. If it is not the case, the first player has a winning strategy. We get a contradiction by showing that the second player has a \emph{probabilistic} strategy that wins with positive probability against any strategy of the first player. So we assume that some (deterministic) strategy of the first player is fixed, and consider the following simple probabilistic  strategy: every set $A$ presented by the first player is marked with probability $p=2^{-k}(n+1)\ln 2$.

The expected number of marked sets is $p2^i=2^{i-k}(n+1)\ln 2$. By Chebyshev's inequality, the number of marked set exceeds the expectation by a factor $2$ with probability less than $1/2$. So it is enough to show that the second bad case (after some move there exists $x$ that belongs to $2^k$ sets of the first player but does not belong to any marked set) happens with probability at most $1/2$.

For that, it is enough to show that for every fixed $x$ the probability of this bad event is at most $2^{-(n+1)}$, and then use the union bound.  The intuitive explanation is simple: if $x$ belongs to $2^k$ sets, the second player had (at least) $2^k$ chances to mark a set containing $x$ (when these $2^k$ sets were presented by the first player), and the probability to miss all these chances is at most $(1-p)^{2^k}$; the choice of $p$ guarantees that this probability is less than $1/2^{-(n+1)}$. Indeed, using the bound $(1-1/x)^x < 1/e$, it is easy to show that $(1-p)^{2^k} < e^{-(n+1)\ln 2}=2^{-(n+1)}$.

The pedantic reader would say that this argument is not formally correct, since the behavior of the first player (and the moment when next set containing $x$ is produced) depends on the moves of the second player, so we do not have independent events with probability $1-p$ each (as it is assumed in the computation).\footnote{The same problem appears if we observe a sequence of  independent coin tossings with probability of success $p$, select some trials (before they are actually performed, based on the information obtained so far), and ask for the probability of the event ``$t$ first selected trials were all unsuccessful''. This probability does not exceed $(1-p)^t$; it can be smaller if the total number of selected trials is less than $t$ with positive probability. This scheme was considered by von Mises when he defined random sequences using selection rules, so it should be familiar to algorithmic randomness people.}  The formal argument considers for each $t$ the event $R_t$: ``after some move of the second player the string $x$ belongs to at least $t$ sets provided by the first player, but does not belong to any marked set''. Then we prove by induction (over $t$) that the probability of $R_t$ does not exceed $(1-p)^t$. Indeed, it is easy to see that $R_t$ in a union of several disjoint subsets (depending on the events happening until the first player provides $t+1$ sets containing $x$), and $R_{t+1}$ is obtained by taking a $(1-p)$-fraction in each of them.

\textsc{Constructive proof}. We consider the same game, but now allow more sets to be marked (replacing the bound $2^{i-k+1}(n+1)\ln 2$ by a bigger bound $2^{i-k}i^2\ln 2$) and also allow the second player to mark sets that were produced earlier (not necessarily at the current move of the first player). The explicit winning strategy for the second player performs in parallel $i-k+\log i$ substrategies (indexed by the numbers $\log (2^k/i),\ldots,i$).

The substrategy number $s$ wakes up once in $2^s$ moves (when the number of moves made by the first player is a multiple of $2^s$).  It considers a family $S$ that consists of $2^s$ last sets produced by the first player, and the set $T$ that consists of all strings $x$ covered by at least $2^k/i$ sets from $S$. Then it selects and marks some elements in $S$ in such a way that all $x\in T$ are covered by one of the selected sets. It is done by a greedy algorithm: first take a set from $S$ that covers maximal part of $T$, then the set that covers maximal number of non-covered elements, etc. How many steps do we need to cover the entire $T$? Let us show that
    $$ (i/2^k)n 2^{s} \ln 2 $$
steps are enough. Indeed, every element of $T$ is covered by at least $2^k/i$ sets from $S$. Therefore, some set from $S$ covers at least $\#T2^k/(i2^s)$ elements, i.e., $2^{k-s}/i$-fraction of $T$. At the next step the non-covered part is multiplied by $(1-2^{k-s}/i)$ again, and after $in2^{s-k}\ln 2$ steps the number of non-covered elements is bounded by
    $$ \#T (1-2^{k-s}/i)^{in2^{s-k}\ln 2} < 2^n (1/e)^{n\ln 2} = 1,$$
therefore all elements of $T$ are covered. (Instead of a greedy algorithm one may use a probabilistic argument and show that randomly chosen $in2^{s-k}\ln 2$ sets from $S$ cover $T$ with positive probability; however, our goal is to construct an explicit strategy.)

Anyway, the number of sets selected by a substrategy number $s$, does not exceed
     $$ in2^{s-k}(\ln 2)2^{i-s} = in2^{i-k}\ln 2, $$
and we get at most $i^2 n 2^{i-k} \ln 2$ for all substrategies.

It remains to prove that after each move of the second player every string $x$ that belongs to $2^k$ or more sets of the first player, also belongs to some selected set. For $t$th move we consider the binary representation of $t$:
   $$t=2^{s_1}+2^{s_2}+\ldots, \text{ where } s_1>s_2>\ldots $$
Since $x$ does not belong to the sets selected by substrategies with numbers $s_1,s_2,\ldots$, the multiplicity of $x$ among the first $2^{s_1}$ sets is less than $2^k/i$, the multiplicity of $x$ among the next $2^{s_2}$ sets is also less than $2^k/i$, etc. For those $j$ with $2^{s_j}<2^k/i$ the multiplicity of $x$ among the respective portion of $2^{s_j}$ sets is obviously less than $2^k/i$. Therefore, we conclude that the total multiplicity of $x$ is less that $i\cdot 2^k/i=2^k$ sets of the first player and the second player does not need to care about~$x$. This finishes the explicit construction of the winning strategy.

Now we can assume without loss of generality that the winning strategy has complexity at most $O(\log(n+k+i+j))$. (In the probabilistic argument we have proved the existence of a winning strategy, but then we can perform the exhaustive search until we find one; the first strategy found will have small complexity.) Then we use this simple strategy to play with the enumeration of all $\mathcal{A}$-sets of complexity less than $i$ and size $2^j$ (or less). The selected sets can be described by their ordinal number (among the selected sets), so their complexity is bounded by $i-k$ (with logarithmic precision). Every string that has $2^k$ different $(i*j)$-descriptions in $\mathcal{A}$, will also have one among the selected sets, and that is what we need.
\end{proof}

As before (for the unrestricted case), this result implies that descriptions with minimal parameters are simple with respect to the data string:
\begin{thm}[\cite{vv10}]\label{thm:improving-descriptions-2-gen}
Let $\mathcal A$ be an enumerable family of finite sets. If a string $x$ of length $n$ has $(i*j)$-description $A\in\mathcal{A}$ such that $\KS(A\cnd x)\ge k$, then $x$ has a $((i-k+O(\log n))*(j+O(\log n)))$-description in $\mathcal{A}$. If the family $\mathcal{A}$ satisfies the condition $(3)$, then $x$ has also a $((i+O(\log n))*(j-k+O(\log n)))$-description in~$\mathcal{A}$.
\end{thm}

This gives us the same corollaries as in the unrestricted case:

\begin{corollary}
Let $\mathcal{A}$ be a family of finite sets that satisfies the conditions (1)--(3). Then for every string $x$ of length $n$ three statements
\begin{itemize}
\item there exists a set $A\in\mathcal{A}$ of complexity at most $\alpha$ with $d(x\cnd A)\le\beta$;
\item there exists a set $A\in\mathcal{A}$ of complexity at most $\alpha$ with $\delta(x,A)\le\beta$;
\item the point $(\alpha,\KS(x)-\alpha+\beta)$ belongs to $P_x^\mathcal{A}$
\end{itemize}
are equivalent with logarithmic precision (the constants before the logarithms depend on the choice of the set $\mathcal{A}$).
\end{corollary}

If we are interested in the uniform statements true for every enumerable family $\mathcal{A}$, the same arguments prove the following result:

\begin{proposition}
Let $\mathcal{A}$ be an arbitrary family of finite sets enumerated by some program $p$. Then for every $x$ of length $n$ the statements
\begin{itemize}
\item there exists a set $A\in\mathcal{A}$ such that $d(x\cnd A)\le \beta$;
\item there exists a set $A\in\mathcal{A}$ such that $\delta(x,A)\le\beta$
\end{itemize}
are equivalent up to $O(\KS(p)+\log\KS(A)+\log n+\log\log\#A)$-change in the parameters.
\end{proposition}

\ver{
\section{Strong models}\label{sec:strong-models}

\subsection{Information in minimal descriptions}

A possible way to bring the theory in accordance to our intuition is to change the definition of ``having the same information''. Although we have not given that definition explicitly, we have adopted so far the following viewpoint: $x$ and $y$ have the same (or almost the same) information if both conditional complexities $\KS(x\cnd y),\KS(y\cnd x)$ are small. If only one complexity, say $\KS(x\cnd y)$, is small, we said that all (or almost all) information contained in $x$ is present in $y$.

Now we will adopt a more restricted viewpoint and say that $x$ and $y$ have the same information if there are short \emph{total} (everywhere defined) programs mapping $x$ to $y$ and vice versa. From this viewpoint we cannot say anymore that a string $x$ and its shortest program $x^*$ have the same information: for example, $x$ may be non-stochastic while $x^*$ is always stochastic, so there is no short total program that maps $x^*$ to $x$ because of Proposition~\ref{prop:stoch-cons}.\footnote{It is worth to mention that  on the other hand, for every string $x$ there is an almost minimal program for $x$ that can be obtained from $x$ by a simple total algorithm~\cite[Theorem 17]{ver2015}.} Let us mention that if $x$ and $y$ have the same information in this new sense, then there exists a simple computable \emph{bijection} that maps $x$ to $y$ (so they have the same properties if the property is defined in the computability language), see~\cite{mezhirov} for the proof.

Formally, let us define the total conditional complexity with respect to a computable function $D$ of two arguments, as
     $$
\KT_D(x\cnd y)=\min\{l(p)\mid D(p,y)=x,\text{ and $D(p,y')$ is defined for all } y'\}.
    $$
(Note that $D$ is not required to be total, but we consider only $p$ such that $D(p,y')$ is defined for all $y'$.)

There is a computable function $D$ such that $\KT_D$ is minimal up to an additive constant. Fixing any such $D$ we obtain the  \emph{total conditional complexity} $\KT(x\cnd y)$. In other way, we may define $\KT(x\cnd y)$ as the minimal plain complexity of a total program that maps $y$ to $x$.

We will think that  $y$  has all (or almost all) the information from $x$ if $\KT(x\cnd y)$ is negligible.  Formally, we write  $x\stackrel{\varepsilon}{\rightarrow}y$ if $\KT(y\cnd x)\le\varepsilon$  and we call $x$ and $y$ \emph{$\varepsilon$-equivalent} and write $x\stackrel{\varepsilon}{\leftrightarrow}y$, if both $\KT(y\cnd x)$ and $\KT(x\cnd y)$ are at most $\varepsilon$.

\begin{proposition}\label{prop:equivalence}
If $x\stackrel{\varepsilon}{\leftrightarrow}y$ then the sets $P_x$ and $P_y$ are in $O(\varepsilon)$ neighborhood of each other.
\end{proposition}

\begin{proof}
Indeed, if $A$ is an $(i*j)$-description of $x$ and $p$ is a total program witnessing  $x\stackrel{\varepsilon}{\rightarrow}y$, then the set $B=\{D(p,x')\mid x'\in A\}$ is an $((i+O(\varepsilon))*j)$-description for $y$. (We need $p$ to be total, as otherwise we cannot produce the list of $B$-elements from the list of $A$-elements and $p$.)
\end{proof}

\subsection{An attempt to separate ``good'' models from ``bad'' ones}

Now we have a more fine-grained classification of descriptions and can try to distinguish between descriptions that were equivalent in the former sense. For example, consider a string $xy$ where $y$ is random conditionally to $x$. Let $A$ be a model for $xy$ consisting of all extensions of $x$ (of the same length). This model looks good (in particular, it has negligible optimality deficiency). On the other hand, we may consider a standard model $B$ for $xy$ of the same (or smaller) complexity. It also has negligible optimality deficiency but looks unnatural. In this section we are interested in the following question: how can we formally distinguish good models like $A$  from bad models like $B$?  We will see that at least for some strings $u$ the value $\KT(A\cnd u)$ can be used to distinguish between good and bad models for $u$. (Indeed, in our example $\KT(A\cnd xy)$ is small, while  $\KT(B\cnd xy)$ can be large.)

\begin{definition}
A set $A\ni x$ is an \emph{$\varepsilon$-strong model} (or \emph{statistic}) for a string $x$ if $\KT(A\cnd x)\le\varepsilon$.
\end{definition}

For instance, the model $A$ discussed above is an $O(\log n)$-strong model for $x$. On the other hand, we will see later that, if $y$ is chosen appropriately, then  no standard description $B$ of the same complexity and log-cardinality as $A$ is  an $\varepsilon$-strong model for $x$, even for $\varepsilon=\Omega(n)$.

The following proposition explains the meaning of a strong model by providing an equivalent definition.
%We start with one more property of strong statistics. 
A finite family of sets $\mathcal{A}$ is called \emph{partition} if 
for every $A_1, A_2 \in \mathcal{A}$ 
we have $A_1 \cap A_2 \not= \varnothing \Rightarrow A_1 = A_2$. 
For any partition we can define its complexity.  
The lemma states that strong  statistics are those that belong 
to simple partitions.

\begin{proposition}\label{part}
Assume that $A$ is a model for $x$ that belongs to a partition of complexity $\eps$. Then 
$A$ is an $\eps+O(1)$-strong model for $x$.
Conversely, assume that $A$ is an $\eps$-strong statistic for а string $x$ of length $n$. Then there is  a partition $\mathcal{A}$ 
of complexity at most  $\eps + O(\log n)$ and a model $A'\in\mathcal A$ for $x$ such that: $\# A' \le  \# A$ and 
%1) $A'$ is an $\eps + O(\log n)$-strong statistic for $x$,
both $\CT(A \cnd A')$ and $\CT(A' \cnd A)$ are at most $\eps + O(\log n)$. 
\end{proposition}
\begin{proof}
Assume that $A$ is a model for $x$ that belongs to a partition $\mathcal A$ of complexity $\eps$. 
The program that maps every given string $x'$ to the set $A'\in \mathcal A$ which $x$ belongs to (if $x'$ belongs to no set in $\mathcal A$,
then the program maps it, say, to the empty set) is total and has length at most $\eps+O(1)$.

Conversely, assume that $A$ is an $\eps$-strong statistic for a string $x$ of length $n$.
Then there is a total program $p$ such that 
$p(x) = A$ and $|p| \le \eps$.  

%We will use the same construction as in Lemma \ref{l1}. 
Consider the set $X$ of all strings $x'$ with $x'\in p(x')$.
Obviously, $x\in X$.
Partition $X$ according to the value of $p(x')$:
strings $x'$ and $x''$ are in the same set of the partition if
$p(x')=p(x'')$. The constructed partition $\mathcal A$ has complexity at most  $\eps + O(\log n)$. 
It includes  the set $A'=\{x'\in X\mid p(x')=A\}$, which 
includes $x$ and can be obtained from $x$ by a total program
of length at most $\eps+O(\log n)$: that program maps a given string $x'$ to the set $\{x''\in X\mid p(x'')=p(x')\}$.

Since $A'\subset A$, we have $ \# A' \le \# A$. It remains to show that
that both $\CT(A \cnd A')$ and 
$\CT(A' \cnd A )$ are less than 
$|p| + O(\log n) = \eps + O(\log n)$.
Indeed, $A'$ can be obtained from $A$ by a total program of length $|p| + O(\log n)$ that maps  
a given set $B$ to $\{x'\in X\mid p(x')=B\}$. On the other hand, $A$ can be obtained from $A'$
by a total program of length $|p| + O(1)$ that for a given set $B$ picks any element $x'$ from $B$  and computes $p(x')$.
If $D$ is empty then that program outputs, say, the empty set. Recall that $A'$ is
non-empty, as it includes $x$.  
\end{proof}

Strong models satisfy an analog of Proposition~\ref{prop:description-shift} (the same proof works):

\begin{proposition}
     \label{prop:description-shift-1}
Let $x$ be a string and $A$ be an $\varepsilon$-strong model for $x$. Let $i$ be a non-negative integer such that $i\le \log\# A$. Then there exists an $\varepsilon+O(\log i)$-strong model $A'$ for $x$ such that $\#A' \le \# A/ 2^i$ and $\KS(A')\le \KS(A)+i + O(\log i)$.
\end{proposition}

To take into account the strength  of models, we may consider the set
$$
P_x(\varepsilon)=\{(i,j)\mid x\text{ has an $\varepsilon$-strong $(i*j)$-description}\}.
$$
Obviously, we have
$$
P_x(\varepsilon)\subset P_x = P_x(n+O(1))
$$
for all strings $x$ of length $n$ and for all $\varepsilon$.

If the set $P_x(\varepsilon)$ is not much smaller than $P_x$ for a reasonably small $\varepsilon$, we will say that $x$ is a ``normal'' string and otherwise we call $x$ ``strange''.  More precisely, a string $x$ is called $(\varepsilon,\delta)$-\emph{normal} if $P_x$ is in $\delta$-neighborhood of $P_x(\varepsilon)$. Otherwise, $x$ is called \emph{$(\varepsilon,\delta)$-strange}.

It turns out that there are $\sqrt{n\log n},O(\log n)$-normal strings with any given set $P_x$ that satisfies the conditions of Theorem~\ref{stat-any-curve}. On the other hand, there are $\Omega(n),\Omega(n)$-strange strings of length $n$. We are going to state these facts accurately.

\begin{thm}[\cite{milovanov-csr}]\label{stat-any-curve-1}
Let $k\le n$ be two integers and let $t_0 > t_1 > \ldots > t_k$  be a strictly decreasing sequence of integers such that $t_0\le n$ and $t_k=0$. Then there exists a string $x$ of complexity $k+O(\sqrt{n\log n})$ and length $n+O(\log n)$ for which the distance between both sets $P_x$ and $P_x(O(\log n))$ and the set  $T=\{( i,j)\mid (i\le k)\Rightarrow (j\ge t_i)\}$ is at most $O(\sqrt{n\log n})$.
\end{thm}

\begin{proof}
Consider the family $\mathcal A$ of all cylinders, i.e., the family of all the sets of the form
$\{ur\mid l(r)=m\}$ for different strings $u$ and natural numbers $m$. Sets from this family have the following feature: if $A\ni x$ then $A$ is  an $O(\log n)$-strong model  for $x$. Hence for all strings $x$ we have  $P_x^{\mathcal A}=P_x^{\mathcal A}(O(\log n))$.

By Theorem~\ref{thm:family-curve} and Remark~\ref{rem:family} there is a string $x$ of length $n+O(\log n)$  and complexity $k+O(\sqrt{n\log n})$ such that all sets $P_x, P_x^{\mathcal A},T$ are  $O(\sqrt{n\log n})$-close to each other. Hence all the three sets are close to the set $P_x^{\mathcal A}(O(\log n))$ as well. As the set $P_x(O(\log n))$ includes the latter set and is included in $P_x$, all the three sets are close to the set $P_x(O(\log n))$ as well.
\end{proof}

The next theorem from~\cite{ver2015} shows that ``strange'' strings do exist.
%\footnote{In this section we omit some proofs; see the original papers and the \texttt{arxiv} version of this paper.}

\begin{thm}\label{t1}
Assume that natural numbers $k,n,\eps$ satisfy the inequalities
$$
O(1)\le\eps\le k\le n.
$$
Then there is a string $x$ of length $n$ and complexity $k+O(\log n)$ such that the sets $P_x$ and $P_x(\eps)$ are $O(\log n)$-close to the sets shown on Fig.~\textup{\ref{f6}}. (The set $P_x$ is to the right of the dashed line. The set $P_x(\eps)$ is to the right of the solid line. The difference between the sets
has the shape of a parallelogram.)
\end{thm}

\begin{figure}[ht]
\begin{center}
\includegraphics[scale=1]{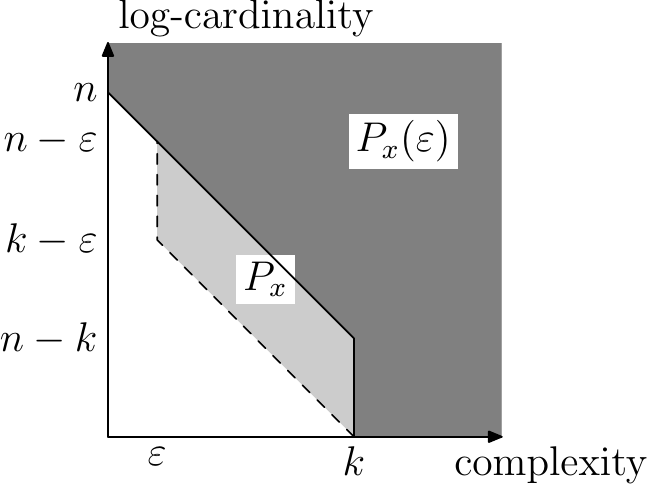}
\caption{The sets $P_x$ and $P_x(\eps)$ for the strange string from Theorem~\ref{t1}, with $O(\log n)$-precision. The set $P_x$ is to the right of the dashed line. The set $P_x(k)$ is to the right of the solid line.}\label{f6}
\end{center}
\end{figure}

We will prove this theorem later in Section~\ref{s:strange}. To illustrate its statement
let $k=2n/3$ and $\eps=n/3$ in Theorem~\ref{t1}. Then the sets $P_x$ and $P_x(n/3)$ are almost $n/3$-apart, since the point $(n/3,n/3)$ is in the $O(\log n)$-neighborhood of $P_x$ while all points from $P_x(n/3)$ are $(n/3-O(\log n))$-apart from $(n/3,n/3)$ (in $l_\infty$-norm). Thus the string $x$ is $(n/3,n/3-O(\log n))$-strange.

\medskip
Recall that we have introduced the notion of a strong model to separate good models from bad ones. We will now present some results that justify this approach.  The following theorem by Milovanov states, roughly speaking, that there exists a string $x$ of length $n$ and a strong model $A$ for $x$ such that the parameters (complexity, log-cardinality) of every strong \emph{standard} model $B$ for $x$ are $\Omega(n)$-far from those of $A$.

\begin{thm}\label{thm:separation}
For some positive $c$ for almost all $k$ there is a string $x$ of length $n=4k$ whose profile $P_x$ is $O(\log n)$-close to the gray set shown on Fig.~\textup{\ref{f4}}
\begin{figure}[h]
\begin{center}
\includegraphics[scale=1]{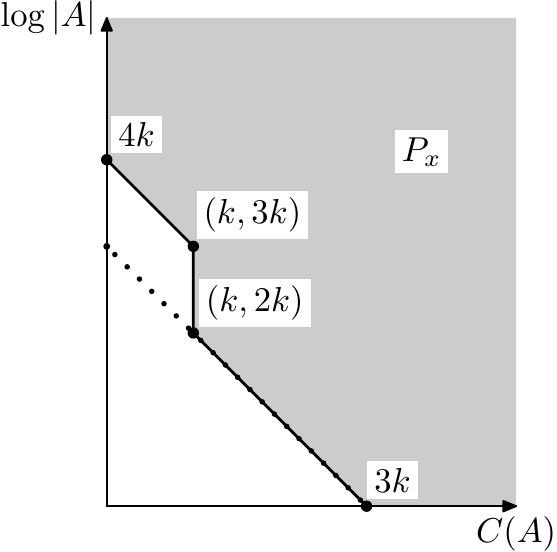}
\end{center}
\caption{The profile $P_x$ of a string $x$ from Theorem \ref{thm:separation}.}
\label{f4}
\end{figure}
such that
\begin{itemize}
\item there is an $O(\log n)$-strong model $A$ for $x$ with complexity $k+O(\log n)$ and log-cardinality $2k$ \textup(that model witnesses the point $(k,2k)$ on the border of $P_x$\textup), but
\item for every $m\ge \KS(x)$ and for every simple enumeration of strings of complexity at most $m$ the standard model $B$ for $x$ obtained from that enumeration is either not strong for $x$ or its parameters are far from the point $(k,2k)$. More specifically, if $B$ is a
%n $\varepsilon$-strong 
model for $x$ obtained from an enumeration provided by some program $q$, then 
at least one of the values 
$\KT(B\cnd x),\KS(q),|\KS(B)-k|,|\log\#B-2k|$ is larger than $ck$.
\end{itemize}
\end{thm}

We will prove this theorem in Section~\ref{s:normal}.

\subsection{Properties of strong models}

Once we have decided that non-strong descriptions are bad, it is natural to restrict ourselves to strong descriptions with negligible randomness deficiency (and hence negligible optimality deficiency). 

Consider some $n$-bit string $x$. Assume that $A$ is an $\varepsilon$-strong  description for $x$. Let $u$ be the ordinal number of $x$ in $A$ with respect to some fixed order. Then  $\KT(x\cnd A,u)=O(1)$ and  $\KT(A,u\cnd x)\le \varepsilon+O(1)$ (the latter inequality holds since $\KT(A\cnd x)\le\varepsilon$ and $u$ can be easily found when $x$ and $A$ are known). 
That is, $x\stackrel{\varepsilon+O(1)}{\longleftrightarrow}(A,u)$. As the pairs $(A,u)$ are naturally partitioned into classes, in this way we obtain and alternative proof
of  Proposition~\ref{part}: the image of that partition under the total mapping witnessing $(A,u)\stackrel{O(1)}{\rightarrow}x$ satisfies the second claim of Proposition~\ref{part}.

Assume further that the randomness deficiency of $x$ in $A$ is at most $\varepsilon$. 
As $u$ is random and independent of $A$ (with precision $\varepsilon$; note that $\KS(u\cnd A)\approx \KS(x\cnd A)\ge \log\#A-\varepsilon$), the sets  $Q_{A,i}$ and $Q_{A}$ are $\varepsilon+O(\log n)$-close (Theorem~\ref{prop:add-noise}). On the other hand, the sets  $Q_{A,u}$ and $Q_x$ are $\varepsilon+O(1)$-close by Proposition~\ref{prop:equivalence}. Thus we obtain the first property of strong models:

\begin{proposition}\label{prop:upward}
If both $\KT(A\cnd x)$ and $\log\# A-\KS(x\cnd A)$ are at most  $\varepsilon$, then the sets  $Q_x$ and $Q_{A}$ are $O(\varepsilon+\log l(x))$-close.
\end{proposition}

Assume that $A$ is an $\varepsilon$-strong model for $x$ with negligible randomness deficiency and~$\varepsilon$; for simplicity we ignore these negligible quantities in the sequel. Assume that  $A$ is normal in the sense described above. Then the string $x$ is normal as well. Indeed, for every pair $(i,j)\in P_x$ with $i\le \KS(A)$ the pair $(i,j-\log\# A)$ is in $P_{A}$ (Theorem~\ref{rem:add-noise}; note that $x$ is equivalent to $(A,u)$ and $u$ is random with condition $A$) and hence there is a strong  $(i*(j-\log\# A))$-description $\mathcal{B}$ for $A$. Consider the ``lifting'' of $\mathcal{B}$, that is, the union of all sets from $\mathcal{B}$ that have approximately the same size as $A$. It is a strong $(i*j)$-description for $x$.

It remains to consider pairs $(i,j)\in P_x$ where $i\ge \KS(A)$. Then $i+j\ge \KS(A)+\log\# A=\KS(x)$. Hence the subset of $A$ consisting of all strings $x'$ whose ordinal number in $A$ has the same $i-\KS(A)$ leading bits as the ordinal number of $x$, is a strong $(i*j)$-description for $x$.

It turns out that for minimal models the converse is true as well. A model $A$ for $x$ is called \emph{$(\delta,\varkappa)$-minimal} if there is no model $B$ for $x$ with $\KS(B)\le \KS(A)-\delta$ and $\delta(x,B)\le \delta(x,A)+\varkappa$. The smaller $\delta$ and the larger $\varkappa$ are the stronger the property of $(\delta,\varkappa)$-minimality is.

Recall that an $\varepsilon$-sufficient statistic is a model whose optimality deficiency is smaller than~$\varepsilon$.

\begin{thm}[\cite{milovanov-csr}]
\label{thm:hereditary}
For some value $\varkappa=O(\log n)$ the following holds. Assume that $A$ is an $\eps$-strong 
and $\varepsilon$-sufficient 
statistic for an $(\varepsilon,\varepsilon)$-normal string $x$ of length $n$. Assume also that $A$ is a $(\delta,\varkappa)$-minimal model for $x$. Then $A$ is $(O(\delta+(\varepsilon+\log n)\sqrt n),O(\delta+(\varepsilon+\log n)\sqrt n))$-normal.
\end{thm}

\begin{remark}
In the original theorem from~\cite{milovanov-csr} it is claimed only that $A$ is $(O((\delta+\varepsilon+\log n)\sqrt n),O((\delta+\varepsilon+\log n)\sqrt n))$-normal.
However, the arguments from~\cite{milovanov-csr} prove the theorem as stated here.
\end{remark}

This theorem will be proved in Section~\ref{s:hereditary}.
The next theorem states that the total conditional complexity of any strong, sufficient and minimal statistic for $x$ conditioned by any other sufficient statistic for $x$ is negligible.

\begin{thm}[\cite{ver}]\label{thm:step-wise}
For some value $\varkappa=O(\log n)$ the following holds. Assume that $A,B$ are $\varepsilon$-sufficient statistics for a string $x$ of length $n$. Assume also that $A$ is  a $(\delta,\varepsilon+\varkappa)$-minimal statistic for $x$.
Then $\K(A\cnd B) =O(\delta+\log n)$. Moreover, if, additionally, $A$ is an $\varepsilon$-strong model for $x$,
then $\KT(A\cnd B) =O(\varepsilon+\delta+\log n)$. 
\end{thm}

This theorem can be interpreted as follows: assume that we have removed some noise from a given data string $x$ by finding its description $B$ with negligible optimality deficiency. Let $A$ be any ``ultimately denoised'' model for $x$, i.e., a minimal model for $x$ with negligible optimality deficiency. Then $\KS(A\cnd B)$ is negligible (the first part of Theorem~\ref{thm:step-wise}). 
Hence to obtain the ``ultimately denoised'' model for $x$ we do not need $x$: any such model can be obtained from $B$ by a short program. The second part of 
Theorem~\ref{thm:step-wise} shows that any such \emph{strong} model $A$ can be obtained from $B$ by a short \emph{total} program.

\subsection{Strange strings}\label{s:strange}
In this section we prove the existence of strange strings (Theorem \ref{t1}). Then we prove that there are many such strings (Theorem \ref{t3}).

Let $A\mapsto [A]$ denote a computable bijection from the family of finite sets to the set of binary strings. This bijection will be used 
to encode finite sets of strings.  
\begin{lemma}\label{l1}
Assume that $A$ is an $\varepsilon$-strong statistic for a string $x$
of length $n$. Let $y=[A]$ be the code of $A$.
Then $y$ has an $(\varepsilon+O(\log n)),n$-description.
\end{lemma}
\begin{proof}
Let $p$ be a string of length at most $\varepsilon$ 
such that $p(x)=y$ and $p(x')$ is defined for all strings $x'$.
Consider the set $\{p(x') \mid x'\in    \{0,1\}^n \}$.
Its cardinality is at most $2^n$ and complexity at most $\varepsilon + O(\log n)$.
\end{proof} 
\begin{proof}[Proof of Theorem \ref{t1}]
To prove the theorem
it suffices to find a set $A\subset \bo^n$ 
with\\
(a) $\K(A)\le \eps+O(\log n)$, $\log \# A\le k-\eps$\\
which is not covered by the union of sets from the union of the following
three families:\\
(b) the family $\calb$ consisting of all sets  
$B\subset\bo^*$ with   with $\K(B)\le \eps$,  $\log \#B\le n-\eps-4$,\\
(c) the family $\calc$ consisting of all sets $M$ with 
$\K(M)\le k$,  $\log \# M\le n-k-4$
whose code $[M]$ has a $(\eps+O(\log n)),n$-description, and\\
(d) the family $\cald$ consisting of all singletons sets $\{x\}$ where $\K(x)< k$.\\

Indeed, assume that we have such set $A$.
As $x$ we can take any non-covered string in $A$.
Notice that item (a) implies that 
the complexity of $x$ is at most $k+O(\log n)$, and item (d)
implies that it is at least $k$. 
Thus the membership of the pair $(k+O(\log n), O(\log n))$ in $P_x(\eps)$
is witnessed by the singleton $\{x\}$, provided $\eps$ is greater than
the constant from the inequality $\CT(\{x\}\cnd x)=O(1)$. 
The membership of the pair $(O(\log n), n)$ in $P_x(\eps)$ and 
$P_x$ will be witnessed by the set 
of all strings of length $n$, provided $\eps$ is greater than
then constant from the inequality $\CT(\{0,1\}^n\cnd x)=O(1)$. 
The membership of the pair $(\eps+O(\log n), k-\eps)$ in $P_x$ will be witnessed by the set $A$. 

The upper  bounds for $P_x(\eps)$ and $P_x$ follow from (b), (c),  and Lemma~\ref{l1}.
Indeed, item (b) implies that the pair $(\eps,n-\eps-O(1))$ is not in $P_x$ (hence by
Proposition~\ref{prop:description-shift-1} for any $i\le \eps$ the pair   $(i-O(\log n),n-i)$ is not in $P_x$ either).
Item (c) implies that the pair $(k-O(\log n), n-k-O(1))$ is not in $P_x(\eps)$.

Let us show that there is a set $A$ of $n$-bit strings satisfying (a), (b), (c) and (d).
A direct counting reveals
that the family $\calb\cup\calc\cup\cald$ covers at most
$$
2^{\eps+1}2^{n-\eps-4}+2^{k+1}2^{n-k-4}+2^{k}\le2^{n-3}+2^{n-3}+2^{n-4}<2^{n-1}
$$
strings
and hence at least half of all $n$-bit strings are non-covered.
However we cannot let $A$ be any $2^{k-\eps}$-element  
non-covered set of $n$-bit strings, as in that case 
$\K(A)$ could be large.

We first show how to find $A$, as in (a), that is not covered
by $\calb\cup\cald$ (but may be covered by $\calc$). This may be done using the same technique as in the proof of Theorem~\ref{stat-any-curve}. 
To construct $A$ notice that both the families 
$\calb$  and $\cald$ can be enumerated given $k,\eps,n$ 
by running the universal machine $U$ in parallel on all inputs.
We start such an enumeration and construct $A$ ``in several attempts''.
During the construction 
we maintain the list of all strings covered by sets
from $\calb\cup\cald$ 
enumerated so far.
Such strings are called \emph{marked}.
Initially, no strings are marked and
$A$ contains the lexicographic first $2^{k-\eps}$ strings of length $n$.
Each time a new set $B\in\calb$ appears, all its elements 
receive a b-mark and we replace $A$ by any set 
consisting of $2^{k-\eps}$ yet non-marked $n$-bit strings.
Each time a new set $\{x\}$ in 
$D$ appears, the string $x$ receives a
d-mark, but we do not  immediately replace $A$.
However we do that when all strings in $A$ receive a d-mark,
replacing it by any set 
consisting of $2^{k-\eps}$ yet non-marked $n$-bit strings.
The above counting shows that such replacements are always possible.

The last version of $A$ 
(i.e. the version  obtained after the last set  
in $\calb\cup\cald$ have appeared)
is the sought set. 
Indeed, by construction $\# A = 2^{k-\eps}$ 
and $A$ is not covered by sets in $\calb\cup\cald$. 
It remains to verify that $\K(A)\le \eps+O(\log n)$. This follows 
from the fact that $A$ 
is replaced at most $O(2^{\eps})$ times, and hence 
can be identified by the number of its replacements and $\eps,k,n$
(we run the above construction of $A$
and wait until the given number of replacements are made).

Why is $A$ 
replaced at most $O(2^{\eps})$ times? The number of replacements
caused by appearance of a new set $B\in\calb$ is 
at most $2^{\eps+1}$. The 
number of strings with a d-mark is at most $2^{k}$
and hence $A$ can be replaced at most $2^{k}/2^{k-\eps}=2^{\eps}$ times 
due to receiving d-marks.

Now we have to take into account 
strings covered by sets from the family $\calc$. 
To this end modify the construction as follows: put a c-mark
on all strings from each set $C$ enumerated into 
$\calc$ and replace $A$ each time when all its elements
have received c or d marks (or when a new set is enumerated into $\calb$). 

However this modification alone is not enough.
Indeed, up to $\Omega(2^{n})$ strings may receive a c-mark,
and hence $A$ might be replaced up to $\Omega(2^{n-(k-\eps)})$ times due to c-marks.
The crucial modification is the following: each time $A$ is replaced, its new version 
is not just any set of $2^{k-\eps}$ non-marked $n$-bit strings but a carefully
chosen such set. 

To explain how to choose $A$ we 
first represent $\calc$ as an intersection of two families,
$\calc'$ and $\calc''$.
The first family $\calc'$ consists of all sets $M$ with
$\K(M)\le k$ and the second family $\calc''$
of all sets $C$ with $\log \# C \le n-k-4$
whose code $[C]$ has a $(\eps+O(\log n),n)$-description.
The first family is small (less than $2^{k+1}$ sets)
and the second family has only small sets (at most $2^{n-k-4}$-element sets)
and is not very large ($\# \calc'' =2^{O(n)}$). 
Both families can be enumerated given $\eps,k,n$
and, moreover, the sets from  $\calc''$ appear in the enumeration in 
at most $2^{\eps+O(\log n)}$ portions. Due to this property of 
$\calc''$ we can update $A$ each time a new portion of sets in 
$\calc''$ appears---this will increase the number of replacements of $A$
by $2^{\eps+O(\log n)}$, which is OK. 

The crucial change in construction
is the following: each time $A$ is replaced, its new version 
is \emph{a set 
of $2^{k-\eps}$ non-marked $n$-bit strings  that has at most $O(n)$ common strings 
with every set from the part of $\calc''$  
enumerated so far}. (We will show later that such a set always exists.)

Why does this solve the problem? 
There are two types 
of replacements of $A$:  
those caused by enumerating a new set in $\calb$ or a new bunch of sets in 
$\calc''$
and those caused by that all elements in $A$ have received c- or d-marks.
The number of replacement of 
the first type is at most $2^{\eps+O(\log n)}$. 
Replacements of the second type 
are caused by enumerating new singleton sets   
in $\cald$ and by enumerating new sets $C$ in 
$\calc'$ which were enumerated into $\calc''$
on earlier steps. Due to the careful choice of $A$,
when each such set $C$ appears in the enumeration of $\calc'$
it can  mark only 
$O(n)$ strings in the current version of $A$.
The total number of sets in $\calc'$ is at most $2^{k+1}$.
Therefore the total number of events ``a string in the current version of 
$A$ receives a c-mark'' is at most $O(n2^{k})$. 
The total number of d-marks is at most $2^{k}$.
Hence the number of replacements of the second type is at most 
$$
(O(n2^{k})+2^{k})/2^{k-\eps}=O(n2^{\eps}).
$$

Thus it remains to show that we indeed can always 
choose $A$, as described above. This will follow
from a lemma that says that in a large universe
one can always choose a large set 
that has a small intersection with every set from
a given small family of small sets.

\begin{lemma}
Assume that a finite family $\calc$ of 
subsets of a finite universe $U$ is given
and each set in $\calc$ has at most $s$ elements. 
If 
$$
\# \calc \binom{N}{t+1}\left(\frac{s}{|U|-t}\right)^{t+1}<1
$$
then there is an    
$N$-element set $A\subset U$ 
that has at most $t$ common elements with each set in $\calc$. 
\end{lemma}
\begin{proof}
To prove the lemma we use probabilistic method.
The first element $a_1$ of $A$ is chosen at random 
among all elements in $U$  
with uniform distribution,
the second element $a_2$ is chosen with uniform distribution
among the remaining elements and so forth.

We have to show that the statement of the theorem 
holds with positive probability.
To this end note that for every fixed $C$ in $\calc$ 
and for every fixed set of indexes 
$\{i_1,\dots,i_{t+1}\}\subset\{1,2,\dots,N\}$
the probability that \emph{all}  $a_{i_1},\dots,a_{i_{t+1}}$ 
fall in $C$ is at most 
$\left(\frac{s}{|U|-t}\right)^{t+1}$.
The number of sets of indexes as above is $\binom{N}{t+1}$.
By union bound the probability that a random set $A$ 
does not satisfy the lemma is  
upper bounded by the left hand side of the displayed inequality.
\end{proof}

We apply the lemma for 
$U$ consisting of all non-marked $n$-bit strings,
for $N=2^{k-\eps}$ and for $\calc$ consisting of
all sets in $\calc''$ appeared so far. Thus $s=2^{n-k-4}$, $\# U \ge 2^{n-1}$, $\# \calc =2^{O(n)}$, and we need to show that
for some $t=O(n)$ it holds
$$
2^{O(n)}\binom{2^{k-\eps}}{t+1}
\left(\frac{2^{n-k-4}}{2^{n-1}-t}\right)^{t+1}<1,
$$
which easily
follows from the inequality $\binom{2^{k-\eps}}{t+1}\le 2^{(k-\eps)(t+1)}$.
\end{proof}
Theorem~\ref{t1} does not say anything about how rare are strange strings.
Such strings are rare, as 
for majority of strings $x$ of length $n$ the set $\bo^n$ is a strong MSS for $x$.
A more meaningful question is whether such strings might 
appear with high probability in a statistical experiment. More specifically,    
assume that we sample a string $x$ in a given set $A\subset\bo^n$, where all elements  
are equiprobable. Might it happen that with high probability 
(say with probability  $99\%$) 
$x$ is strange? 
An affirmative answer to this question is given in the following

\begin{thm}\label{t3}
Assume that natural 
$k,n,\eps,\delta$ satisfy the inequalities 
$$
O(1)\le\eps\le k\le n,\quad \delta \le k-\eps.
$$
Then there is set $A\subset\bo^n$ of cardinality $2^{k-\eps}$ and complexity at most $\eps+O(\delta+\log n)$
such that all but $2^{k-\eps-\delta}$ its elements $x$ have complexity $k+O(\delta+\log n)$ and 
the sets $P_x$ and $P_x(\eps)$ are $O(\delta+\log n)$-close to the sets shown on Fig.~\ref{f6}.
\end{thm}

Theorem~\ref{t3}
is proved similarly to Theorem~\ref{t1}. The only difference is that we change $A$ each time
when at least $2^{k-\eps-\delta}$ strings in $A$ receive 
c- or d-marks. 
As the result, the number of changes of $A$ will increase $2^\delta$ times and the complexity of $A$ will increase by $\delta$.

\subsection{Strong sufficient statistics}
In this section we prove Theorem \ref{thm:step-wise}.

\emph{The proof of the first claim $\K(A\cnd B)\le O(\delta+\log n)$}.
Recall the notion of a standard description and  Proposition \ref{prop:better-std}. By that 
proposition there exist standard descriptions $A',B' \ni x$ that are ``better'' 
than $A,B$, respectively, i.e. 
\begin{align*}
\delta(x,A') &\le \delta (x, A) + O(\log n),\quad \K(A' \cnd A) = O(\log n),\\
\delta(x,B') &\le \delta (x, B) + O(\log n),\quad \K(B' \cnd B) = O(\log n).
\end{align*} 
 Let $i,j$ denote the complexities of $A',B'$, respectively.
 Ignoring $O(\log n)$ terms we have
 $$
 \K(A\cnd B)\le \K(A\cnd A')+\K(A'\cnd \Omega_i)+\K(\Omega_i\cnd \Omega_j)+\K(\Omega_j\cnd B')+\K(B'\cnd B).
 $$
  We will show that in the right hand side of this inequality, each term is at most $\delta+O(\log n)$.

 For the last term this directly follows from the construction of $B'$.
 The second and forth terms can be estimated as follows.
By Proposition \ref{prop:std-omega}, $A',B'$ are $O(\log n)$-equivalent to $\Omega_i,\Omega_j$, respectively,
thus the second and forth terms are $O(\log n)$.

To upper bound the third term, we first show that  
that $i\le j+\delta+O(\log n)$. By construction of $B'$ we have $\delta(x,B')\le \delta(x,B) + O(\log n)$. Since $B$ is
$\eps$-sufficient, this inequality implies that  $\delta(x,B')\le \eps+O(\log n)$.
On the other hand, $\delta(x,A)\ge -O(\log n)$ and thus $\delta(x,B')\le \delta(x,A) + \eps+O(\log n)$.
Choose the constant $d$ in the statement of the theorem so that this inequality implies that 
$\delta(x,B')\le \delta(x,A) + \eps+\varkappa$. We assume that 
 $A$ is  a $(\delta,\eps+\varkappa)$-minimal statistic for $x$ and hence $\K(B')\ge \K(A)- \delta$,
 that is, $j\ge \K(A)-\delta$. On the other hand, by construction of $A'$ we have $i=\K(A')\le \K(A)+O(\log n)$
 and hence $i\le j+\delta+O(\log n)$. 

By Proposition
\ref{prop:omega-equivalence} 
we have 
$$
\K(\Omega_i\cnd \Omega_j)\le \K(\Omega_i\cnd \Omega_{j+\delta+O(\log n)})+\delta+O(\log n).
$$
On the other hand,  
the  inequality $i\le j+\delta+O(\log n)$ and  Proposition
\ref{prop:omega-equivalence} imply that 
$$
\K(\Omega_i\cnd \Omega_{j+\delta+O(\log n)})=O(\log n).
$$

 It remains to estimate the first term.
 Repeating the argument from the beginning of the last but one paragraph we can show that $\K(A')\ge \K(A)- \delta$,
On the other hand, by construction of $A'$ we have $\K(A'\cnd A)=O(\log n)$.
By Symmetry of information this implies that $\K(A\cnd A')=\delta+O(\log n)$. 
 Indeed, ignoring $O(\log n)$ terms, we have:
 $$
  \K(A\cnd A')=\K(A)+\K(A'\cnd A)-\K(A')=\K(A)-\K(A')\le\delta.
  $$

\medskip

\emph{Proof of the second claim $\KT(A\cnd B)\le O(\eps+\delta+\log n)$} (provided $A$ is a strong models for $x$).
It suffices to we show that  $\KT(A\cnd B)$ is close to $\K(A\cnd B)$  provided $A$ is a strong model for $x$.

Recall that $A\mapsto [A]$ denotes a computable bijection from the family of finite sets to the set of binary strings which is used 
to encode finite sets of strings. 

On the top level the argument is as follows. Let $p$ be a program witnessing $\KT(A\cnd x)\le\eps$.
We first show that $A$ has the following feature: there are many strings $x'\in  B$  
with  $p(x')=[A]$. More specifically, at least $2^{-\K(A\cnd B)}$ fraction of $x'$
from $B$ have this property.  
At most $2^{\K(A\cnd B)}$ sets  $A'$ can have this feature, as each such $A'$ can be identified
by the portion of $x'\in  B$  with  $p(x')=[A']$.
Given $B$ and $p$ we are able to find a list 
of all such $A'$ by means of a short 
total program.  Given $B$, the set $A$ can be identified 
by $p$ and its index in that list.  

Let us proceed to the detailed proof. In the proof, we will ignore terms of order $O(\eps+\log k)$.
First we show that there are many  
$x'\in B$ with $p(x')=[A]$ (otherwise
$B$ could not be a sufficient statistic for $x$). 
Let 
$$
D=\{x'\in B\mid p(x')=[A]\}.
$$
We have 
$$
\K(D \cnd B)\le \K(A \cnd B),
$$
as 
$$
\K(D \cnd B)\le \K(D\cnd A)+\K(A \cnd B)
$$
and $\K(D\cnd A)=\K(D\cnd [A])\le \K(p)+O(1)\le\eps+O(1)$. 

Obviously, $D$ includes $x$, thus given $B$ and $p$ the string $x$ can be identified 
by its index in $D$. Therefore 
$$
\K(x \cnd B)\le\K(D \cnd B)+\log \# D\le\K(A \cnd B)+ \log \# D.
$$
On the other hand, $\K(x \cnd B)=\log \# B$, as
$B$ is $\eps$-sufficient. Hence
$$
\log \# D\ge\log \# B-\K(A \cnd B).
$$ 

Recall that we ignored terms of order $O(\eps+\log k)$.
Actually, we have shown that 
$\log \# D\ge l$ for some
$$
l=\log \# B-\K(A \cnd B)-O(\eps+\log k).
$$
Consider now all $A'$ for which the set $D'$ defined in a similar way has the same lower bound for its cardinality.
That is, consider sets $A'$ with  
$$
\log \# \{x'\in B\mid p(x')=[A']\} \ge l.
$$
Each such $A'$ can be identified
by the portion of $x'\in  B$  with  $p(x')=[A']$. 
Thus there are at most $2^{\K(A\cnd B)+O(\eps+\log k)}$ 
different such $A'$s.
Given $B$ and $\K(A \cnd B),p,\eps$ we are able to find
the list of all such $A'$s. The program that maps $B$ 
to the list of such $A'$s is obviously total. 
Therefore there is a $\K(A \cnd B)+O(\eps+\log k)$-bit  total program 
that maps $B$ to $A$ and
$\KT(A \cnd B)=\K(A \cnd B)+O(\eps+\log k)$.

\subsection{Normal strings and standard descriptions}\label{s:normal}
Here we prove Theorem \ref{thm:separation}, i.e. we 
exhibit an example of a normal string $x$ such that every standard description is not a strong sufficient minimal statistic for $x$.
Our string $x$ will be obtained from an antistochastic string in the sense of Section \ref{sub:trade-off}.

\begin{definition}
A string $x$ of length $n$ and complexity
$k$ is called \emph{$\eps$-antistochastic}
if for all $(m,l)\in P_x$ either
$m>k-\eps$, or $m+l>n-\eps$.  (The profile of an antistochastic string is shown on Fig.~\ref{f22}.)
\end{definition}
\begin{figure}[h]
\begin{center}
\includegraphics[scale=1]{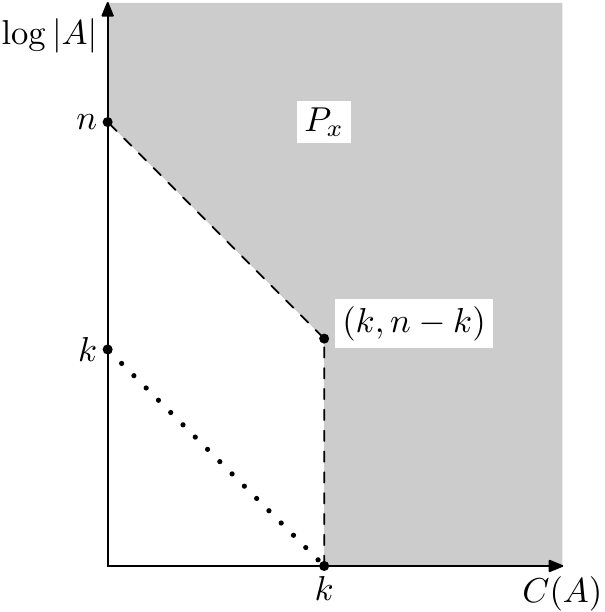}
\end{center}
\caption{The profile of an $\eps$-antistochastic string $x$ of length $n$ and complexity $k$ for a 
small $\eps$.}
\label{f22}
\end{figure}
By Theorem~\ref{stat-any-curve}  for every $k< n$ there exists an $O(\log n)$-antistochastic string of length $n$ and complexity $k + O(\log n)$.
It is easy to see that all antistochastic strings are normal.
Recall that a string $x$ is called $(\varepsilon,\delta)$-normal if $P_x$ is in $\delta$-neighborhood of $P_x(\varepsilon)$.

\begin{proposition} 
Let $x$ be an $\eps$-antistochastic string of length $n$ and complexity $k$. Then $x$ is $(O(\log n),O(\log n) + \eps)$-normal. 
\end{proposition}
\begin{proof}
To prove the claim, it suffices to construct for every point $(i,j)$ on the boundary of the set shown on Figure~\ref{f22}
an $O(\log n)$-strong $(i+O(\log n)*j)$-description $A$ for  $x$. 
If $i\ge k$  then
let $A=\{x\}$. 
Otherwise let 
$A$ be the set of all strings of length $n$ whose
first $i$ bits are the same as those of $x$. 
By construction $\K(A) \le i + O(\log n)$ and 
$\log \# A = n - i=j$.
\end{proof}

%In the proof we will use Proposition \ref{prop:upward}. We will reformulate it,  
%using Theorem \ref{thm:def-opt}, as follows: 
%\begin{proposition}
%\label{prop:upward_optimal}
%Let $x$ be a string of length $n$ and $A$ be an $\eps$-strong statistic for $x$ such that $d(x\cnd A)\le\eps$.
%Then for all $b \ge \log \# A$ we have 
%$$
%(a, b) \in P_x \Leftrightarrow (a + O(\eps + \log n), b - \log \# A + O(\eps + \log n)) \in P_{A} 
%$$
%and for $b\le \log \# A$ we have 
%$(a, b) \in P_x \Leftrightarrow a+b\ge \K(x)-O(\log n)$.
%\end{proposition}

\begin{proof}[Proof of Theorem \ref{thm:separation}]
Let $y$ be an 
$O(\log k)$-antistochastic string 
of complexity $k+O(\log k)$ and length $2k$, which exists by Theorem~\ref{stat-any-curve}.
Let $z$ be a string of 
length $2k$ such that $\K(z \cnd y) \ge 2k$. Finally, let $x=yz$
be the concatenation of $y$ and $z$. Obviously, $\K(x) = 3k$ with accuracy $O(\log k)$.

By Theorem~\ref{rem:add-noise}
%Proposition \ref{prop:upward_optimal} 
%the sets $Q_x$ and $Q_y$ are $O(\log k)$-close to each other and hence $Q_x$ is close to the gray set on Fig. \ref{f22}. 
%By Theorem \ref{thm:def-opt}
the set $P_x$ is $O(\log k)$-close to the gray set on Fig. \ref{f4}. From normality of $y$ 
it is not difficult to derive that  $x$ 
is $O(\log k),O(\log k)$-normal.

Let $A=\{yz'\mid l(z')=2k\}$. Then $A$ is an $O(\log k)$-strong statistic for $x$. 
Now we need to show that every standard description whose parameters (complexity, log-cardinality) are close to those of $A$
is not strong.
Let $q$ be a program enumerating all strings of complexity at most $m$. Let $B$ be the standard model for $x$ obtaining from this enumeration. 
%Assume that $B$ is an $\KT(B\cnd x)$-strong model for $x$. 
We have to show that  for some positive $c$ for almost all $k$ we have
\begin{equation}
\label{req}
\min\{\KT(B\cnd x),\KS(q),|\KS(B)-k|,|\log\#B-2k|\}\ge c k.
\end{equation}

Fix some small positive $c$. For the sake of contradiction 
assume that for infinitely many $k$ there are $m,q,B$ such that 
the inequality (\ref{req}) does not hold. 

On the top level, the argument is as follows. We have assumed 
$\K(B)$ and $\log\#B$ are close to $k$ and $2k$, respectively. The shape of $P_x$
guarantees that in this case $B$ is a sufficient minimal statistic for $x$. 
We have assumed also that $\KT(B\cnd x)$ is small, that is $B$ is a strong statistic for
$x$. Hence we can apply Theorem \ref{thm:step-wise} to $B$ and $A$ and
conclude that $\CT(B\cnd A)$ is small. As $A$ is strongly equivalent to $y$,
the total conditional complexity $\CT(B\cnd y)$ is small as well.
On the other hand we will show that the total 
complexity of any standard statistic $B$ obtained from an enumeration
of strings of complexity at most $m$ conditional to any string $y$
is larger than $\min\{\K(B),m-l(y)\}$. In our case $m$ is at least $\K(x)\approx 3k$,
$l(y)=2k$ and $\K(B)\approx k$ thus $\CT(B\cnd y)\ge k$ with accuracy $O(\log k)$,
which is a contradiction, if $c$ is small enough.

Let us proceed to the detailed proof.
Let us show first that $m=O(k)$. Recall that 
$B$ is a standard description obtained from a list of complexity at most $m$
enumerated by program $q$.  We have seen that in this case 
$$
\KS(B)+\log\#B=m+O(\K(q)+\log m). 
$$
Hence $m\le (k+c k)+(2k+c k)=O(k)$ with accuracy $O(\K(q)+\log m)$.
We have assumed that $\K(q)<c k=O(k)$ and therefore $m=O(k)$.

%We may assume that $c<1/2$ and hence $\log\#B < 5k/2$. As the profile of $x$ is $O(\log k)$-close to the gray set shown on Fig.~\ref{f4},
%the latter inequality implies that $\K(B) \ge k - O(\log k)$. 

Now we will apply Theorem~\ref{thm:step-wise}. 
Recall its statement:
\begin{quote}
For some value $\varkappa=O(\log n)$ the following holds. Assume that $A,B$ are $\varepsilon$-sufficient statistics for a string $x$ of length $n$. 
Assume also that $B$ is  a $(\delta,\varepsilon+\varkappa)$-minimal statistic for $x$ and 
$B$ is an $\varepsilon$-strong model for $x$.
Then $\KT(B\cnd A) =O(\varepsilon+\delta+\log n)$. 
\end{quote}
Fix such $\varkappa=O(\log n)=O(\log k)$ (recall that we have $n=4k$). 

The models $A,B$ are $\eps$-sufficient and $B$ is $\eps$-strong for 
$$
\eps=\max\{\delta(x, B),\delta(x,A),\CT(B\cnd x)\}\le 2ck+O(\log k).
$$
Besides, the shape of $P_x$ guarantees that $B$ is 
$\delta,t$-minimal for
\begin{align*}
\delta&=|\K(B)-k|+O(\log k)\le ck+O(\log k), \\
t&=k-\delta(x,B)-O(\log k)\ge k-2ck-O(\log k).
\end{align*}
If $c<1/4$ then for all large enough $k$ we have 
$t\ge \eps+\varkappa$, thus $B$ is $\delta,\eps+\varkappa$-minimal and hence we can apply Theorem~\ref{thm:step-wise}.  
Therefore
$\CT(B \cnd A)\le O(\eps+\delta + \log k)=O(ck+\log k)$
and hence  $\CT(B \cnd y) = O(c k + \log k)$
%Note that the inequality  $ \KT(B\cnd x) + t + \lambda + \K(q) \ge \Omega(k)$ is equivalent to  (\ref{req}). 

The contradiction is obtained from the latter inequality and the following lemma. 
\begin{lemma}\label{l4}
Assume that $B$ is any standard model obtained from an enumeration of strings of complexity at most $m$
by a program $q$ and $y$ is any string. Then 
$\CT(B\cnd y)\ge \min\{\K(B),m-l(y)\}-O(\K(q)+\log n)$, where $n=\max\{l(y),m\}$.
\end{lemma}
\begin{proof}
Denote by $b$ the lexicographic first element in $B$. We can estimate the total complexity of $b$ conditional to
$y$ as follows: 
$$
\CT(b \cnd y) \le \CT(b\cnd B)+\CT(B\cnd y)+O(\log n).
$$
The first term here is $O(1)$ and  
hence 
$\CT(b \cnd y)\le  \CT(B\cnd y)+O(\log n)$.
Denote by $p$ a total program of this  length   
transforming $y$ to  $b$ and consider the set 
$$
D:= \{p(y') \mid l(y')=l(y) \}.
$$ 
Obviously $\K(D)\le |p|+O(\log n)$.
Thus it suffices to show that 
$\K(D)\ge \min\{\K(B),m-l(y)\}-O(\log n)$.

Every  element from $D$ has complexity  
at most $\K(D) + \log \#D + O(\log n) \le \K(D)+l(y)+O(\log n)$.
If the right hand side of this inequality is larger than $m$ then $\K(D)\ge m-l(y)-O(\log n)$
and the lemma follows.

Therefore we may assume that  all   
elements from $D$ have complexity  
at most $m$.
Run the program
$q$ until it prints all elements from $D$. 
Since $b \in D$ and $b\in B$, there are at most $2\# B$ elements of complexity $m$ 
that have not been printed yet (all the elements of $B$ that are enumerated after $b$ and less than $\# B$ 
that are enumerated after all elements of $B$). So, we can find the list of all strings of 
complexity at most $m$ from $D$, $q$ and some extra
$\log \#B+1$ bits. 
Since this list has complexity  $m - O(\log m)$,
we get
$$
\K(D) + \K(q) + \log \# B \ge m - O(\log m).
$$

Recall that for any standard model obtained from enumeration of strings of 
complexity at most $m$ the sum of its complexity and log-cardinality is at most $m$ plus the complexity
of the enumerating program. As $B$ is a standard model, we have  
 $$ 
 \K(B) + \log \# B \le m + O(\K(q) + \log m).
 $$ 
 Summing these inequalities we get 
$$ 
\K (D) + O( \K (q) + \log m) \ge \K (B) .
$$ 
\end{proof}
Lemma~\ref{l4} thus implies that
$$
\min\{\K(B),m-2k\}-O(\K(q)+\log k)\le \CT(B\cnd y)= O(c k +\log k)
$$ 
Since $\K(B)\ge k-c k$ and $m\ge \K(x)=3k-O(\log k)$, we have
$$
k\le O(c k +\log k).
$$
The constant hidden in this $O(c k +\log k)$ notation is an absolute constant, call it $C$.
Therefore  if $c<1/C$ we will obtain a contradiction.
\end{proof}

\subsection{A strong sufficient minimal model for a normal string is normal}\label{s:hereditary}
Here we prove Theorem \ref{thm:hereditary}. 
We will need two lemmas. We first state the lemmas in an informal way
and then we  outline the proof of the theorem. 
Then we provide rigorous formulations of lemmas and the rigorous proof  of the theorem.

By Propositions \ref{prop:better-std} and \ref{prop:std-omega}
for every $A \ni x$ there is $B \ni x$ (a standard model) 
such $\K(B\cnd \Omega_{\K(B)})\approx0$ and parameters of
$B$  are not worse than those of $A$.
We will need a similar result for normal strings and for strong models.

\begin{lemma}[informal]
\label{prop:min_hereditary}
Assume that $A$ is a minimal statistic for some string. 
Then $ \K(\Omega_{\K (A)} \cnd A)\approx0$.
\end{lemma}

  \begin{lemma}[informal]
\label{lch}
For every normal string $x$ and 
every model $A$ for $x$ there exists a \emph{strong} statistic $H$ for $x$
with (1) $\delta (x,H) \lesssim \delta (x,A)$,
(2) $\K (H \cnd \Omega_{\K (H)})\approx0$ and 
(3)  $\K (H) \le \K (A)$.
\end{lemma}

Now we sketch the proof of Theorem  \ref{thm:hereditary} using the lemmas.

\begin{proof}[A sketch of proof of Theorem \ref{thm:hereditary}]
Let $x$ be a normal string  and let $A$ be a strong 
sufficient 
and minimal statistic for $x$.
We have to prove that  the profile of $A$ is close to 
its strong profile. 

First we show that w.l.o.g we may assume that $A$ belongs to a simple partition.
Indeed, since $A$ is a strong statistic for $x$, we may apply Proposition \ref{part} to $A$ and $x$.
Let $\mathcal{A}$ be a simple partition and $A_1$ a model from
$\mathcal{A}$ which exists by Proposition \ref{part}. 
As the total conditional complexities $\CT(A_1 \cnd A)$ and  $\CT(A \cnd A_1)$
are small, the profiles 
of $A$ and $A_1$ are close to each other. This also applies to strong profiles. Therefore it suffices to
show that   $A_1$ is normal. 

%Therefore we will assume further that $A$ belongs to a to a simple partition $\mathcal{A}$.
%We have to prove that  the profile of $A_1$ is close to 
%its strong profile. 
To this end consider any model $\mathcal G$ (a family of sets) for $A_1$. Our goal is to find a strong model $\mathcal F$ for 
$A_1$ whose parameters (complexity, log-cardinality) are not worse than those of $\mathcal G$. To do that we will
find a model $M_1$ for $x$ such that 
\begin{equation}\label{eqm}
\begin{split}
&\CT(M_1\cnd A_1)\approx 0,\quad
\log\#(M_1\cap A_1)\approx\log\#A_1,\\
&\K(M_1)\le \K(\mathcal G),\quad \K(M)+\log\#M_1\le\K(\mathcal G)+\log\#\mathcal G+\log\#A_1.
\end{split}
\end{equation}
Then we will let 
$$
\mathcal F=\{A'\in \mathcal{A}\mid \log \#(A' \cap M_1) = \log \# (A_1 \cap M_1 )\}
$$
(Here and further by $\log$ we mean the integer part of the binary logarithm.)
The family $\mathcal F$ can be computed from $M_1$, $\mathcal A$ and $\log \#( A_1 \cap M_1)$.
As  $\mathcal{A}$ is simple, we conclude that $\K(\mathcal F) \le \K(M_1)\le \K(\mathcal G)$.

Moreover, $\CT(\mathcal F\cnd M_1)\approx 0$, as
the mapping 
$$
M'\mapsto \{A'\in \mathcal{A}\mid \log \#(A' \cap M') = \log \# (A_1 \cap M' )\}
$$
is total. Since  $\CT(M_1 \cnd A_1) \approx 0$, this implies 
that $\CT(\mathcal F\cnd A_1)\approx 0$, that is,
$\mathcal F$ is a strong model for $A_1$.

Finally $\log \# \mathcal F \le \log \# M_1 - \log \# A_1$, because $\mathcal A$
is a partition and thus it has few sets   
that have $\log \# (A_1 \cap M_1) \approx \log \# A_1$ common elements with $M_1$. 
Thus the sum of complexity and log-cardinality
of $\mathcal F$ is at most 
\begin{align*}
\K(M_1)+(\log\# M_1 - \log \# A_1)&\le
(\K(\mathcal G)+\log\#\mathcal G+\log \# A_1)-\log \# A_1\\
&=\K(\mathcal G)+\log\#\mathcal G.
\end{align*}
Hence $\mathcal F$ is a strong model for $A_1$ whose parameters (complexity, complexity + log-cardinality)
are not worse than required. From Proposition \ref{prop:description-shift-1} it follows that the strong profile of $A_1$ 
includes the point $(\K(\mathcal G),\log\#\mathcal G)$.

How to find a model $M_1$ for $x$ satisfying~\eqref{eqm}?
We will do that in three steps.
On the first step we construct a model $L$ for $x$ such that $\K(L)\le \K(\mathcal G)$ and $\log\#L\le \log\#\mathcal G+\log\#A_1$.
More specifically, we let 
$$
L=\bigcup\{A'\in\mathcal G\mid \log\# A'=\log\#A_1\}.
$$
By construction we have $\K(L)\le \K(\mathcal G)$ and
$\log\#L\le \log\#\mathcal G+\log\#A_1$  (see Fig. \ref{f676}).
\begin{figure}[h]
\begin{center}
\includegraphics[scale=1]{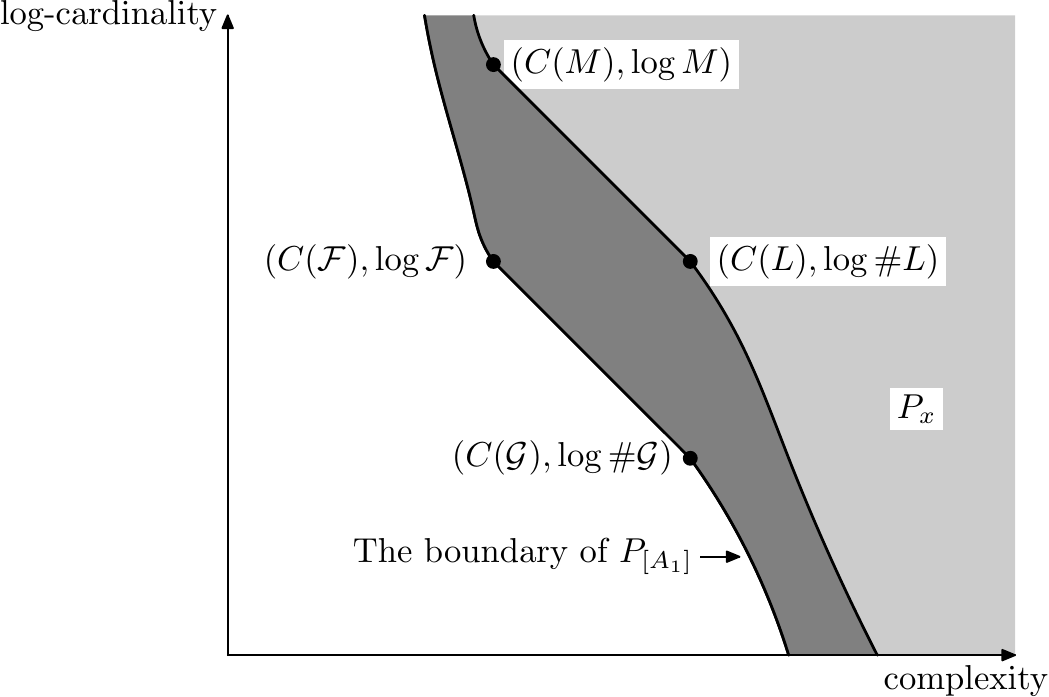}
\end{center}
\caption{The picture shows parameters
  (complexity, log-cardinality)
  of models $\mathcal G,\mathcal F$ (for $A$)
  and $L,M$ (for $x$).}
\label{f676}
\end{figure}

On the second step we find a strong model $M$ for $x$ whose 
parameters (complexity, complexity + log-cardinality) are not worse than those of $L$ and such that 
$\K(M\cnd \Omega_{\K(M)})\approx0$; such model exists by Lemma~\ref{lch}.  
On the third step we find a model $M_1$ with the same parameters as those of $M$
that belongs to a simple partition $\mathcal M$ and such that $\CT(M_1\cnd M)\approx0$; such model exists by
Proposition~\ref{part}.

By construction we have $\K(M_1)\le \K(M)\le\K(L)\le \K(\mathcal G)$.
We have to show that $\CT(M_1\cnd A_1)\approx 0$ 
and $\log\#(M_1\cap A_1)\approx\log\#A_1$.

To prove the first claim we show first that $\K(M \cnd A_1)$ is small. 
Indeed, from $A_1$ we can compute $A$ (by a short program),
from $A$ we can compute $\Omega_{\K(A)}$ (Lemma~\ref{prop:min_hereditary}),
from $\Omega_{\K(A)}$ we can compute 
$\Omega_{\K(M)}$ (indeed, Lemma \ref{lch} guarantees that $\K(M)\le\K(L)\le \K(\mathcal G)$ and w.l.o.g. we may assume that $\K(\mathcal G)\le \K(A)$ as 
$\mathcal G$ is a model for $A$) and then compute $M$ 
(as $\K(M \cnd \Omega_{\K(M)}) \approx 0$ by Lemma \ref{lch}). 
Proposition~\ref{part} guarantees
that $M_1$ is simple given $M$. Since $\K(M_1 \cnd M) \approx 0$ and $\K(M \cnd A_1) \approx 0$, we have 
$\K(M_1 \cnd A_1)\approx0$ as well.

To show the stronger equality $\CT(M_1\cnd A_1)\approx 0$
consider the model
$A_1\cap M_1$ for $x$. We claim that its cardinality cannot be much less than $A_1$. 
Indeed, since $\K(M_1 \cnd A_1)\approx0$, we have $\K (A_1\cap M_1) \le \K(A_1)$.
Obviously, $\# (A_1\cap M_1) \le \# A_1$.
Therefore the parameters of  $M_1 \cap A_1$ are not worse then those of
$A_1$. The model $M_1 \cap A_1$  cannot have much better parameters than $A_1$, since $A_1$ is a sufficient statistic for $x$ (recall that the parameters of $A_1$ are not worse than
those of $A$ and $A$ is assumed to be a sufficient statistic for $x$). 
Hence $\log \# (A_1\cap M_1)  \approx \log \# A_1$.

Recall that $M_1$ belongs to a simple partition $\mathcal M$.  
The model $M_1$ can be computed
by a total program from $A_1$ and its index among all $M'\in\mathcal M$ with
$\log \# (A_1\cap M') \approx \log \# A_1$. As $\mathcal M$ is a partition,
there are few such sets $M'\in \mathcal M$. Hence  $\CT(M_1 \cnd A_1)\approx 0$.
\end{proof}

 Now we provide rigorous formulations and proofs of the used lemmas. 
 %From those formulations it will be easily seen that the above sketch indeed
 %proves Theorem~\ref{thm:hereditary}.

\smallskip
\textbf{Lemma \ref{prop:min_hereditary}} (rigorous).
For some $\kappa =O(\log n)$ the following holds.
Assume that $A$ is 
%an $\eps$-sufficient statistic for a string $x$ of length $n$. Assume also that $A$ is 
a $(\delta, \kappa)$-minimal statistic for a string $x$ of length $n$.
Then $ \K(\Omega_{\K (A)} \cnd A) = O(\delta + \log n)$.

\begin{proof}
Let $B$ be a standard model for $x$ that is an improvement of $A$ existing  by Proposition \ref{prop:better-std}
and hence $\delta(x, B) \le \delta(x, A) + O(\log n)$. 
%$B$ is $\epsilon + O(\log n)$-sufficient. 
If the function $\varkappa$ is chosen appropriately then $\K (B) > \K (A) - \delta$.
We can estimate $\K (\Omega_{\K(A)}  \cnd A)$
as follows
$$\K (\Omega_{\K(A)}  \cnd A) \le \K (\Omega_{\K(A)}  \cnd \Omega_{\K(B)}) + \K (\Omega_{\K(B)}  \cnd B) + \K(B \cnd A).$$
Let us show that every term in the right hand side of this inequality is $O(\delta + \log n)$. 
For the third term it holds by construction. The second term is equal to $O(\log n)$ since $B$ is a standard model. For the first term it holds since $\K (A) < \K (B) + \delta$.
\end{proof}

\smallskip
\textbf{Lemma \ref{lch}} (rigorous).
Assume that $A$ is a model for an $\eps, \alpha$-normal string $x$ of length $n$ with
$\eps \le n$, $\alpha < \sqrt{n}/2$. 
%Assume further that $A$ is an $\eps$-strong  statistic for $x$.  
Then there is a set $H$ such that:

1) $H$ is an $\eps$-strong statistic for $x$,

2) $\delta (x,H) \le \delta (x,A) + O((\alpha + \log n)\cdot \sqrt{n})$,

3) $\K (H \cnd \Omega_{\K (H)}) = O(\sqrt{n})$, 

4)  $\K (H) \le \K (A)+\alpha$.

\begin{proof}
%Since $x$ is normal, without loss of generality we may assume that $A$ is an $\eps$-strong model for $x$.
%The statement is obvious if $A$ is a minimal model for $x$. Indeed, in that case   $ \K(\Omega_{\K (A)} \cnd A)\approx 0$ by the previous lemma.
%As $\K(\Omega_{\K (A)})\approx \K(A)$, the symmetry of information implies that we also have $ \K(A\cnd \Omega_{\K (A)})\approx 0$. 
%Thus we can let $H=A$.

%In general case we
Consider the sequence $B_0,A_1, B_1, A_2, B_2, \dots $ of statistics for $x$
defined as follows.
Let $B_0=A$.
Then for all $i$ let $A_{i+1}$ be a strong statistic for $x$ 
obtained from $B_i$ by using the assumption that $x$ is $\eps,\alpha$-normal:
$$
\K(A_{i+1})\le \K(B_i)+\alpha, \quad \log\#A_{i+1}\le \log\#B_i+\alpha, \quad \KT(A_{i+1}\cnd x)\le\eps.
$$
 (See Fig. \ref{f5}.)
 
 Let for all $i$ let $B_i$ be a standard description that is the improvement of $A_i$
obtained by Proposition \ref{prop:better-std}:
$$
\delta(B_{i},x)\le \delta(A_i,x)+O(\log n), \quad \K(B_i\cnd A_i)=O(\log n).
$$
 
\begin{figure}[h]
\begin{center}
\includegraphics[scale=1]{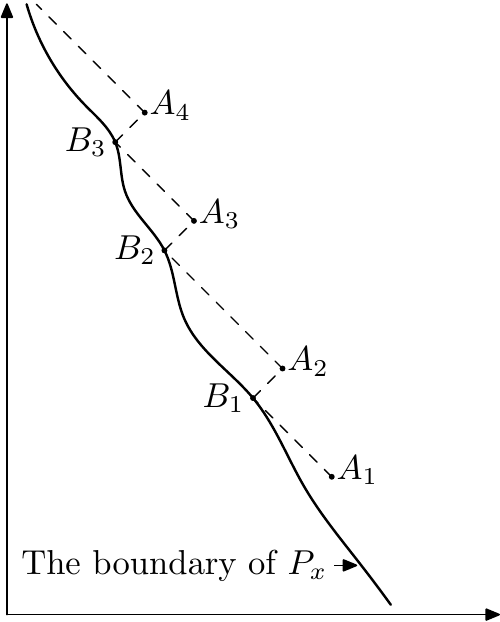}
\end{center}
\caption{Parameters of statistics $A_i$ and $B_i$}
\label{f5}
\end{figure}

Denote by $N$ the minimal integer such that $\K (A_N) - \K (B_N) \le \sqrt{n}$. 
Thus for all $i<N$ we have $\K(B_i)<\K(A_i)-\sqrt n$. On the other hand, the complexity  of 
$A_{i+1}$  is larger than that of $B_i$ by at most $\alpha<\sqrt n/2$.
Therefore for all $i<N$ we have 
$\K(A_{i+1})<\K(A_i)-\sqrt n/2$. Since $\K(A_1)=O(n)$ (recall that $\KT(A_1\cnd x)\le\eps\le n$) and $\K(A_N)\ge 0$,
we have $N = O(\sqrt{n})$. 

Let $H:=A_N$. By construction $H$ is strong (the first claim of the lemma) and $\K(H)\le \K(A_1)\le \K(A)+\alpha$ (the last claim). 
From $N = O(\sqrt{n})$ it follows that the second condition is satisfied. 

It remains to estimate $\K (H \cnd \Omega_{\K(H)})$. To this end we use the following inequality:
$$ 
\K (A_N \cnd \Omega_{ \K (A_N)}) 
\le \K (A_N \cnd B_N)+ 
  \K (B_N \cnd \Omega_{\K (B_N)}) + \K (\Omega_{ \K (B_N)} \cnd \Omega_{ \K (A_N)}).
 $$
We have to show that all terms in the right hand side 
are equal to $O(\sqrt{n})$.
This bound holds for the first term because $N$ was chosen so that 
$\K(A_N)-\K(B_N)\le\sqrt n$.
By construction $\K(B_N\cnd A_N)=O(\log n)$.
Therefore we can estimate $\K (A_N \cnd B_N)$ using the symmetry of information:
$$
\K (A_N \cnd B_N)=\K(A_N)+\K(B_N\cnd A_N)-\K(B_N)=\K(A_N)-\K(B_N)\le \sqrt n
$$
(with logarithmic accuracy).

The second term is $O(\log n)$ by construction (recall that $B_N$ is a standard description).
To estimate the third term note that by construction
we have $\K (B_N \cnd A_N) = O(\log n)$ and hence $\K(B_N) - \K(A_N) = O(\log n)$. 
\end{proof}

\begin{proof}[The proof of Theorem~\ref{thm:hereditary}]
To complete the proof of Theorem~\ref{thm:hereditary} 
we have to verify that all approximate equalities and inequalities used in the sketch of the proof 
hold with accuracy $O(\delta+(\varepsilon+\log n)\sqrt n)$.

In construction of  $A_1$ we use Proposition~\ref{part}. Hence the error terms in the inequalities relating $A$ and $A_1$
are of order $O(\eps+\log n)$. The complexity of the partition is of the same order.

Construction of $L$: the error terms in the inequality $\K(L)\le\K(\mathcal G)$ is of order $O(\log\log\#A_1)=O(\log  n)$ and 
the error term in the second inequality is constant.

%Step 3: we use $\eps,\eps$-normality of $x$ thus the error terms are  $\eps$ and $L_{\text{strong}}$ is $\eps$-strong.

In construction of $M$ we use Lemma~\ref{lch} for $\alpha=\eps$ (we can assume that the condition $\eps<\sqrt n/2$ of the lemma is met, as otherwise the statement of the theorem is obvious).
The error terms in  Lemma~\ref{lch} are of order $O((\eps+\log n)\sqrt n)$ and hence the parameters (complexity, complexity + log-cardinality)
of $M$ 
exceed those of $L$ by at most    $O((\eps+\log n)\sqrt n)$. 

Then we use Lemma~\ref{prop:min_hereditary} to estimate $\K(\Omega_{\K(A)}\cnd A)$, 
which is thus of order $O(\delta +\log n)$. Therefore the error term in the equality $\K(M\cnd A_1)\approx0$ is at most
$O(\delta+(\eps+\log n)\sqrt n)$.

In construction of $M_1$ we use Proposition~\ref{part}. Hence both complexities $\KT(M_1\cnd M)$ and  $\KT(M\cnd M_1)$ are $O(\eps+\log n)$.
The complexity of $A_1\cap M_1$ is at most $O(\delta+(\eps+\log n)\sqrt n)$ larger than that of $A_1$ and 
hence the equality  $\log \# (A_1\cap M_1)  \approx \log \# A_1$ holds with accuracy $O(\delta+(\eps+\log n)\sqrt n)$.
This implies that the equality $\KT(M_1\cnd A_1)\approx 0$ holds also with this accuracy.

Construction of $\mathcal F$:  we have seen that the equality $\KT(M_1\cnd A_1)\approx 0$ holds with accuracy $O(\delta+(\eps+\log n)\sqrt n)$.  
The equality $\CT(\mathcal F\cnd M_1)\approx 0$ holds with accuracy $O(\eps+\log n)$, since 
the complexity of $\mathcal A$ is $O(\eps+\log n)$.  The error terms in other inequalities are those from the previous steps and hence
of order $O(\delta+(\eps+\log n)\sqrt n)$.
\end{proof}

 \begin{remark}
 We have used the sufficiency of $A$ only to show that $\log \# (A_1\cap M_1)  \approx \log \# A_1$.
 This conclusion can be derived from a weaker assumption about $\K(A)$ and the profile of $x$. Namely, we 
 can assume that the profile of $x$ does not drop to the right of the point $(\K(A),\log\#A)$.
 More specifically, let  $c$ be the absolute constant such that $\K(M_1\cap A_1)\le \K(A)+c(\delta+(\eps+\log n)\sqrt n)$
 (see the proof). Then we can drop the assumption of $\eps$-sufficiency of $A$ at the expense
 of adding $\log\#A_1-\log\#(M_1\cap A_1)$ to the error term in the conclusion of the theorem.
 In this way we can prove the following version of the theorem:
 {\em 
 For some value $\varkappa=O(\log n)$ and for some constant $c$ the following holds. Assume that $A$ is an $\eps$-strong 
%and $\varepsilon$-sufficient 
statistic for an $(\varepsilon,\varepsilon)$-normal string $x$ of length $n$. Assume also that $A$ is a $(\delta,\varkappa)$-minimal model for $x$. 
Finally, assume that there is no model $A'\ni x$ with $\K(A')\le \K(A)+c(\delta+(\varepsilon+\log n)\sqrt n)$ and 
$\log\#A'\le\log\# A-\xi$.
Then $A$ is $(O(\xi+\delta+(\varepsilon+\log n)\sqrt n),O(\xi+\delta+(\varepsilon+\log n)\sqrt n))$-normal.}
If $A$ is an $\varepsilon$-sufficient model for $x$ then the last assumption holds $\xi =O(\delta+(\varepsilon+\log n)\sqrt n)$ and hence
this version implies the original one (but not the other way around). 
  \end{remark}

\subsection{The number of strings with a given profile}

In this section we consider the following questions. Let $P$ be a non-empty subset of $\mathbb{N}^2$. How many strings have the profile that is close to $P$? How many \emph{normal}
strings have the profile that is close to $P$? 

We assume that the set $P$ satisfies the necessary conditions from Theorem~\ref{stat-any-curve}, i.e. 
$P$ is an upward close set such that $(a, b +c) \in P$ implies $(a + b,c) \in P$ for all integers $a$, $b$ and $c$.
Thus by Theorem~\ref{stat-any-curve} there is at least one string whose profile is close to $P$ and by Theorem~\ref{stat-any-curve-1} there is at least one normal string whose profile is close to $P$.
To estimate the number of such strings better, we introduce the following three parameters:
\begin{align*}  
k_P &= \min\{ t \mid (t, 0) \in P\},\\
n_P &= \min\{ t \mid (0, t)\in P\},\\
m_P &= \min\{ t \mid (t, k_P -t) \in P\}.
 \end{align*}
The meaning of these numbers is as follows:
if $P$ is close to the profile of some string $x$ then 
$k_P$ is close to its complexity
and $m_P$ is close  to the complexity of a minimal sufficient statistic for $x$.
The parameter $n_P$ can be understood as ``the generalized length'' of $x$ (minimal log-cardinality of a 
set with negligible complexity containing $x$).
  
Note that $m_P\le k_P\le n_P$. Indeed, by definition the pair $(k_P,0)$ is in $P$ hence the set in which $m_P$ is the minimal element is not empty,
thus $m_P$ is well defined and is not larger than $k_P$. The second inequality is implied by 
the property  $(a, b +c) \in P\Rightarrow (a + b,c) \in P$
applied to $a=b=0$ and $c=n_P$.

\begin{thm}
\label{card}
1) There exist at least $2^{k_P - m_P - O(1)}$ strings whose profile is $O(\K(P) + \log n_P)$-close to $P$.

2) There exist at least $2^{k_P - m_P - O(1)}$ strings  that are 
$O(\log n_P)$, $O(\sqrt{n_P \log n_P})$-normal and whose profile is $O(\K(P) + \sqrt{n_P\log n_P})$-close to $P$.
\end{thm}
\begin{proof}
To prove the first statement of the theorem
consider the following auxiliary set $\tilde{P}$:
\begin{align*}
\tilde{P}%&=\{(i,j)\mid (i + k_P - m_P,j)\in P\}\\
&=\{(i,j)\mid\ i\le m_P,\ (i,j+k_P-m_P)\in P\}\cup
\{(i,j)\mid i\ge m_P\}.
\end{align*}
Note that  the term $k_P - m_P$ is non-negative, since $k_P\ge m_P$.
Our first goal is to construct a string $x$ such that 
$P_x$ is close to $\tilde P$.

To this end we apply Theorem \ref{stat-any-curve} 
to $n=n_P + m_P - k_P$, $k= m_P$ and the 
numbers $t_0,t_1,\dots,t_k$ such that
the line $(0,t_0)$--$(1,t_1)$--$\dots$--$(k,t_k)$ is 
the boundary curve of $\tilde P$. By that theorem 
there is a string $y$ of length $n + O(\log n)$ and complexity $m_P + O(\log n)$ whose 
profile is $O(\K(P) + \log n)$-close to $\tilde{P}$.

%Assume that such $x$ is already constructed.
Then we consider all strings $y$ of length $k_P - m_P$ such that 
$\K(y\cnd x)\ge k_P - m_P - 1$. By counting arguments there are at least 
$2^{k_P - m_P - 1}$ such strings. 
We claim that 
for each such $y$ the profile of the string $(x,y)$ is 
close to $P$.

Indeed, by Theorem~\ref{rem:add-noise}
the set $P_{(x,y)}$ can be obtained from the set $P_x$ by the following
transformation $\phi$:
$$
\phi(P_y)=
\{(i,j+k_P-m_P)\mid i\le \K(x),\ (i,j)\in P_x\}
\cup
\{(i,j)\mid i> \K(x),\ i+j\ge \K(x,y)\}.
$$
Note that if we replace here $\K(x)$ by $m_P$, $\K(x,y)$ by $k_P$, and
$P_x$ by $\tilde P$, we will obtain
exactly the original set
$P$.

It is easy to verify that the transformation
$P_y\mapsto \phi(P_y)$ is continuous: if sets
$P'$ and $P''$ are in an $\eps$-neighborhood of each other
then $\phi(P')$ and $\phi(P'')$ are also in an
$O(\eps)$-neighborhood of each other.
Besides, if, in the definition of $\phi$, we use some $\eps$-approximations
of $\K(x),\K(x,y)$ in place of $\K(x),\K(x,y)$,
then the resulting set is at most
$O(\eps)$-apart from the original set.
Since 
by construction
$\K(x)\approx m_P$ and $\K(x,y)\approx k_P$
(with accuracy $O(\K(P)+\log n_P)$), the sets $\phi(P_y)$ and
$P$ are at most  $O(\K(P)+\log n_P)$-apart.
Thus, by Theorem~\ref{rem:add-noise} 
the set $P_{(x,y)}$ is close to $P$.

The second statement is proved by a similar argument: we just use Theorem \ref{stat-any-curve-1} (about existence of normal strings with a given profile) instead of Theorem \ref{stat-any-curve}.
\end{proof}

Now we provide an upper bound for the number of strings whose profile is close to $P$.
In the proof we will use standard descriptions and Proposition \ref{prop:std-pos}.
By  Proposition \ref{prop:std-pos} for every standard description $A$ obtained from an enumeration
of strings of complexity at most $k$ the sum of complexity of $A$ and its log-cardinality is not greater than $k + c\log k$ for some constant $c$.
Fix such constant $c$ and let 
$$ 
m_P(\eps):= \min\{ t \mid  (t+\eps, k_P - t + c\log (k_P +2\eps) + \eps) \in P\}.
$$
The definition of $m_P(\eps)$ is similar to that of $m_P$.

Let $L(P, \eps)$ denote the set of all strings whose profile is $\eps$-close to $P$. We will assume that $\eps\le k_P$.
\begin{thm}
\label{uppest}
$\log \# L(P, \eps) \le k_P -m_P(\eps) + 2\eps + O(\log n_P)$.
\end{thm}
\begin{proof}
To simplify the notation let $m=m_P(\eps)$.
We will show that for every $x\in L(P,\eps)$ we have 
$\K(x\cnd \Omega_m)\le k_P-m+\eps$ (ignoring logarithmic terms). This bound obviously implies the theorem,
as there are few short programs and each program can map 
$\Omega_m$ only to one string.

First we show that $\K(\Omega_m\cnd x)$ is small for all $x\in L(P,\eps)$:
\begin{lemma}
\label{omp}
For every $x$ from $L(P, \eps)$ we have 
$\K(\Omega_m \cnd x) = O(\log n_P).$ 
\end{lemma}
\begin{proof}
Denote by $k$ be the complexity of $x$.  Since $x$ belongs to $L(P, \eps)$ 
the point $(k_P + \eps, \eps)$ belongs to $P_x$, hence $k=\K(x) \le k_P + 2\eps$.

Let $A$ be the standard description  obtained from the list of all strings of complexity at most $k$. By the choice of $c$ we have $\K(A) + \log 
\# A \le k + c \log k$. 
Hence the pair $(\K(A),k-\K(A)+c\log k )$ is in the profile of $x$.
This implies that the pair
$(\K(A)+\eps,k-\K(A)+c\log (k_P+2\eps)+\eps )$ is in $P$.

Recall that $m$ is defined as the minimal $\K(A)$ satisfying this property, thus we have  
$\K(A) \ge m$. 
By  Proposition \ref{prop:std-pos} we have $\K(\Omega_{\K(A)} \cnd A) = O(\log k)$ 
and hence $\K(\Omega_m \cnd A) = O(\log k)$. 
Also $\K(A \cnd x) = O(\log k)$, as $A$ is a sufficient statistic for $x$. So we have 
$$\K(\Omega_m \cnd x) \le \K(\Omega_m \cnd A) + \K(A \cnd x) + O(\log n_P) = O(\log n_P).$$
\end{proof}

It remains to upper bound  $\K (x \cnd \Omega_m)$ for every  $x\in L(P, \eps)$. 
We will do that by using the Symmetry of information:
 $$
 \K (x \cnd \Omega_m) = \K (x) + \K (\Omega_m \cnd x)-  \K (\Omega_m) + O(\log n_P).
 $$
Recall that $\K(\Omega_m) = m + O(\log m)$ and that $\K(x)\le k_P+2\eps$.
 So, by Lemma~\ref{omp} we get 
$$
\K(x \cnd \Omega_m) \le k_P +2\eps -m + O(\log n_P).
$$
\end{proof}

\subsection{Open questions}

1. Is it true that for every minimal strong sufficient statistic
$A$ and for every strong sufficient statistic $B$ for $x$
we have $\CT(B\cnd A)\approx 0$?
More specifically, is there a constant $c$ such that the following holds true:
Assume that $A,B$ are $\varepsilon$-strong,
$\varepsilon$-sufficient statistics for a string $x$ of length $n$. 
Assume further that there is no $(\eps+c\log n)$-strong $(\eps+c\log n)$-sufficient statistic $A'$ for $x$
with $\K(A')\le \K(A)-\delta$.
Is it true that $\KT(A\cnd B)=O(\eps+\delta+\log n)$ in this case?

2. The same question, but this
time we further assume that $B$ also satisfies the minimality
requirement:
there is no $(\eps+c\log n)$-strong $(\eps+c\log n)$-sufficient statistic $B'$ for $x$
with $\K(B')\le \K(B)-\delta$.

Note that if, in these two questions, we replace total conditional complexity with the plain conditional complexity then the answers are positive and moreover, we do not need to assume that $A,B$ are $\varepsilon$-strong (see 
Theorem~\ref{thm:step-wise}).

3. (Merging strong sufficient statistics.) Assume that $A,B$ are strong sufficient statistics for $x$ that have small intersection compared to the cardinality of at least one of them. Then it is natural to conjecture that there is a strong sufficient  statistic $D$ for $x$ of larger cardinality (=of smaller complexity) that is simple given both $A,B$. Formally, is it true (for some constant $c$) that if $A,B$ are  $\varepsilon$-strong $\varepsilon$-sufficient statistics for $x$, then there is a $c\varepsilon$-strong $c\varepsilon$-sufficient statistic $D$ for $x$ with $\log\#D\ge\log\# A+\log\# B-\log\#(A\cap B)-c(\varepsilon+\log n)$ and $\KT(D\cnd A),\KT(D\cnd B)$ at most $c(\varepsilon+\log n)$? (A motivating example: let $x$ be a random string of length $n$, let $A$ consist of all strings of length $n$ that have the same prefix of length $n/2$ as $x$, and let $B$ consist of all strings of length $n$ that have the same bits with numbers $n/4+1,\dots, 3n/4$ as $x$. In this case it is natural to let $D$ consist of all strings of length $n$ that have the same bits $n/4+1,\dots, n/2$ as $x$, so that $\log\#D=\log\# A+\log\# B-\log\#(A\cap B)$.)
}

\section{Acknowledgments}

We are grateful to several people who contributed and/or carefully read preliminary versions of this survey, in particular, to B. Bauwens, P. G\'acs, A. Milovanov, G. Novikov, A. Romashchenko, P. Vit\'anyi, and to all participants of Kolmogorov seminar in Moscow State University and ESCAPE group in LIRMM. We are also grateful to an anonymous referee for correcting several mistakes.


\begin{thebibliography}{99}

\bibitem{adleman1979}
L.M.~Adleman, \emph{Time, space and randomness}. MIT report MIT/LCS/TM-131. March 1979.

\begin{literature}
Considers the dependence of time bounded complexity on time and possible relation to computational complexity and P=NP question.
\end{literature}

\bibitem{ABST}
L.~Antunes, B.~Bauwens, A.~Souto, A.~Teixeira,
\emph{Sophistication vs. Logical Depth}, Theory of Computing Systems, First Online, \url{10.1007/s00224-016-9672-6}

\bibitem{af}
L. Antunes and L. Fortnow.
Sophistication revisited.
\emph{Theory of Computing Systems}, \textbf{45}(1), 150--161 (June 2009).

\begin{literature}
%BB
Theorem 3.1: there exists a string with large sophistication, i.e.,  an $n$-bit string for which each $c$-sufficient statistic has length at
least $n-2\log n - 2c$.

(They claim it was independently proven by Gacs, Tromp Vitanyi 2001.)

In Sect.~4 they define coarse sophistication \emph{csoph}. To understand motivation of Def.~4.1, note that $|p|+|d|-\KS(x) < c$ for any $c$-sufficient
statistic. So the idea is to use the $c$ as a penalty in selecting the model. [But then one can wonder why not have a penalty linear in $c$.]  They show that csoph can be large. [The trivial upper bound is $|x|/2$,
and it is tight within logarithms.]

In Sect.~5 they define something like coarse logical depth (Def.~5.2), which adds a similar penalty $|p| - \KS(x)$ to the inverse Busy Beaver of the running time of $p$. Then they show that this depth equals coarse
sophistication. The proof is similar to the proof in \cite{ABST}.
\end{literature}

\bibitem{AFM}
L. Antunes, L. Fortnow, and D. van Melkebeek. Computational depth,  \emph{Proceedings of the 16th IEEE Conference on Computational Complexity}, 266--273. IEEE, New York, 2001.  Journal version: Computational depth: Concept and applications, \emph{Theoretical Computer Science}, \textbf{354}(3), 391--404 (2006)

\bibitem{AMSV2009}
L.~Antunes, A.~Matos, A.~Souto, P.~Vit\'anyi, Depth as Randomness Deficiency, \emph{Theory of Computing Systems}, \textbf{45}(4), 724--739 (2009)

\begin{literature}
Section 3 repeats Bennett's lemma saying that two versions of logical depth are the same. Section 4 has lemma 4.2 and the authors claim it follows directly from the definitions. The second part of the paper is about independence for infinite strings.
%BB
\end{literature}

\bibitem{bauwens}
B.~Bauwens.
\newblock {\em Computability in statistical hypotheses testing, and characterizations of independence and directed influences in time series using Kolmogorov complexity}, PhD thesis, University of Gent, May 2010.

\bibitem{bennett}
C.H.~Bennett, Logical Depth and Physical Complexity,
in \emph{The Universal Turing Machine: a Half-Century Survey}. Edited by Rolf Herken, 227--257, Oxford University Press, 1988.

\begin{literature}

From the abstract (p.~227):
\begin{quote}
Some mathematical and natural objects (a random sequence, a sequence of zeros, a perfect crystal, a gas) are intuitively trivial, while others (e.g., the human body, the digits of $\pi$) contain internal evidence of a nontrivial causal history.

We formalize this distinction by defining an object's ``logical depth'' as the time required by a standard universal Turing machine to generate it from an input that is algorithmically random.

\end{quote}

Page 230:
\begin{quote}
We propose depth as a formal measure of value. From the earliest days of information theory it has been appreciated that information per se is not a good measure of message value. For example, a typical sequence of coin tosses has high information content but little value; an ephemeris, giving the positions of the moon and the planets every day for a hundred years, has no more information than the equations of motion and initial conditions from which it was calculated, but saves its owner the effort of recalculating these positions.  The value of a message thus appears to reside not in its information (its absolutely unpredictable parts), nor in its obvious redundancy (verbatim repetitions, unequal digit frequencies), but rather is what might be called its buried redundancy --- parts predictable only with difficulty, things the receiver could in principle have figured out without being told, but only at considerable cost in money, time, or computation. In other words, the value of a message is the amount of mathematical or other work plausibly done by its originator, which its receiver is saved from having to repeat.

\end{quote}

About the history (p.~227):
\begin{quote}
The notion of logical depth developed in the present paper was first described in [Chaitin 1977], and at greater length on [Bennett 1982] and [Bennett 1985]; similar notions have been independently introduces by [Adleman 1979] (``potential''), [Levin and V'jugin 1977] (``incomplete sequence''), [Levin 1984] (``hitting time''), and [Koppel, this volume] (``sophistication''). See also Wolfram's work on ``computational irreducibility'' [Wolfram 1985] and Hartmanis' work on time- and space-bounded algorithmic information [Hartmanis 1983].
\end{quote}

([Bennett 1982] is an unpublished manuscript, [Bennett 1985]: Information, dissipation and the definition of organization. In \emph{Emerging syntheses in Science}, ed. D.~Pines, 297--313, NM: Santa Fe Institute, 1985. [Adleman 1979]:~see~\cite{adleman1979} [Levin and V'jugin 1977]: Invariant properties of information bulks, Lecture notes in computer science, \textbf{53}, 359--364; only infinite sequences are considered there, and \emph{complete sequences} (=random with respect to some computable measure) are considered.  [Levin 1984]:  Randomness conservation inequalities: Information and independence in mathematical theories. \emph{Information and Control}, \textbf{61}, 15--37 (1984), based on draft report MIT/LCS/TR-235 (1980).  Here $Kt$ is defined (length of the program plus logarithm of the computation time); the rest of the paper is about infinite sequences. [Wolfram 1985]: Undecidability and intractability in theoretical physics, \emph{Phys. Rev. L.} \textbf{54}, 735--738 (1985). [Hartmanis 1983]: Generalized Kolmogorov complexity and the structure of feasible computations. In: \emph{Proceedings of the 25th IEEE Symposium on Foundations of Computer Science}, 1984. [Definition of time-bounded complexity and relations to computational complexity.]

Considers self-delimiting machines of standard type, prefix complexity, algorithmic probability (=a priori probability in standard terminology). Compressible strings by $s$ bits: $\KP(x)\le l(x)-s$ (compares prefix complexity and length). Time-bounded probability $P_t(x)$, if only program that terminate in time $t$ are considered. Busy beaver $B(n)$, the prefix version.

tentative definition of depth:

\begin{quote}
\textbf{Tentative Definition 0.1}: A string's depth might be defined as the execution time of its minimal program. $\langle\ldots\rangle$

\textbf{Tentative Definition 0.2}: A string's depth at significance level $s$ [might] be defined as the time required to compute the string by a program no more than $s$ bits larger than the minimal program.

This proposed definition solves the stability problem, but is unsatisfactory in the way it treats multiple programs of the same length. Intuitively, $2^k$ distinct $(n+k)$-bit programs that compute same output ought to be accorded the same weight as one $n$-bit program; $\langle\ldots\rangle$

\textbf{Tentative Definition 0.3}: A string's depth at significance level $s$ might be defined as the time $t$ required for the string's time-bounded algorithmic probability $P_t(x)$ to rise to within a factor $2^{-s}$ of its asymptotic time-unbounded value $P(x)$.

$\langle\ldots\rangle$ Although Definition 0.2 satisfactorily captures the informal notion of depth, we propose a slightly stronger definition for the technical reason that it appears to yield a stronger slow growth property (Theorem 1 below).

\textbf{Definition 1} (Depth of Finite Strings): Let $x$ and $w$ be strings and $s$ a significance parameter. A string's depth at significance level $s$, denoted $D_s(x)$, will be defined as $\min\{T(p)\colon (|p|-|p^*|<s) \land (U(p)=x)\}$, the least time required to compute it by a $s$-incompressible program $\langle\ldots\rangle$

The difference between this definition and the previous one is rather subtle philosophically and not very great quantitatively. Philosophically, Definition 1 says that each \emph{individual} hypothesis for the rapid origin of $x$ is implausible at the $2^{-s}$ confidence level, whereas the previous definition 0.3 requires only that a weighted average of all such hypotheses be implausible. The following lemma shows that the difference between Definition 1 and Definition 0.3 is also small quantitatively.

\textbf{Lemma 3.} There exists constants $c_1$ and $c_2$ such that for any string $x$, if programs running in time $\le t$ contribute a fraction between $2^{-s}$ and $2^{-s+1}$ of the string's total algorithmic probability, then $x$ has depth at most $t$ at significance level $s+c_1$ and depth at least $t$ at significance level $s-\min\{H(s),H(t)\}-c_2$.

\end{quote}

Here $H$ stands for prefix complexity. Proof sketch: if programs that run in some time (or have some other property) provide a significant part of total probability, then one of the programs with this property is incompressible (otherwise the probability would be greater). In other direction, if they provide only a small part of the total probability, they all can be compressed  (since they form a set of small measure).

Main theorems 1 and 1.1 say that the probability to increase depth significantly in a random process is small.

\end{literature}

\bibitem{bienvenu-desfontaines-shen}
L.~Bienvenu,
D.~Desfontaines,
A.~Shen,
What Percentage of Programs Halt?
\emph{Proceedings of ICALP 2015, 42nd International Colloquium, Kyoto, Japan, July 6-10, 2015}, Lecture notes in computer science, \textbf{9134}, 219--230. Extended version: Generic algorithms for halting problem and optimal machines revisited, \texttt{arXiv:1505.00731}.

\bibitem{bienvenu-gacs-et-al}
L.~Bienvenu,
P.~G\'acs,
M.~Hoyrup,
C.~Rojas,
A.~Shen,
Algorithmic tests and randomness with respect to a class of measures, \emph{Proceedings of the Steklov Institute of Mathematics}, \textbf{274}, 34--89 (2011).

Russian version: Алгоритмические тесты и случайность относительно классов мер, \emph{Труды математического института имени В.А.Стеклова}, \textbf{274}. 41--102 (2011).

\bibitem{cover1985}
T.~Cover, Kolmogorov complexity, data compression and inference. In: \emph{The Impact of Processing Techniques on Communications}, ed. J.K.~Skwirzynski. Martinus Nijhoff Publishers, 1985.

\begin{literature}
p.23: ``Special attention will be given to the so-called Kolmogorov $H$ function, a function that has not yet made its appearance in the literature. We argue that it plays the role of a minimal sufficient statistics''.

There is an interpretation of  learning some predicate as gambling against this predicate. Theorem 1 on p.~28 says that if we know in advance that $n$-bit sequence is in some set $F$, then we can bet (sequentially in the prescribed order) to win $2^{n-\log\#F}$. Strangely, Cover writes then ``The proof will not be given here but can be found in [1]'' where [1] is Kolmogorov 1965 paper (that contains nothing related to betting!)

Section 4, p.~31:
\begin{quote}
4. \textbf{Kolmogorov's $H_k$ Function}

Consider the function $H_k\colon \{0,1\}^k\to N$,  $H_k(x)=\min_{p\colon l(p)\le k} \log |S|$, where the minimum is taken over all subsets $S\subseteq\{0,1\}^n$, such that $x\in S$, $U(p)=S$, $l(p)\le k$. This definition was introduces by Kolmogorov in a talk at the Information Symposium, Tallinn, Estonia, in 1974. Thus $H_k(x)$ is the log of the size of the smallest set containing $x$ over all sets specifiable by a program of $k$ or fewer bits. Of special interest is the value
$$
k^*(x)=\min\{k\colon H_k(x)+k=K(x)\}.
$$
Note that $\log |S|$ is the maximal number of bits necessary to describe an arbitrary element $x\in S$. Thus a program for $x$ can be written in two stages: ``Use $p$ to print the indicator function for $S$; the desired sequence is the $i$th sequence in a lexicographic ordering of the elements of this set''. This program has length $l(p)+\log |S|$, and $k^*(x)$ is the length of the shortest program $p$ for which this $2$-stage description is as short as the best $1$-stage description $p^*$. We observe that $x$ must be maximally random with respect to $S$ --- otherwise the $2$-stage description could be improved, contradicting the minimality of $K(x)$. Thus $k^*(x)$ and its associated program $p$ constitute a minimal sufficient description for $x$. $\langle\ldots\rangle$

Arguments can be provided to establish that $k^*(x)$ and its associated set $S^*$ describe all of the ``structure'' of $x$. The remaining details about $x$ are conditionally maximally complex. Thus $pp^{**}$, the program for $S^*$, plays the role of a sufficient statistic.

\end{quote}

\end{literature}

\bibitem{epstein-levin}
S.~Epstein, L.~Levin, Sets have simple members, \url{http://arxiv.org/abs/1107.1458}, reposted as \url{http://arxiv.org/abs/1403.4539}.

\bibitem{gacs1984}
P.~G\'acs, On the relation between descriptional complexity and algorithmic probability, \emph{Theoretical Computer Science}, \textbf{22}, 71--93 (1983).

\bibitem{gtv}
P. G\'acs, J. Tromp, P.M.B. Vit\'anyi,
Algorithmic statistics, \emph{IEEE Transactions on Information Theory}, \textbf{47}(6), 2443--2463 (2001).

\bibitem{kolm65}
A.N.~Kolmogorov, Three Approaches to the Quantitative Definition of Information  [Russian: Три подхода к определению понятия <<количество информации>>] \emph{Problems of Information Transmission} [Проблемы передачи информации], \textbf{1}(1), 4--11 (1965). English translation published in: \emph{International Journal of Computer Mathematics}, \textbf{2}, 157--168 (1968).

\bibitem{kolm}
A.N. Kolmogorov, Talk
at the Information Theory Symposium in Tallinn, Estonia (then USSR), 1974. [As reported by Cover in his 1985 paper~\cite{cover1985}]

\bibitem{kolmmmo}
A.N.~Kolmogorov, The complexity of algorithms and the objective definition of randomness. Summary of the talk presented April 16, 1974 at Moscow Mathematical Society. \emph{Успехи математических наук} (Uspekhi matematicheskikh nauk, Russian), \textbf{29}(4[178]), 155 (1974). See~\url{http://mi.mathnet.ru/rus/umn/v29/i4/p153}. A short note in Russian.

\begin{literature}
Russian text:

Заседание 16 апреля 1974 г.

1. А.\,Н.\,Колмогоров ``Сложность алгоритмов и объективное
определение случайности''.

Любому конструктивному объекту $x$ можно поставить в соответствие функцию $\Phi_x(k)$ от натурального числа $k$ --- логарифм минимума мощности содержащего элемент $x$ множества, допускающего определение сложности не более $k$. Если сам элемент $x$ допускает простое определение, то функция $\Phi$ принимает значение единица уже при небольших $k$. Если такого простого определения нет, элемент в негативном смысле ``случаен''. Но он позитивно ``вероятностно случаен'' лишь в случае, если функция $\Phi$, получив при сравнительно небольшом значении $k=k_0$ значение $\Phi_0$, далее меняется приблизительно по закону $\Phi(k)=\Phi_0 -(k-k_0)$.
\end{literature}

\bibitem{kolmogorov81}
A.~Kolmogorov. Talk at the seminar at Moscow State University Mathematics Department (Logic Division), 26 November 1981. [The definition of $(\alpha,\beta)$-stochasticity was defined in this talk, and the question about the fraction of non-stochastic objects was posed.]

\bibitem{koppel87}
M.~Koppel, Complexity, Depth and Sophistication, \emph{Complex Systems}, \textbf{1}, 1087--1091 (1987).

\begin{literature}
\begin{quote}
p.~1087: The total complexity of an object is defined as the size of its most concise description . The total complexity of an object can be large while its ``meaningful'' complexity is low; for example, a random object is by definition maximally complex but completely lacking in structure.

$\langle\ldots\rangle$ The ``static'' approach to the formalization of meaningful complexity is ``sophistication'' defined and discussed by Koppel and Atlan [3]. [Reference to unpublished paper ``Program-length complexity, sophistication, and induction'' is given, but later a paper of the same authors appeared: Moshe Koppel, Henri Atlan, An almost machine-independent theory of program-length complexity, sophistication, and induction. \emph{Journal of Information Sciences: an International Journal}, \textbf{56}(1--3), 23--33 (Aug.~1991), \url{http://www.sciencedirect.com/science/article/pii/002002559190021L}. -- AS] Sophistication is a generalization of the ``H-function'' or ``minimal sufficient statistic'' by Cover and Kolmogorov [2] [Reference to Cover 1985 paper --- AS], using the mo[no]tonic complexity of Levin [4] [On the notion of random sequence, 1973 -- AS] The sophistication of an object in the size of that part of that object which describes its structure, i.e. the aggregate of its projectible properties.

$\langle\ldots\rangle$ The ``dynamic'' approach to the formalization of meaningful complexity
is ``depth'' defined and discussed by Bennett [1]. [Reference to an unpublished paper ``On the logical `depth' of sequences and their reducibilities to incompressible sequences''.] The depth of an object is the running-time of its most concise description. Since it is reasonable to assume that an object has been generated by its most concise description, the depth of an object can be thought of as a measure of its evolvedness.

Although sophistication is measured in integers and depth is measured in functions, it is not difficult to translate to a common range. It has already been shown by Schnorr and Fuchs [5] [C.P.~Schnorr, P.~Fuchs, General random sequences and learnable sequences, \emph{Journal of Symbolic Logic}, \textbf{42}, 329--340 (1977).] that the sophistication of an infinite string is infinite if and only if its depth is infinite. In this paper, we will prove that for all infinite strings sophistication and depth are essentially equal (that is, they differ by at most some constant ). Thus, the more sophisticated an object the more time needed for its evolution.

One way of demonstrating the naturalness of a concept is by proving the equivalence of a variety of prime facie different formalizations (e.g. computability) . It is hoped that the proof of the equivalence of two approaches to meaningful complexity, one using static resources (program size) and the other using dynamic resources (time), will demonstrate not only the naturalness of the concept but also the correctness of the specifications used in each formalization to ensure robustness and generality.
\end{quote}

The definition of sophistication is technical incomplete/incorrect: it uses Turing machine with separate inputs for programs and data, but never specifies which machine is used. (1988 paper~\cite{koppel} requires universality in the sense that it is not enough.) Still the idea is quite clear (p.~1089):

\begin{quote}
\textbf{Definition 3.} The $c$-sophistication of a finite string $S$ [is defined as] $$\text{SOPH}_c(S)=\min\{|P|\mid \exists D\text{ s. t. } (P,D)\text{ is a $c$-minimal description of $\alpha$}\}.$$
\end{quote}

(Probably the last symbol should be $S$, not $\alpha$; before in Definition 1 the description is called $c$-minimal if $|P|+|D|\le H(\alpha)+c$; here $P$ and $D$ and program and data inputs, $H$ stands for complexity.)

Then sophistication and depth are somehow defined for infinite sequences (the exact meaning is hard to understand due to technical flaws). On p.~1090: ``Thus, the depth of an infinite string $\alpha$ is the size of the smallest program which computes an upper-bound on the running time of minimal descriptions of segments of $\alpha$.'' Theorem 1 and 2 (p.~1090 and 1091) claim that depth coincides with sophistication up to a $O(1)$-term (where sophistication is also somehow defined for infinite sequences).

\end{literature}

\bibitem{koppel}
M.~Koppel, Structure, in \emph{The Universal Turing Machine: a Half-Century Survey}. Edited by Rolf Herken, 435--452, Oxford University press, 1988.

\begin{literature}
p.~435: ``What is the sophistication of the string, that is, what is the minimal amount of planning which must go into the generation of the string? More picturesquely, if the string is being broadcast by some unknown source, what is the minimum amount of intelligence we must attribute to this source?''

p.~436: ``both simple strings and random strings are not sophisticated''

p.~436: ``The minimal description of a string consists of two parts. One part is a description of the string's structure, and the other part specifies the string from among [``from among'' in the text --- AS] the class of strings sharing that structure (Cover 1985). The sophistication of a string is the size of that part of the description which description the string's structure. Thus, for example, the description of the structure of a random string is empty and thus, though its complexity is high, its sophistication is low''

Then some technical approach is suggested based on a specific Turing machines with separate tapes for data and program, but the definition is faulty (there is no optimality requirement, only universality).

References the same unpublished paper of Koppel and Atlan as in 1987 Koppel paper, and the same paper of Schnorr and Fuchs.

\end{literature}

\bibitem{koppel-atlan}
M.~Koppel, H.~Atlan, An almost machine-independent theory of program-length complexity, sophistication, and induction, \emph{Information Sciences}. \textbf{56}(1--3), 23--33 (1991).

\bibitem{levin-conservation}
L.~Levin, Randomness conservation inequalities; information and independence in mathematical theories, \emph{Information and Control}, \textbf{61}(1), 15--37 (1984).

\bibitem{LiVi97}
M. Li, P.M.B. Vit\'anyi,
\emph{An Introduction to Kolmogorov Complexity and its Applications}, 3rd ed., Springer, New York, 2008.

\bibitem{longpre-thesis}
L.~Longpr\'e, \emph{Resource bounded Kolmogorov complexity, a link between computational complexity and information theory}. Ph.~D. Thesis, Department of Computer Science, Cornell University, TR~86-776, 1986.

\bibitem{milovanov-antistochastic}
A.~Milovanov, Some properties of antistochastic strings. In: \emph{Computer Science -- Theory and Applications, 10th International Computer Science Symposium in Russia, Russia, July 13--17, 2015.} (CSR 2015), Lecture Notes in Computer Science, \textbf{9139}, 339--349. Journal version: \emph{Theory of Computing Systems}, Online First DOI \url{10.1007/s00224-016-9695-z}.

\bibitem{milovanov-stacs}
A.~Milovanov, Algorithmic statistic, prediction and machine learning, \emph{33rd Symposium on Theoretical Aspects of Computer Science} (STACS 2016), Leibnitz International Proceedings in Informatics (LIPIcs), \textbf{47}, 2016, DOI \url{10.4230/LIPIcs.STACS.2016.54}, \url{http://drops.dagstuhl.de/opus/volltexte/2016/5755/}, 54:1--54:13.

\bibitem{milovanov-csr}
A.~Milovanov, Algorithmic Statistics: Normal Objects and Universal Models, \emph{Computer Science -- Theory and Applications}, Proceedings of CSR~2016 conference, Lecture Notes in Computer Science, \textbf{9691}, 280--293.

\bibitem{MAAS}
F.~Mota, S.~Aaronson, L.~Antunes, A.~Souto, Sophistication as Randomness Deficiency, \emph{Descriptional Complexity of Formal Systems}, Lecture Notes in Computer Science, \textbf{8031}, 172--181 (2013)

\bibitem{mezhirov}
An.A.~Muchnik, I.~Mezhirov, A.~Shen, N.K.~Vereshchagin,
\emph{Game interpretation of Kolmogorov complexity}, \url{https://arxiv.org/abs/1003.4712}

\bibitem{muchnik-romashchenko}
An.A.~Muchnik, A.~Romashchenko,
Stability of properties of Kolmogorov complexity under relativization, \emph{Problems of Information Transmission}, \textbf{46}(1), 38--61 (2010).

\bibitem{muchnik-meta}
An.A.~Muchnik, A.L.~Semenov, V.A.~Uspensky, Mathematical metaphysics of randomness, \emph{Theoretical Computer Science},
\textbf{207}(2), 263--317 (November 1998).

\bibitem{shen-survey}
An.A.~Muchnik, A.~Shen, M.~Vyugin, \emph{Game arguments in computability theory and algorithmic information theory}, \url{https://arxiv.org/pdf/1204.0198.pdf}.

\bibitem{deRV}
S.~de~Rooij, P.M.B.~Vitanyi, Approximating Rate-Distortion Graphs of Individual Data: Experiments in Lossy Compression and Denoising, \emph{IEEE Transaction on Computers}, \textbf{61}, No.~3, March 2012, 395--407.

\bibitem{rissanen78}
J.~Rissanen, Modeling by shortest data description, \emph{Automatica}, \textbf{14}, 465--471 (1978).

\bibitem{schnorr}
C.P.~Schnorr.
\newblock Optimal enumerations and optimal {G}{\"o}del numberings.
\newblock {\em Mathematical Systems Theory}, 8(2):182--191, 1975.

\bibitem{shen83}
А.\,Шень.
Понятие $(\alpha,\beta)$-стохастичности по Колмогорову и его свойства.
\emph{Доклады Академии наук СССР}, \textbf{271}(6), 1337--1340 (1983).

English translation: A.~Shen, The concept of
$(\alpha,\beta)$-stochasticity in the Kolmogorov sense, and its
properties.
\emph{Soviet Math. Dokl.}, \textbf{28}(1),  295--299 (1983).

\bibitem{shen99}
A.~Shen, Discussion on Kolmogorov complexity and statistical
analysis, \emph{The Computer Journal}, 42:4(1999), 340--342.

\bibitem{shen-vovk}
A.~Shen,
Around Kolmogorov complexity: Basic Notions and Results,
\emph{Measures of Complexity. Festschrift for Alexey Chervonenkis}, Springer, 2015, 75--116. See also: \texttt{arXiv:1504.04955}.

\bibitem{sipser83}
M.~Sipser, A complexity theoretic approach to randomness, \emph{Proceedings of the fifteenth annual ACM symposium on Theory of computing} (STOC), 1983, 330--335.

\begin{literature}
The definition of time-bounded complexity $K^t(x\cnd A)$ is given. Also contains the result $\text{BPP}\subset\Sigma^2$ with Gacs' proof, too.
\end{literature}

\bibitem{sol1}
R.~Solomonoff, A formal theory of inductive inference. Part I, \emph{Information and Control}, \textbf{7}(1), 1--22 (1964)

\bibitem{sol2}
R.~Solomonoff,  A formal theory of inductive inference. Part~II. Applications of the Systems to Various Problems in Induction, \emph{Information and Control}, \textbf{7}(2), 224--254 (1964)

\bibitem{ver}
N.~Vereshchagin,
Algorithmic Minimal Sufficient Statistic Revisited.
In: \emph{Mathematical Theory and Computational Practice, 5th
Conference on Computability in Europe}, CiE 2009,
Heidelberg, Germany, July 19--24, 2009. Proceedings. LNCS 5635.

\bibitem{ver2015}
N.~Vereshchagin. Algorithmic Minimal Sufficient Statistics: a New Approach. \emph{Theory of Computing Systems}, \textbf{58}(3), 463--481 (2016).

\bibitem{vs-vovk}
N.~Vereshchagin, A.~Shen,
Algorithmic statistics revisited,
\emph{Measures of Complexity. Festschrift for Alexey Chervonenkis}, Springer, 2015, 2035--252. See also: \texttt{arXiv:1504.04950}.

\bibitem{usv}
Н.К. Верещагин, В.А. Успенский, А. Шень.
Колмогоровская сложность и алгоритмическая случайность. 576~с.
Москва, МЦНМО, 2013. Electronic version:
\url{http://www.lirmm.fr/~ashen/kolmbook.pdf}.

Draft English translation:
\url{http://www.lirmm.fr/~ashen/kolmbook-eng.pdf}.

\bibitem{vv}
N.K.~Vereshchagin, P.M.B.~Vit\'anyi, Kolmogorov's structure functions and model selection, \emph{IEEE Transactions on Information Theory}, \textbf{50}(12), 3265--3290 (2004).

\bibitem{vv10}
N.K.~Vereshchagin, P.M.B.~Vit\'anyi.
Rate Distortion and Denoising of Individual Data Using
Kolmogorov Complexity. \emph{ IEEE Transactions on Information Theory}, v.~56(7), 2010, p.~3438--3454.

\bibitem{Vi01}
P.M.B.~Vit\'anyi, \emph{Meaningful information},  \emph{IEEE Transactions on Information Theory}, \textbf{52}(10), 4617--4626 (2006). See also: \texttt{arXiv:cs/0111053}.

\bibitem{Vy87}
V.V. V'yugin,
On the defect of randomness of a finite object with respect to
measures with given complexity bounds, \emph{SIAM Theory Probab. Appl.}, \textbf{32}(3), 508--512 (1987).

\bibitem{Vy99}
V.V. V'yugin,
Algorithmic complexity and stochastic properties of finite binary sequences,
\emph{The Computer Journal}, \textbf{42}(4), 294--317 (1999).

\bibitem{Vy01}
V.V. V'yugin.
Does snooping help? \emph{Theoretical Computer Science}, \textbf{276}(1), 407--415 (2002).

\bibitem{wallace-boulton68}
 C.S.~Wallace, D.M.~Boulton, An information measure for classification, \emph{Computer Journal}, \textbf{11}(2), 185--194 (1968).

 \end{thebibliography}
\end{document}